\documentclass[12pt]{article}

\usepackage{colonequals}
\usepackage{graphicx}
\usepackage{amssymb}
\usepackage{epstopdf}
\usepackage{graphicx}
\usepackage{subfigure}
\usepackage{pstricks-add}
\usepackage{pst-circ}
\usepackage{pst-slpe}
\usepackage{color}
\usepackage{amsthm}
\usepackage{authblk}
\usepackage{placeins}
\usepackage[small]{caption}
\usepackage{algorithmic}
\usepackage{algorithm}
\usepackage{amsmath}
\usepackage{wasysym}
\usepackage{chemarr}
\usepackage{pst-node}
\usepackage{pst-coil}
\usepackage[bookmarks, colorlinks=true, linkcolor=blue, citecolor=green]{hyperref}

\title{Noise Propagation in Biological and Chemical Reaction Networks }
\author[1,2]{Dionysios Barmpoutis}
\author[2]{Richard M. Murray}
\affil[1]{Computation and Neural Systems}
\affil[2]{Control and Dynamical Systems \bigskip}
\affil[ ]{California Institute of Technology\bigskip}
\affil[ ]{dionysios$@$caltech.edu}
\affil[ ]{murray$@$cds.caltech.edu}
\date{\today}             

\newtheorem{lemma}{Lemma}

\newtheorem{corollary}{Corollary}

\theoremstyle{definition}
\newtheorem{definition}{Definition}

\providecommand{\norm}[1]{\lVert#1\rVert}

\providecommand{\OO}[1]{\operatorname{O}\left(#1\right)}

\begin{document}
\maketitle

\begin{abstract}
We describe how noise propagates through a network by calculating the variance of the outputs. 
Using stochastic calculus and dynamical systems theory, we study the network topologies that accentuate or alleviate  the effect of random variance in the network for both directed and undirected graphs.
Given a linear tree network, the variance in the output is a convex function of the poles of the individual nodes.
Cycles create correlations which in turn increase the variance in the output.
Feedforward and feedback have a limited effect on noise propagation when the respective cycles is sufficiently long.
Crosstalk between the elements of different pathways helps reduce the output noise, but makes the network slower. 
Next, we study the differences between disturbances in the inputs and disturbances in the network parameters, and how they propagate to the outputs. 
Finally, we show how noise correlations can affect the steady state of the system in chemical reaction networks with reactions of two or more reactants, each of which may be affected by independent or correlated noise sources.
\end{abstract}
\section{Introduction and Overview}
Noise is ubiquitous in nature, and virtually all signals carry some amount of random noise. 
In addition, even the simplest systems can be represented as a set of smaller subsystems interconnected with each other.
There have been numerous studies on how noise affects specific functions (e.g. \cite{Paulsson2004}, \cite{Raj2008} and references therein), but none of them has looked at how noise propagates in general networks, and how various network structures impact the robustness of each system to noise.
Although there is evidence that it may degrade the system performance, noise is sometimes necessary for specific functions \cite{FunctionalNoise}.
Networks in which information is transmitted through a means that is accessible by all the individual units of the network are prone to unwanted crosstalk interactions between various unrelated subsystems \cite{MinXtalkNetworks}.
Both noise and crosstalk have been treated as something unwanted in engineering systems. 
However, they do not seem to be a problem in the cell, or in natural  biological systems in general, despite the large number of noise sources, the variety of molecules, and the intricate patterns of interactions.

We present a new method to quantify the noise propagation in a system, and the vulnerability of each of its subsystems.
We use results from graph theory and control systems theory to quantify noise propagation in networks, and use them to evaluate various network structures in terms of how well they filter out noise.
We study how crosstalk can help suppress noise, when the noise sources are independent or correlated.
We show that perturbations that depend on the state of the system (for example, feedback loops that are prone to noise or noisy degradation rates) have a fundamentally different effect on the system output, compared to noise in the inputs.
Finally, we study noise propagation in chemical reaction networks where all reactants may introduce noise, and analytically find that noise correlations may affect the expected behavior of such systems.

\section{Background}
\subsection{General Response of Linear Systems}
In this section, we will briefly revisit some basic tools from control systems theory.
Consider a linear time invariant system with impulse response $h(t,s)$ \cite{ControlBook}.
The general form of the output when the input signal is $u(t)$ is
\begin{equation}
y(t)=\int _{-\infty}^{t} h (t,s)u(s)ds
\label{GeneralImpulseResponse}
\end{equation}
where $h(t)$ is the impulse response of the dynamical system.
A system with $m$ inputs, $n$ states and $p$ outputs can be written in the form
\begin{equation}
\mathcal{S}: \left\{ \begin{array}{ll}
 &\frac{dx}{dt}=Ax+Bu \\
& y=Cx,
\end{array} \right.
\label{ClassicLinearSystem}
\end{equation}
where the dimensions of matrices $A, B$ and $C$ are $n\times n$, $n\times m$ and $p\times n$ respectively.
The output of the system at time $t$ when the input is an impulse applied at time $s$ is
\begin{equation}
h(t,s)=C e^{A(t-s)}B
\end{equation}
and equation \eqref{GeneralImpulseResponse} can be simplified to 
\begin{equation}
y(t)=C \int _{-\infty}^{t} e^{A(t-s)}B u(s)ds.
\label{LinearImpulseResponse}
\end{equation}

When the network in question is comprised of elements whose outputs obey linear time-invariant differential equations, we can also find the Fourier transform of the network output:
\begin{equation}
H(\omega)=\int _{-\infty}^{+\infty} h(t)e^{-j \omega t} dt,
\end{equation}
where $h(t)=h(t,0)$ is the impulse response of the system and $\omega =2\pi f$ is the angular frequency.
If the system is causal ($h(t)=0$ for $t<0$), then the expression above can be simplified by replacing the lower limit of the integral with zero.

When the input is a stochastic process, its output will be a stochastic process as well.
We are interested in the mean, the variance, and occasionally the higher central moments of the system output  once the system has reached its equilibrium state.
The mean $\mathbb{E}[y(t)]$ and the variance $\mathbb{V}[y(t)]$ of the output $y(t)$ will be denoted as $\mathbb{E}[y]$ and $\mathbb{V}[y]$ respectively:
\begin{equation}
\mathbb{E}[y]=\lim _{t\rightarrow \infty} \mathbb{E}[y(t)] \qquad \textrm{and}\qquad  \mathbb{V}[y]=\lim _{t\rightarrow \infty} \mathbb{V}[y(t)].
\label{AsymptoticExpectationAndVariance}
\end{equation}
If we know the impulse response of the system, the mean of the output vector can be expressed as 
\begin{equation}
\begin{aligned}
\mathbb{E}[y(t)] &=\mathbb{E} \left [\int _{-\infty}^{t} h (t-s)u(s)ds \right ] \\
&=\int _{-\infty}^{t} h (t-s)  \cdot \mathbb{E} \left [ u(s)\right] ds, \\
\label{GeneralOutputExpectedValue}
\end{aligned}
\end{equation}
where in the last equation we have interchanged the expectation with the integration operator, assuming that the input functions are non-pathological, and the quantities are finite, such that all the integrands are measurable in the respective measure space (Fubini's theorem, \cite{ProbabilityBook}).
In what follows, we will always assume that all such conditions are satisfied.

The covariance matrix of the outputs,  when applying the same input is
\begin{equation}
\begin{aligned}
\mathbb{V}[y(t)]&=\mathbb{E}[y(t) \cdot y^{T}(t)]\\
&=\int _{-\infty}^{t} \int _{-\infty}^{t} h (t-r) \left( \mathbb{E}\left[ u(r)u^{T}(s)\right] -\mathbb{E}\left[ u(r) \right] \mathbb{E}\left[u^{T}(s)\right] \right)  h^{T} (t-s)  dr ds.\\
\end{aligned}
\label{GeneralOutputVariance}
\end{equation}
If in addition $u(t)=0$ for $t<0$, then according to equation \eqref{AsymptoticExpectationAndVariance},
\begin{equation}
\begin{aligned}
\mathbb{V}[y]=\lim _{t\rightarrow \infty}\int _{0}^{t} \int _{0}^{t} h (t-r) \left( \mathbb{E}\left[ u(r)u^{T}(s)\right] -\mathbb{E}\left[ u(r) \right] \mathbb{E}\left[u^{T}(s)\right] \right)h^{T} (t-s)  dr ds.\\
\end{aligned}
\end{equation}

\subsection{Wiener Process}

In this subsection, we will be describing some elementary properties of the Wiener process that will be used in the following analysis.
Let $\xi_{n}$, $n\in \mathbb{N}$ be a sequence of independent identically distributed random variables with zero mean and unit standard deviation.
Their sum is
\begin{equation}
S_{n}=\sum_{k=1}^{n} \xi_{n}.
\end{equation}
We now define the piecewise constant function
\begin{equation}
W_{t}=\lim_{n\rightarrow \infty}\frac{S_{\lfloor n t\rfloor}}{\sqrt{n}}.
\end{equation}
According to the Central Limit Theorem, the distribution of $W_{t}$ is {\it independent} of the distribution of the sequence of $\xi _{n}$, as long as they have finite variance, are identically distributed and independent of each other.
The random process $W_{t}$ is normally distributed with variance equal to the time interval it which it is measured:
\begin{equation}
W_{t}=\lim_{n\rightarrow \infty}\frac{S_{\lfloor n t\rfloor}}{\sqrt{nt}} \frac{\sqrt{nt}}{\sqrt{n}} \implies W_{t}\sim \mathcal{N}(0,t).
\label{CLTApplication}
\end{equation}
The difference of two sums $S_{b}-S_{a}$ with $a<b$ has the same distribution of the random variable $S_{b-a}$ and as a result
\begin{equation}
W_{b}-W_{a}\sim W_{b-a} \qquad 0\leq a <b.
\end{equation}
Lastly, the random variables $W_{b}-W_{a}$ and $W_{d}-W_{c}$ are independent  when $0\leq a<b\leq c<d$, since the respective sums consist of independent random variables.
More details on the properties of the Wiener process can be found in \cite{ProbabilityBook}.

\subsection{Graph Theory}
A \textit{graph} (also called a \textit{network}) is an ordered pair $\mathcal{G}=(\mathcal{V},\mathcal{E})$ comprised of a set $\mathcal{V}=\mathcal{V}(\mathcal{G})$ of \textit{vertices} together with a set $\mathcal{E}=\mathcal{E}(\mathcal{G})$ of \textit{edges} that are unordered 2-element subsets of $\mathcal{V}$.
Two vertices $u$ and $v$ are called \textit{neighbors} if they are connected through an edge ($(u,v)\in \mathcal{E}$) and we write $u-v$, otherwise we write $u \notslash v$.
The \textit{neighborhood} $\mathcal{N}_{u}$ of a vertex $u$ is the set of its neighbors.
The \textit{degree}  of a vertex is the number of its neighbors.
The \textit{order} $N$ of a graph is the number of its vertices, $N=|\mathcal{V}|$.
A graph's \textit{size} (denoted by $m=|\mathcal{E}|$), is the number of its edges.
We will denote a graph $\mathcal{G}$ of order $N$ and size $m$ as $\mathcal{G}(N,m)$ or simply $\mathcal{G}_{N,m}$.
A \textit{path} is a sequence of consecutive edges in a graph and the length of the path is the number of edges traversed.
The \textit{distance} between two vertices $u$ and $v$, usually denoted by $d=d(u,v)$, is the length of the shortest path that connects these two vertices.
A \textit{full cycle} is a cycle that includes all the vertices of the network.
A graph is \textit{connected} if for every pair of vertices $u$ and $v$, there is a path from $u$ to $v$.
Otherwise the graph is called \textit{disconnected}.
We will be focusing exclusively on connected graphs, because every disconnected graph can be analyzed as the sum of its connected \textit{components}.
A \textit{tree} is a graph in which any two vertices are connected by exactly one path.
A \textit{path graph} is a tree with two or more vertices that has two vertices with degree 1, while all other vertices have degree 2.
A thorough treatment of the graph theory notions used in this article can be found in \cite{GraphTheoryBook}.

\section{White Noise Input}

In the state space, when the parameters of the system are deterministic and the input consists of a deterministic and a random component (white noise), then the system \eqref{ClassicLinearSystem} is defined by the stochastic differential equation:
\begin{equation}
\mathcal{S}: \left\{ \begin{array}{ll}
 &dx=Ax \cdot dt+B (u_{t} dt+\Sigma _{t}dW_{t})  \\
& y=Cx, \\
\end{array} \right.
\label{StochasticLinearSystem}
\end{equation}
where $dW_{t}=W_{t+dt}-W_{t}$ is the standard vector Wiener process in the time interval $[t,t+dt)$ and $u_{t}$ is a deterministic input.
We will denote the value of a function $f$ at time $t$ as $f(t)$ or $f_{t}$ interchangeably.
The matrix $\Sigma_{t}$ consists of nonnegative entries, possibly time-varying, each of which is proportional to the strength of the corresponding disturbance input.
Note that the only difference with the system \eqref{ClassicLinearSystem} is that now the infinitesimal state difference $dx$ depends not only on the current state and the deterministic input, but also a random term $dW_{t}\sim \mathcal{N}(0,dt)$.

It should be noted that the fraction $dW_{t}/dt$ does not exist as $dt\rightarrow 0$, so dividing both sides of equation \eqref{StochasticLinearSystem} by $dt$ would not make sense.
But this notation also helps us to intuitively understand the effect of randomness in the system, when we know how the state of the system is affected by the randomness in the inputs.
It also helps us to easily generalize these results when the randomness is a product of many noise sources as we will see in the last section.

The different Wiener processes may be correlated with each other but since each input may consist of a weighted sum of all of the different processes through multiplication by matrix $\Sigma _{t}$, the analysis is simplified if we assume that they are independent.

The output of the system is the superposition of the deterministic output, and the response to the random input:
\begin{equation}
\begin{aligned}
y(t) &=\int _{-\infty}^{t} h(t-s)(u(s)ds+ \Sigma_{s} dW_{s}) \\
&=\int _{-\infty}^{t} h(t-s)u(s)ds+\int _{-\infty}^{t} h(t-s)\Sigma_{s} dW_{s}. \\
\end{aligned}
\end{equation}
The expected value for the output, according to equation \eqref{GeneralOutputExpectedValue}  will be
\begin{equation}
\begin{aligned}
\mathbb{E}[y(t)] &=\int _{-\infty}^{t} h (t-s)\mathbb{E} \left [u(s)ds+ \Sigma_{s} dW_{s} \right ] \\
&=\int _{-\infty}^{t} h (t-s) u(s)ds, \\
\end{aligned}
\end{equation}
since Brownian motion is a martingale \cite{ProbabilityBook}.

Applying equation \eqref{GeneralOutputVariance} when the input is white noise,  the covariance matrix can be written as
\begin{equation}
\begin{aligned}
\mathbb{V}[y] &=\lim _{t\rightarrow \infty} \mathbb{V}[y(t)] \\
&=\lim _{t\rightarrow \infty}\int _{-\infty}^{t} \int _{-\infty}^{t} h (t-r) \mathbb{E}\left[dW_{r} \Sigma_{r} \Sigma_{s}^{T} dW^{T}_{s} \right]  h^{T} (t-s). \\
\end{aligned}
\label{LinearSystemVariance}
\end{equation}
But since the inputs are assumed to be white noise processes, the covariance among all of them is nonzero only if they take place during the same interval, and in that case, the covariance is proportional to the length of this interval.
\begin{equation}
\begin{aligned}
\mathbb{V}[y] &=\int _{-\infty}^{t} \int _{-\infty}^{t} h (t-r) \left( \Sigma_{r} \sqrt{dr}  \delta(r-s) \sqrt{ds} \Sigma_{s}^{T} \right) h^{T} (t-s) \\
 &=\int _{-\infty}^{t}h (t-s) \cdot V(s) \cdot h^{T} (t-s) ds, \\
\end{aligned}
\label{OutputVarianceFromImpulseResponse}
\end{equation}
where $V(s)=\Sigma_{s}\Sigma_{s}^{T}$ is the covariance matrix of the input random vector.
For the linear time invariant system \eqref{ClassicLinearSystem} and white noise inputs of constant variance $V(s)$ is a constant matrix, and we can write
\begin{equation}
\begin{aligned}
\mathbb{V}[y] &=\int _{-\infty}^{t} (C e^{A (t-s)} B) \cdot V \cdot (C e^{A(t-s)} B)^{T} ds \\
&= C \left( \int _{0}^{+\infty} e^{A x} B V B^{T} e^{A^{T}x} dx \right) C^{T}.
\end{aligned}
\label{TimeDomainLTINoiseResponse}
\end{equation}

The mean and the variance of the output signal in the steady state can be written as a function of the Fourier transforms of the input signal and the network transfer function.
From equation \eqref{GeneralOutputExpectedValue}
\begin{equation}
\begin{aligned}
\mathbb{E} \left[y(t) \right]&=\mathbb{E} \left[ \int _{-\infty}^{t}  h(t-s) u(s)\right]ds \\
&=h(t)*\mathbb{E} \left[ u(t) \right]
\end{aligned}
\end{equation}
where $f(t)*g(t)$ denotes the convolution of two functions $f(t)$ and $g(t)$ given that it exists.

When the input is constant with time, the expected value of the input is constant as well ($\mathbb{E}[u(t)]=\mu_{x}$) and the last expression can be simplified to
\begin{equation}
\mathbb{E}[y]=\mu _{x} \int_{0}^{+\infty}h(u)du = \mu_{x} H(0).
\end{equation}

If the input itself is not known, but its frequency content can be estimated, we can find the variance of the output using Parseval's theorem:
\begin{equation}
\begin{aligned}
\mathbb{V}[y]&=\mathbb{E}[y\cdot y^{T}]=\lim _{t\rightarrow \infty} \int_{-\infty}^{t} y(t) y^{T}(t)dt \\
&=\int _{-\infty}^{+\infty} |Y(f)|^{2}df =\int _{-\infty}^{+\infty} Y(f)\cdot Y^{*}(f)df \\
&=\int _{-\infty}^{+\infty} H(f) X(f) X^{*}(f)  H^{*}(f) df.
\end{aligned}
\label{FrequencyDomainVarianceInputFrequency}
\end{equation}

The formula above is useful if we know or we can estimate the various frequencies of the input random processes.
More generally, if we know the autocorrelation function of the random processes in the input, we may find the expected autocorrelation in the output, and then estimate the output variance.

\begin{equation}
\begin{aligned}
R_{y}(\tau) &=\int _{-\infty}^{+\infty}S_{y}(f) \cos(2\pi f\tau) df \\
&=\int _{-\infty}^{+\infty} |H(f)|^{2} S_{x}(f) \cos(2\pi f\tau) df \\
&=\int _{-\infty}^{+\infty} |H(f)|^{2} \left(  \int _{-\infty}^{+\infty}R_{x}(u) \cos(2\pi f u) du \right)  \cos(2\pi f\tau) df .
\end{aligned}
\end{equation}

We will be focusing on Wiener processes exclusively, because this is the most general approach for sums of random disturbances.
The Central Limit Theorem shows that the sum of a large number of independent identically distributed random variables with finite mean and variance always approaches the normal distribution (see also equation \eqref{CLTApplication}).
The only assumption in the case of additive disturbances is that the inputs at every time are sums of independent random variables of arbitrary distribution of finite standard deviation.
This is a reasonable assumption in most settings.
For example, in biology the Poisson distribution is frequently used to model random disturbances \cite{Paulsson2004}.
The Poisson distribution can be well approximated by a Gaussian when the event rate is greater than $10$ (see \cite{BiostatisticsBook}), and the same can be said for small sums of Poisson random variables.
When the input disturbance at each time is correlated with the disturbances during earlier times, the correlation structure can be emulated by passing white noise through a filter that produces it.
Also, in some applications, noise cannot be expected to have equal frequency content for all frequencies up to infinity.
We can still use white noise as an input, which we can pass through a filter with zero response for all the frequencies outside the desired range.

\section{Tree Networks}

Tree networks are a special case of networks where there is a unique path among every pair of vertices.
In other words, there are no cycles, which makes the analysis of such networks easier.
Many natural networks have been found to be locally tree-like \cite{ScaleFreeMetabolicNetworks}.
When analyzing the behavior of a network around an equilibrium point, or if the network is linear, then the analysis can be significantly simplified.
Since there is a unique path from any vertex to another, it suffices to analyze path networks, which consist of all their vertices connected in series.
For each output, the total response of the system is the superposition of the signals caused for all the individual inputs.
First, we will show that in the case of random signals, the order of the nodes in the network does not matter in the case of linear pathways.
Then, we will find the variance of a linear path graph assuming that every node is a first order filter. 
The result can easily be generalized for the case of arbitrary tree graphs.
Finally, we are going to find the optimal placement of poles so that the noise suppression is maximized.

\subsection{Output Variance of Linear Pathways}

\begin{lemma}
The noise response of a linear pathway is independent of the relative position of its nodes.
\end{lemma}
\begin{proof}
Without loss of generality, we can assume that the linear pathway has one input and one output.
Otherwise, since the system is linear, we can repeat the process each time considering only the respective subtree.
Under the last assumption, the output is the state of the last node, and all inputs affect only the first node.
From equation \eqref{FrequencyDomainVarianceInputFrequency}: 
\begin{equation}
\begin{aligned}
\mathbb{V}[y]&=\int _{-\infty}^{+\infty} H(f) X(f) X^{*}(f)  H^{*}(f) df \\
&=\int _{-\infty}^{+\infty} H(f) (X_{1}(f) +\ldots +X_{n}(f))(X^{*}_{1}(f) +\ldots +X^{*}_{n}(f))  H^{*}(f) df \\
&=\sum _{k=1}^{n} \sum _{m=1}^{n} \int _{-\infty}^{+\infty} X_{k}(f)X^{*}_{m}(f)  H(f) H^{*}(f) df \\
&=\sum _{k=1}^{n} \sum _{m=1}^{n} \int _{-\infty}^{+\infty} X_{k}(f)X^{*}_{m}(f)  (h_{1}(f)\cdot \ldots \cdot h_{N}(f)) (h^{*}_{n}(f) \cdot \ldots \cdot h^{*}_{1}(f))   df \\
&=\sum _{k=1}^{M} \sum _{m=1}^{M} \int _{-\infty}^{+\infty} X_{k}(f)X^{*}_{m}(f)  \prod _{n=1}^{N}|h_{n}(f)|^{2} df.
\end{aligned}
\end{equation}
It is evident that we can interchange the transfer functions inside the product in the integral, without changing its value.
\end{proof}

Assume that we have a linear pathway such that the system is linear, described by the equation \eqref{ClassicLinearSystem}
where the dynamical and input matrices are 
\begin{equation*}
\small
A=\begin{bmatrix} -d_{1} & 0 &\dots &0 &0 \\ f_{2} & -d_{2} &\dots &0 &0 \\ 0 & f_{3} &\dots &0 &0 \\ 0 & 0 &\dots &f_{N} & -d_{N} \end{bmatrix}
\qquad
B=\begin{bmatrix} 1\\0\\ \vdots \\0 \end{bmatrix}
\qquad
C^{T}=\begin{bmatrix} 0\\ \vdots \\0\\1 \end{bmatrix}
\normalsize
\end{equation*}
For simplicity, we assume that there is only one noise source and only one output, but since there are no cycles, there is a unique path from each node to every other,  which means we can use the result for a linear pathway repeatedly, in order to find the total variance.
The variance is independent of the deterministic input that is applied to the pathway, since the system is linear.

Using equation \eqref{TimeDomainLTINoiseResponse}, and after performing all calculations, the variance at the output will be 
\begin{equation}
\mathbb{V}_{out}=\left( \prod _{u=1}^{N-1}f_{u} \right) \left( \sum _{k=1}^{N}\sum _{m=1}^{N} \frac{1}{(d_{k}+d_{m})\displaystyle\prod_{a=1,a\neq k}^{N} (d_{k}-d_{a})\displaystyle\prod_{b=1,b\neq m}^{N}(d_{m}-d_{b})}\right).
\end{equation}

The expression above holds even if there exist two vertices $a$ and $b$ such that their reaction rates are equal, according to the next Lemma.
\begin{lemma}
The output variance of a linear pathway does {\it not} depend on the difference of any of the reaction rates.
\end{lemma}
\begin{proof}
We pick two rates $d_{x}$ and $d_{y}$ and show that the $\mathbb{V}_{out}$ does not depend on their difference.
If we denote 
\begin{equation}
T_{k,m}=\frac{1}{(d_{k}+d_{m})\displaystyle\prod_{a=1,a\neq k}^{N} (d_{k}-d_{a})\displaystyle\prod_{b=1,b\neq m}^{N}(d_{m}-d_{b})},
\end{equation}
the difference $d_{x}-d_{y}$ appears only in the terms $T_{x,x}, T_{x,y},T_{y,x}$ and $T_{y,y}$.
Their sum $T_{x-y}$ is equal to
\begin{equation}
\begin{aligned}
T_{x-y} &=T_{x,x}+T_{x,y}+T_{y,x}+T_{y,y} \\
&=\frac{1}{2d_{x}(d_{x}-d_{y})^{2} \displaystyle \prod_{s\neq x,y}(d_{x}-d_{s})}+\frac{1}{2d_{y}(d_{y}-d_{x})^{2} \displaystyle \prod_{s\neq x,y}(d_{y}-d_{s})} \\
&\qquad - \frac{2}{(d_{x}+d_{y})(d_{y}-d_{x})^{2} \displaystyle \prod_{s\neq x,y}(d_{y}-d_{s})}.
\end{aligned}
\end{equation}
We set 
\begin{equation}
P_{x}=\displaystyle \prod_{s\neq x,y}(d_{x}-d_{s}) \quad \textrm{and} \quad P_{y}=\displaystyle \prod_{s\neq x,y}(d_{y}-d_{s}) 
\end{equation}
so that sum above can be written as
\begin{equation}
\begin{aligned}
T_{x-y} &=\frac{d_{x}(d_{x}+d_{y})P_{x}^{2}+d_{y}(d_{x}+d_{y})P_{y}^{2}-4d_{x}d_{y}P_{x}P_{y}}{2d_{x}d_{y}(d_{x}+d_{y}) (d_{y}-d_{x})^{2} P_{x}^{2}P_{y}^{2}}.
\end{aligned}
\end{equation}
Expanding the nominator of $T_{x-y}$ and grouping the relevant terms together:
\begin{equation}
\begin{aligned}
T_{x-y} &=\frac{d_{x}^{2}P_{x}^{2}+d_{x}d_{y}P_{x}^{2}+d_{x}d_{y}P_{y}^{2}+d_{y}^{2}P_{y}^{2}-4d_{x}d_{y}P_{x}P_{y}}{2d_{x}d_{y}(d_{x}+d_{y}) (d_{y}-d_{x})^{2} P_{x}^{2}P_{y}^{2}} \\
&=\frac{ \left( d_{x}^{2}P_{x}^{2}-2d_{x}d_{y}P_{x}P_{y} +d_{y}^{2}P_{y}^{2} \right) +d_{x}d_{y}\left(P_{x}^{2}-2P_{x}P_{y} +P_{y}^{2}\right)   }{2d_{x}d_{y}(d_{x}+d_{y}) (d_{y}-d_{x})^{2} P_{x}^{2}P_{y}^{2}} \\
&=\frac{ \left( d_{x}P_{x}-d_{y}P_{y}\right)^{2} +d_{x}d_{y}\left(P_{x}-P_{y}\right)^{2}   }{2d_{x}d_{y}(d_{x}+d_{y}) (d_{y}-d_{x})^{2} P_{x}^{2}P_{y}^{2}}.
\end{aligned}
\end{equation}
It is easy to see that both terms in the nominator of the last fraction have a factor of order $(d_{y}-d_{x})^{2}$, and the Lemma is proved, and the fraction does not depend on the square difference $(d_{y}-d_{x})^{2}$.
\end{proof}

\subsection{Optimization of Linear Pathways}

\begin{lemma}
Assume that the same noise source is applied to two different pathways with impulse responses $h_{1}(t)$ and $h_{2}(t)$ respectively.
The covariance of the signals in their output will be equal to
\begin{equation}
\begin{aligned}
C(\tau)&=\lim _{t\rightarrow \infty} \mathbb{E} \left[y_{1}(t)y_{2}(t+\tau) \right] \\
&=\int _{0}^{\infty} h_{1}(r)h_{2}(r+\tau) dr.
\end{aligned}
\end{equation}
\end{lemma}
\begin{proof}
The two outputs $y_{1}(t)$  and $y_{2}(t)$ are equal to 
\begin{equation}
y_{1}(t)=\int _{-\infty}^{t} h_{1}(t-x) dW_{x} \quad \textrm{and} \quad y_{2}(t)=\int _{-\infty}^{t} h_{2}(t-y) dW_{y} 
\end{equation}
where $W_{t}$ is the Wiener process that drives both systems simultaneously.
Taking the expected value of the product of the first and a delayed version of the second, 
\begin{equation}
\begin{aligned}
C(\tau) &=\lim_{t\rightarrow \infty} \mathbb{E} \left[ \int _{-\infty}^{t} h_{1}(t-x) dW_{x} \cdot \int _{-\infty}^{t} h_{2}(t+\tau-y) dW_{y}   \right] \\
&=\lim_{t\rightarrow \infty} \int _{-\infty}^{t} \int _{-\infty}^{t+\tau} h_{1}(t-x)  h_{2}(t+\tau-y) \mathbb{E} \left[  dW_{x}  dW_{y} \right] \\
&=\lim_{t\rightarrow \infty} \int _{-\infty}^{t} h_{1}(t-s)  h_{2}(t+\tau-s) ds \\
&=\int _{0}^{\infty} h_{1}(r)  h_{2}(r+\tau) dr. \\
\end{aligned}
\end{equation}
\end{proof}

\begin{corollary}
Assume that noise from a single noise source with standard deviation $\sigma$ enters a network, and propagates through $N$ independent pathways to reach the output.
If the impulse response of each of the independent pathways is $h_{1}(t),h_{2}(t),\ldots, h_{N}(t)$ respectively, the mean of the output $y$ will be zero, and its variance equal to
\begin{equation}
V_{out}=\sigma^{2} \int _{0}^{\infty} \left( \sum _{k=1}^{N} a_{k}h_{k}(x)\right)^{2} dx.
\end{equation}
\end{corollary}
\begin{proof}
The output vertex will receive a weighted sum of the outputs of the two independent pathways
\begin{equation}
z(t)=\sum _{k=1}^{N} a_{k}y_{k}(t).
\end{equation}
Its expected value is equal to zero at all times:
\begin{equation}
\begin{aligned}
\mathbb{E}[z(t)] &=\mathbb{E}\left[ \sum _{k=1}^{N} a_{k}y_{k}(t) \right] \\
&= \sum _{k=1}^{N} \mathbb{E}\left[ a_{k}y_{k}(t) \right] \\
&=\sum _{k=1}^{N}  a_{k} \int _{-\infty}^{t} h_{k}(t-x)  \sigma \mathbb{E}\left[ dW_{x}   \right] \\
&=0.
\end{aligned}
\end{equation}

The variance is equal to:
\begin{equation}
\begin{aligned}
\mathbb{V}_{y}&=\lim _{t\rightarrow \infty} \mathbb{V}_{y}(t) =\lim _{t\rightarrow \infty}  \mathbb{E}[z^{2}(t)] \\
&=\lim _{t\rightarrow \infty}  \mathbb{E} \left[ \left( \int _{-\infty}^{t} \sum _{k=1}^{N} a_{k}h_{k}(t-x)dW_{x} \right) \cdot \left(\int _{-\infty}^{t} \sum _{k=1}^{N} a_{k}h_{k}(t-y) dW_{y}\right)  \right]\\
&=\lim _{t\rightarrow \infty} \int _{-\infty}^{t} \int _{-\infty}^{t}  \left( \sum _{k=1}^{N} a_{k}h_{k}(t-x)\right) \cdot \left(\sum _{k=1}^{N} a_{k}h_{k}(t-y)\right)  \mathbb{E} \left[dW_{x} dW_{y} \right] \\
&=\lim _{t\rightarrow \infty} \int _{-\infty}^{t} \sigma^{2} \left( \sum _{k=1}^{N} a_{k}h_{k}(t-s)\right)^{2} ds= \sigma^{2} \int _{0}^{\infty}  \left( \sum _{k=1}^{N} a_{k}h_{k}(x)\right)^{2} dx.
\end{aligned}
\end{equation}
\end{proof}
Suppose we have a linear pathway with each element representing a single-pole linear filter, and we need to pick the position of the poles such that the variance in the output is minimized.
The next lemma shows an easy way to find the pathway if all its vertices are identical and subject to the symmetric constraints.
\begin{definition}
A symmetric multivariable function $f:\mathbb{R}^{n}\rightarrow \mathbb{R}$ is a function for which $f(x)=f(\pi(x))$ where $\pi(x)$ is an arbitrary permutation of the input vector $x$.
\end{definition}

\begin{lemma}
Assume that a symmetric multivariable function $f:\mathbb{R}^{n}\rightarrow \mathbb{R}$ is nowhere constant and  has a sign definite Hessian matrix.
Then it has a unique extremum under symmetric constraints, such that all the elements of the input vector $x$ are equal.
\label{SymmetricMultivariableOptimization}
\end{lemma}
\begin{proof}
Since the Hessian has the same sign everywhere, the function $f$ is strictly convex or strictly concave.
We will assume that $f$ is strictly convex, noting that the proof is the similar when $f$ is concave.
Assume that the extremum of the function $f$ is equal to $f^{*}$, and the argument that achieves this is $x^{*}$.
Further assume that $\min(x^{*})=m$ and $\max(x^{*})=M$ are the minimum and maximum elements of the vector $x^{*}$ respectively.
Since $f$ is symmetric, 
\begin{equation}
f(m,M,x^{*}_{3},\dots ,x^{*}_{n})=f(M,m,x^{*}_{3},\dots ,x^{*}_{n})=f^{*}
\end{equation}
where the arguments still satisfy the symmetric constraints.
But since $f$ is strictly convex, every convex combination of these values will be 
\begin{equation}
\begin{aligned}
f(a,b,x_{3},\dots x_{n}) &\leq tf(m,M,x_{3},\dots ,x_{n})+(1-t)f(M,m,x_{3},\dots ,x_{n})\\
&=tf^{*}+(1-t)f^{*}\\
&=f^{*}.
\end{aligned}
\end{equation}
Generalizing the last argument, it is straightforward to see that
\begin{equation}
f(x_{1},x_{2},\dots x_{n})=f^{*} \quad \textrm{for every} \quad m\leq x_{1},x_{2},\dots x_{n}\leq M.
\end{equation}
Therefore, $f(x)$ needs to be constant in that area, which contradicts the assumption that the function has sign definite Hessian.
\end{proof}

When the constraints are convex but not necessarily symmetric, then we can use the Lagrangian to find the optimal parameters. 
Coming back to the linear pathway network, and assuming that the input is white noise, if the poles of the different nodes are placed at $a_{1},a_{2},\ldots, a_{N}$, the total variance in the output is equal to (see equation \eqref{FrequencyDomainVarianceInputFrequency}): 
\begin{equation}
\begin{aligned}
\mathbb{V}_{out}(a_{1},a_{2},\dots a_{N})&=\frac{1}{2\pi}\int _{-\infty}^{+\infty}\left| \frac{1}{j\omega+a_{1}} \right|^{2}\cdot\left| \frac{1}{j\omega+a_{2}} \right|^{2} \cdots \left| \frac{1}{j\omega+a_{N}} \right|^{2} d\omega \\
&=\frac{1}{2\pi}\int _{-\infty}^{+\infty} \frac{1}{\omega^{2}+a_{1}^{2}}\cdot \frac{1}{\omega^{2}+a_{2}^{2}} \cdots \frac{1}{\omega^{2}+a_{N}^{2}} d\omega.\\
\end{aligned}
\end{equation}
The function $\mathbb{V}_{out}$ is convex with respect to all its arguments $a_{1},a_{2}\dots a_{N}$, as an (infinite) sum of products of convex functions.
Consequently, it has a unique minimum under convex constraints.

The Lagrangian of the function for $\mathbb{V}_{out}$ is
\begin{equation}
\mathcal{L}(a_{1},a_{2},\dots a_{N})=\frac{1}{2\pi}\int _{-\infty}^{+\infty} \frac{1}{\omega^{2}+a_{1}^{2}}\cdot \frac{1}{\omega^{2}+a_{2}^{2}} \cdots \frac{1}{\omega^{2}+a_{N}^{2}} d\omega -\lambda g(a_{1},a_{2},\ldots,a_{N}).\\
\end{equation}
Differentiating with respect to $a_{k}$, under the Leibnitz integral rule:
\begin{equation}
\frac{\partial \mathcal{L}}{\partial a_{k}}=\frac{1}{2\pi}\int _{-\infty}^{+\infty} \frac{1}{\omega^{2}+a_{1}^{2}} \cdots \frac{-2a_{k}}{(\omega^{2}+a_{k}^{2})^2}  \cdots \frac{1}{\omega^{2}+a_{N}^{2}} d\omega= \lambda \frac{\partial g(a_{1},\ldots ,a_{N})}{\partial a_{k}} \\
\end{equation}
for every $k$.
Differentiating with respect to all the parameters will give us $N$ equations, and we have one more equation by requiring $g(a_{1}, \ldots , a_{N})=0$.
So we can solve the system of $N+1$ equations and $N+1$ unknowns $\lambda, a_{1},\ldots a_{N}$, which is guaranteed to have a unique solution as all functions are convex.

In conclusion, we can find the unique minimum of the variance of a linear pathway, when each node is a single pole linear filter with real negative poles.
Given that a linear tree network with independent noise inputs can be decomposed to many linear pathways, this method can be applied to any arbitrary network without cycles.

\section{Correlation, Feedforward and Feedback Cycles}

In a serial pathway where each vertex acts as a filter, the output at each node has a different frequency content as the noise propagates through the network, being filtered at each step.
The variance at each node is decreasing as we move further from the noise source, as is shown in Figure \ref{VarianceVsSerialPathwayLength}.
\begin{figure}[!htb]
\begin{center}
 \begin{minipage}[t]{0.5\linewidth}
 \subfigure{
\begin{pspicture}(-1.5,-1)(5,3)
\psscalebox{0.45}{
{
\cnodeput[](0,0){A}{\strut}
\cnodeput[](2,0){B}{\strut}
\rput[l](3.5,0){\rnode{Etc1}{}}
\rput[l](3.6,0){\rnode{Etc}{\Huge{...}}}
\rput[l](4.3,0){\rnode{Etc2}{}}
\cnodeput[](6,0){D}{\strut}
\cnodeput[](8,0){E}{\strut}
\rput[l](-3,1){\rnode{NoiseInput}{\Huge{Noise}}}
\rput[l](-2,0){\rnode{Input}{}}
\rput[l](8.5,1){\rnode{Output}{\Huge{Output}}}
\rput[l](10,0){\rnode{Output}{}}
}
\nczigzag[coilwidth=0.2cm,coilheight=1.5,coilarm=0.4cm,arrowsize=9pt]{->}{Input}{A}
\ncline[arrowsize=8pt]{->}{A}{B}
\ncline[arrowsize=8pt]{->}{B}{Etc1}
\ncline[arrowsize=8pt]{->}{Etc2}{D}
\ncline[arrowsize=8pt]{->}{D}{E}
\ncline[arrowsize=8pt]{->}{E}{Output}
}
\end{pspicture}
}
\setcounter{subfigure}{0}
\subfigure{
\includegraphics[scale=0.18]{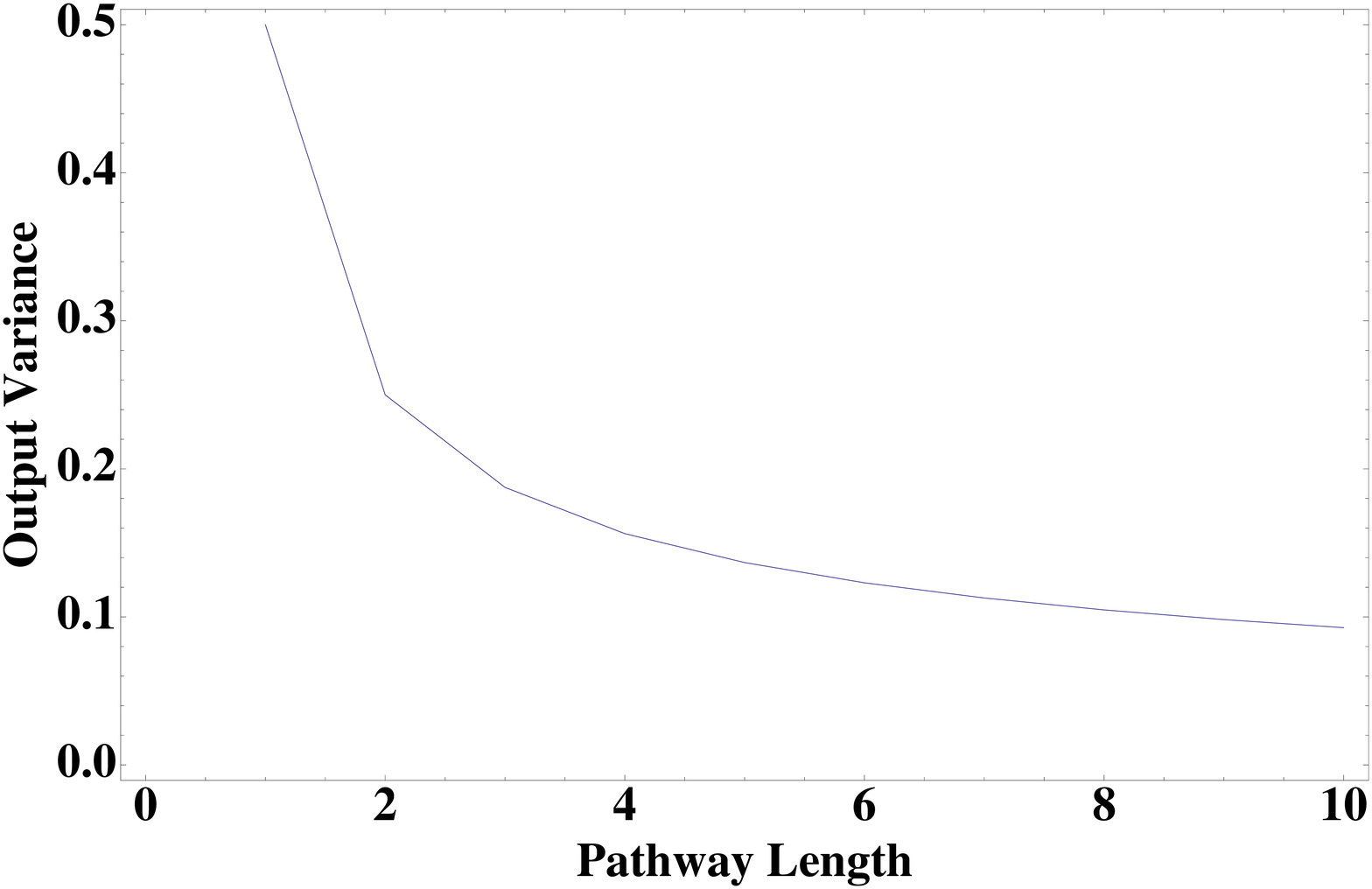}
}
\end{minipage}\hfill
\begin{minipage}[t]{0.5\linewidth}
\subfigure{
\begin{pspicture}(-1.5,-1)(5,3)
\psscalebox{0.45}{
{
\cnodeput[](0,0){A}{\strut}
\cnodeput[](2,0){B}{\strut}
\rput[l](3.5,0.15){\rnode{Etc1U}{}}
\rput[l](3.5,-0.15){\rnode{Etc1D}{}}
\rput[l](3.6,0){\rnode{Etc}{\Huge{...}}}
\rput[l](4.3,0.15){\rnode{Etc2U}{}}
\rput[l](4.3,-0.15){\rnode{Etc2D}{}}
\cnodeput[](6,0){D}{\strut}
\cnodeput[](8,0){E}{\strut}
\rput[l](-3,1){\rnode{NoiseInput}{\Huge{Noise}}}
\rput[l](-2,0){\rnode{Input}{}}
\rput[l](8.5,1){\rnode{Output}{\Huge{Output}}}
\rput[l](10,0){\rnode{Output}{}}
}
\nczigzag[coilwidth=0.2cm,coilheight=1.5,coilarm=0.4cm,arrowsize=9pt]{->}{Input}{A}
\ncarc[arrowsize=5pt,arcangle=15]{->}{A}{B}
\ncarc[arrowsize=5pt,arcangle=15]{->}{B}{A}
\ncarc[arrowsize=5pt,arcangle=15]{->}{B}{Etc1U}
\ncarc[arrowsize=5pt,arcangle=15]{->}{Etc1D}{B}
\ncarc[arrowsize=5pt,arcangle=15]{->}{Etc2U}{D}
\ncarc[arrowsize=5pt,arcangle=15]{->}{D}{Etc2D}
\ncarc[arrowsize=5pt,arcangle=15]{->}{D}{E}
\ncarc[arrowsize=5pt,arcangle=15]{->}{E}{D}
\ncline[arrowsize=8pt]{->}{E}{Output}
}
\end{pspicture}
}
\setcounter{subfigure}{1}
\subfigure{
\includegraphics[scale=0.18]{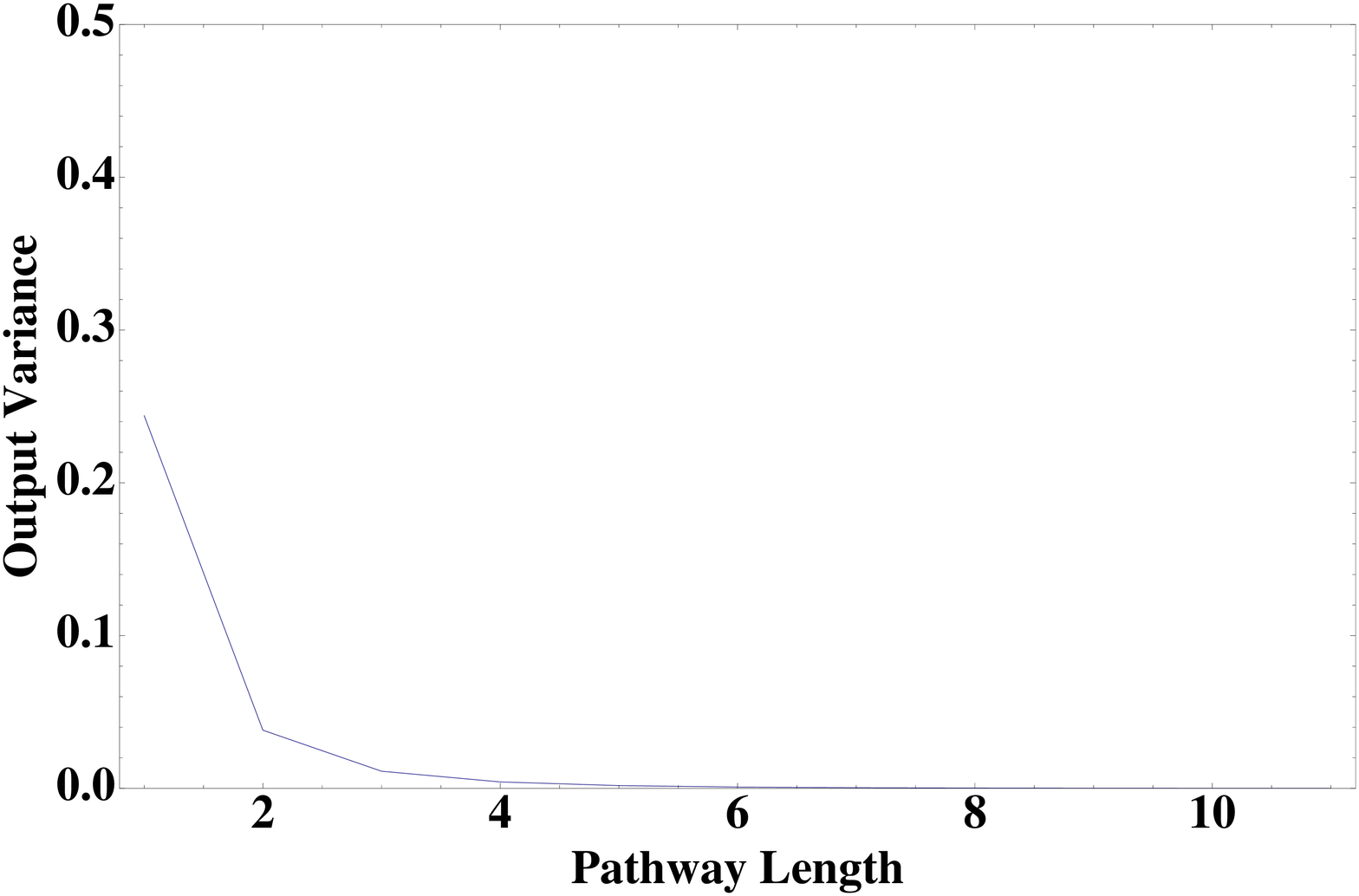}
}
\end{minipage}
\end{center}
\caption{Variance of the output of a unidirectional and a bidirectional serial pathway as a function of the pathway length. All nodes are assumed to be identical single-pole filters. In the unidirectional pathway, each node is affected only by the node immediately preceding it, whereas in the bidirectional pathway each intermediate vertex is receiving input from the node preceding and the node succeeding it. The bidirectional pathway is much more efficient in filtering out noise. The variance for both pathways decreases with the pathway length.
The bidirectional pathway has variance very close to zero even when it is relatively short.}
\label{VarianceVsSerialPathwayLength}
\end{figure}
As the serial pathway becomes longer, the input and the output become less correlated since their distance  increases.
In addition, every node changes the phase of its inputs, which also contributes to the decreased correlation.
Therefore, applying negative feedback or setting up a feedforward cycle can only have a measurable effect if the cycle length is relatively small.
Figure \ref{CovarianceAndCorrelationSerialPathway} shows the covariances and correlations among the vertices of two simple linear pathways, one unidirectional one bidirectional, as they are depicted in Figure \ref{VarianceVsSerialPathwayLength}.
\begin{figure}[htbp]
\centering
\subfigure[Unidirectional Pathway Covariance Matrix]{
\centering
\includegraphics[scale=0.19]{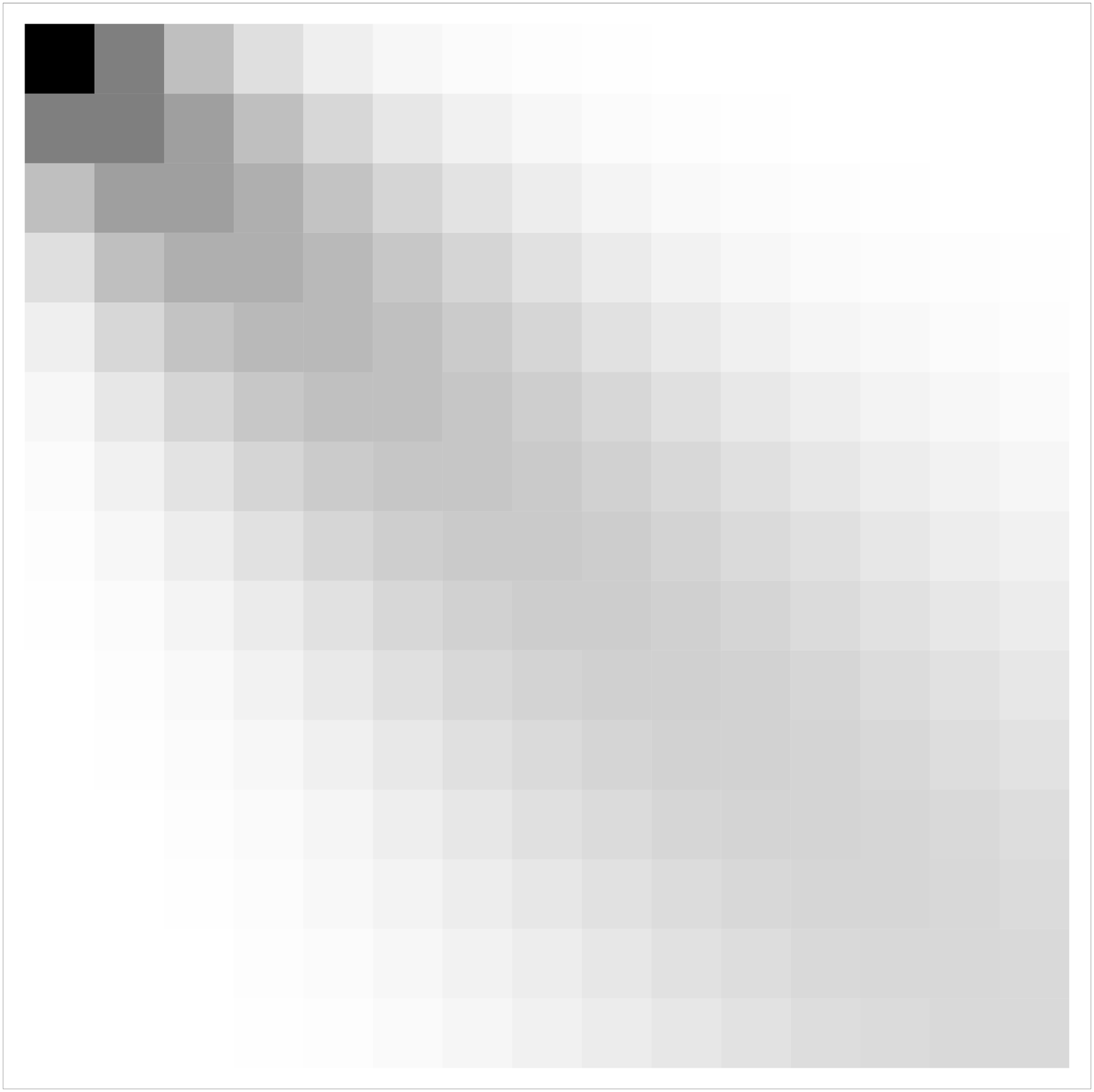}
}
\subfigure[Unidirectional Pathway Correlation Matrix]{
\centering
\includegraphics[scale=0.19]{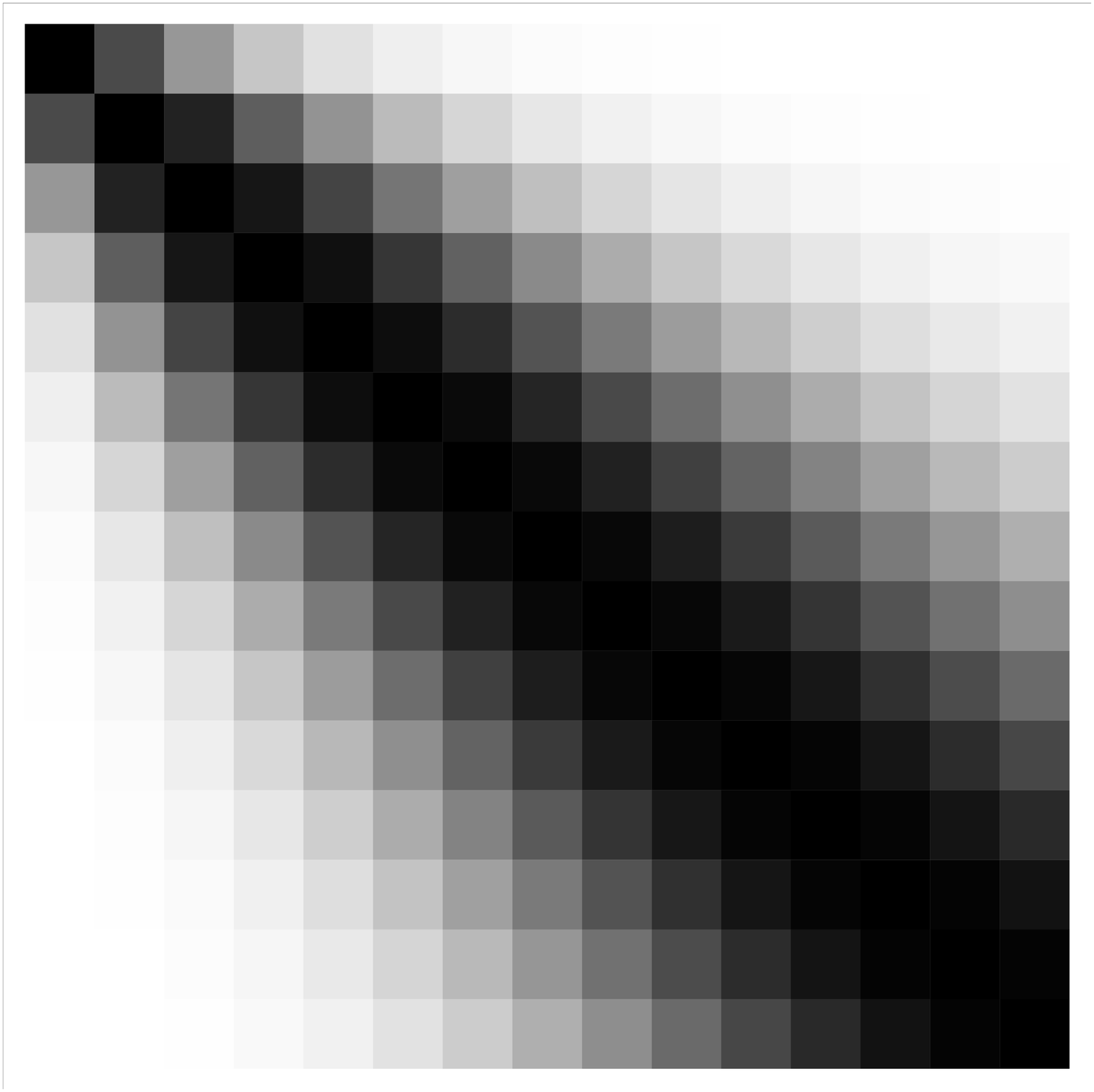}
}
\subfigure[Bidirectional Pathway Covariance Matrix]{
\centering
\includegraphics[scale=0.65]{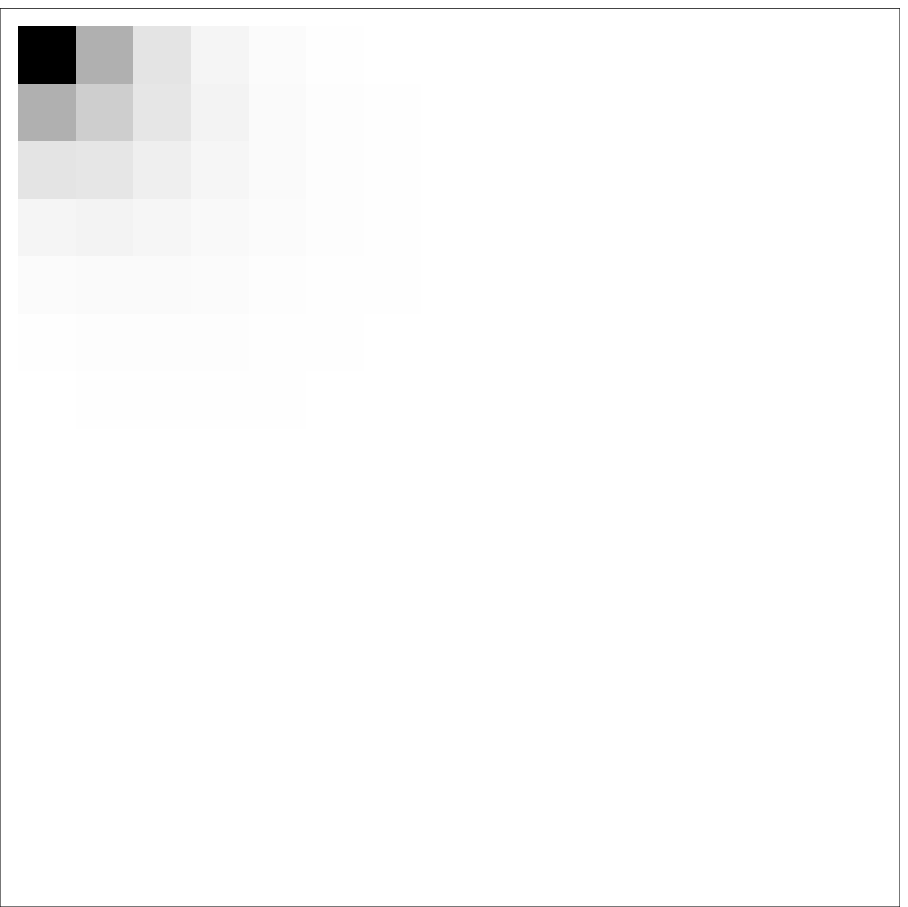}
}
\subfigure[Bidirectional Pathway Correlation Matrix]{
\centering
\includegraphics[scale=0.65]{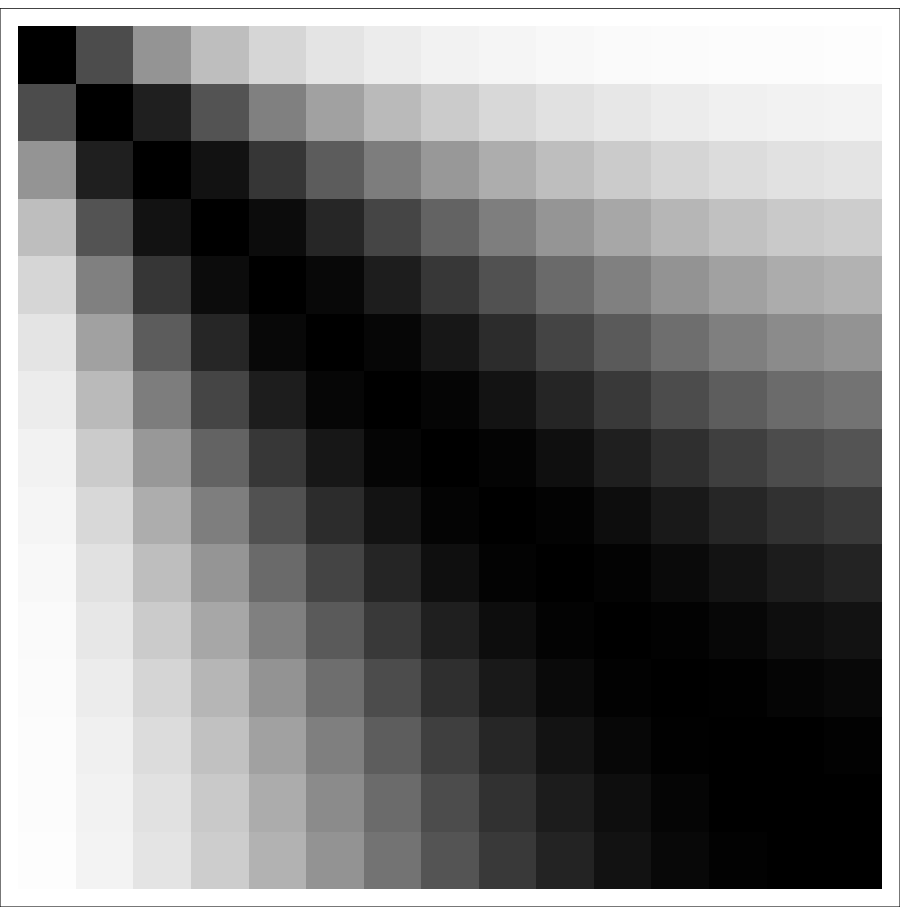}
}
\caption{Covariance and correlation among all pairs of nodes in a linear pathway. Every square $(x,y)$ in the matrices above corresponds to the value of their correlation $R_{x,y}(\tau=0)$ of nodes of distance $x$ and $y$ from the origin, $0\leq x,y \leq N-1$. The larger the correlation, the darker the respective square.  As the distance $|x-y|$ among the nodes increases,  their covariance and correlation decreases. The covariance among nodes of the same distance in the unidirectional pathway decreases, and the correlation among them increases towards the end. The covariance of the nodes in the bidirectional pathway is essentially zero within a small distance, and the correlation is larger even when the distance is relatively large.}
\label{CovarianceAndCorrelationSerialPathway}
\end{figure}

Cycles can significantly increase the effect of noise in the system.
There are two reasons for this: First the noise can now reach more vertices since the average distance among nodes decreases, and second, every node now receives the same disturbance from at least two different paths, and the two signals are  correlated, contributing to larger variance. 
An example is shown in Figure \ref{AverageVariancePathwayVsCycle}, where we compare the average variance of two systems whose only difference is the connection between the first and the last node.
Both networks receive the same inputs, but in the cycle network, the variance is much larger.
The result of the noise is even more pronounced when there is correlation among the noise inputs to different nodes.
\begin{figure}[!htb]
\begin{center}
 \begin{minipage}[t]{0.5\linewidth}
 \centering
 \subfigure{
 \psscalebox{0.5}{
\begin{pspicture}(-1,-2)(8,3)
{
\cnodeput[](0,0){A}{\strut}
\cnodeput[](2,0){B}{\strut}
\cnodeput[](4,0){C}{\strut}
\cnodeput[](6,0){D}{\strut}
\cnodeput[](8,0){E}{\strut}
\rput[l](0,2){\rnode{InputA}{}}
\rput[l](2,2){\rnode{InputB}{}}
\rput[l](4,2){\rnode{InputC}{}}
\rput[l](6,2){\rnode{InputD}{}}
\rput[l](8,2){\rnode{InputE}{}}
}
\ncline[arrowsize=8pt]{->}{Input}{A}
\ncarc[arrowsize=5pt,arcangle=15]{->}{A}{B}
\ncarc[arrowsize=5pt,arcangle=15]{->}{B}{A}
\ncarc[arrowsize=5pt,arcangle=15]{->}{B}{C}
\ncarc[arrowsize=5pt,arcangle=15]{->}{C}{B}
\ncarc[arrowsize=5pt,arcangle=15]{->}{C}{D}
\ncarc[arrowsize=5pt,arcangle=15]{->}{D}{C}
\ncarc[arrowsize=5pt,arcangle=15]{->}{D}{E}
\ncarc[arrowsize=5pt,arcangle=15]{->}{E}{D}
\nczigzag[coilwidth=0.2cm,coilheight=1.5,coilarm=0.4cm,arrowsize=9pt]{->}{InputA}{A}
\nczigzag[coilwidth=0.2cm,coilheight=1.5,coilarm=0.4cm,arrowsize=9pt]{->}{InputB}{B}
\nczigzag[coilwidth=0.2cm,coilheight=1.5,coilarm=0.4cm,arrowsize=9pt]{->}{InputC}{C}
\nczigzag[coilwidth=0.2cm,coilheight=1.5,coilarm=0.4cm,arrowsize=9pt]{->}{InputD}{D}
\nczigzag[coilwidth=0.2cm,coilheight=1.5,coilarm=0.4cm,arrowsize=9pt]{->}{InputE}{E}
\end{pspicture}
}
}
\setcounter{subfigure}{0}
\centering
\subfigure{
\includegraphics[scale=0.18]{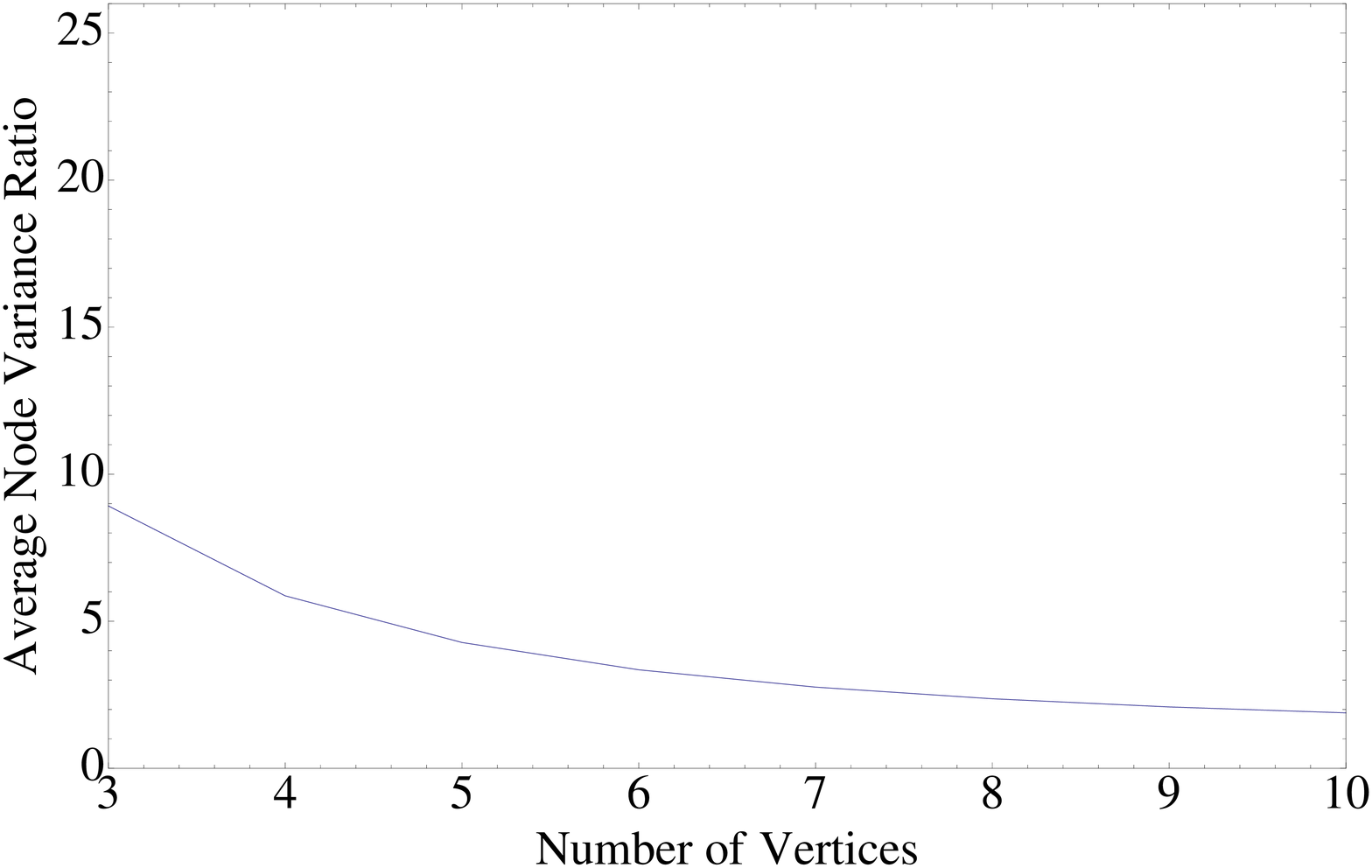}
}
\end{minipage}\hfill
\begin{minipage}[t]{0.5\linewidth}
\centering
\subfigure{
\psscalebox{0.5}{
\begin{pspicture}(-4,-3)(4,2)
{
\cnodeput[](2,0){A}{\strut}
\cnodeput[](0.618,1.902){B}{\strut}
\cnodeput[](-1.618,1.175){C}{\strut}
\cnodeput[](-1.618,-1.175){D}{\strut}
\cnodeput[](0.618,-1.902){E}{\strut}
\rput[l](4,0){\rnode{InputA}{}}
\rput[l](2.6,1.9){\rnode{InputB}{}}
\rput[l](-3.5,1.175){\rnode{InputC}{}}
\rput[l](-3.5,-1.175){\rnode{InputD}{}}
\rput[l](2.6,-1.902){\rnode{InputE}{}}
}
\ncline[arrowsize=8pt]{->}{Input}{A}
\ncarc[arrowsize=5pt,arcangle=15]{->}{A}{B}
\ncarc[arrowsize=5pt,arcangle=15]{->}{B}{A}
\ncarc[arrowsize=5pt,arcangle=15]{->}{B}{C}
\ncarc[arrowsize=5pt,arcangle=15]{->}{C}{B}
\ncarc[arrowsize=5pt,arcangle=15]{->}{C}{D}
\ncarc[arrowsize=5pt,arcangle=15]{->}{D}{C}
\ncarc[arrowsize=5pt,arcangle=15]{->}{D}{E}
\ncarc[arrowsize=5pt,arcangle=15]{->}{E}{D}
\ncarc[arrowsize=5pt,arcangle=15]{->}{A}{E}
\ncarc[arrowsize=5pt,arcangle=15]{->}{E}{A}
\nczigzag[coilwidth=0.2cm,coilheight=1.5,coilarm=0.4cm,arrowsize=9pt]{->}{InputA}{A}
\nczigzag[coilwidth=0.2cm,coilheight=1.5,coilarm=0.4cm,arrowsize=9pt]{->}{InputB}{B}
\nczigzag[coilwidth=0.2cm,coilheight=1.5,coilarm=0.4cm,arrowsize=9pt]{->}{InputC}{C}
\nczigzag[coilwidth=0.2cm,coilheight=1.5,coilarm=0.4cm,arrowsize=9pt]{->}{InputD}{D}
\nczigzag[coilwidth=0.2cm,coilheight=1.5,coilarm=0.4cm,arrowsize=9pt]{->}{InputE}{E}
\end{pspicture}
}
}
\setcounter{subfigure}{1}
\centering
\subfigure{
\includegraphics[scale=0.18]{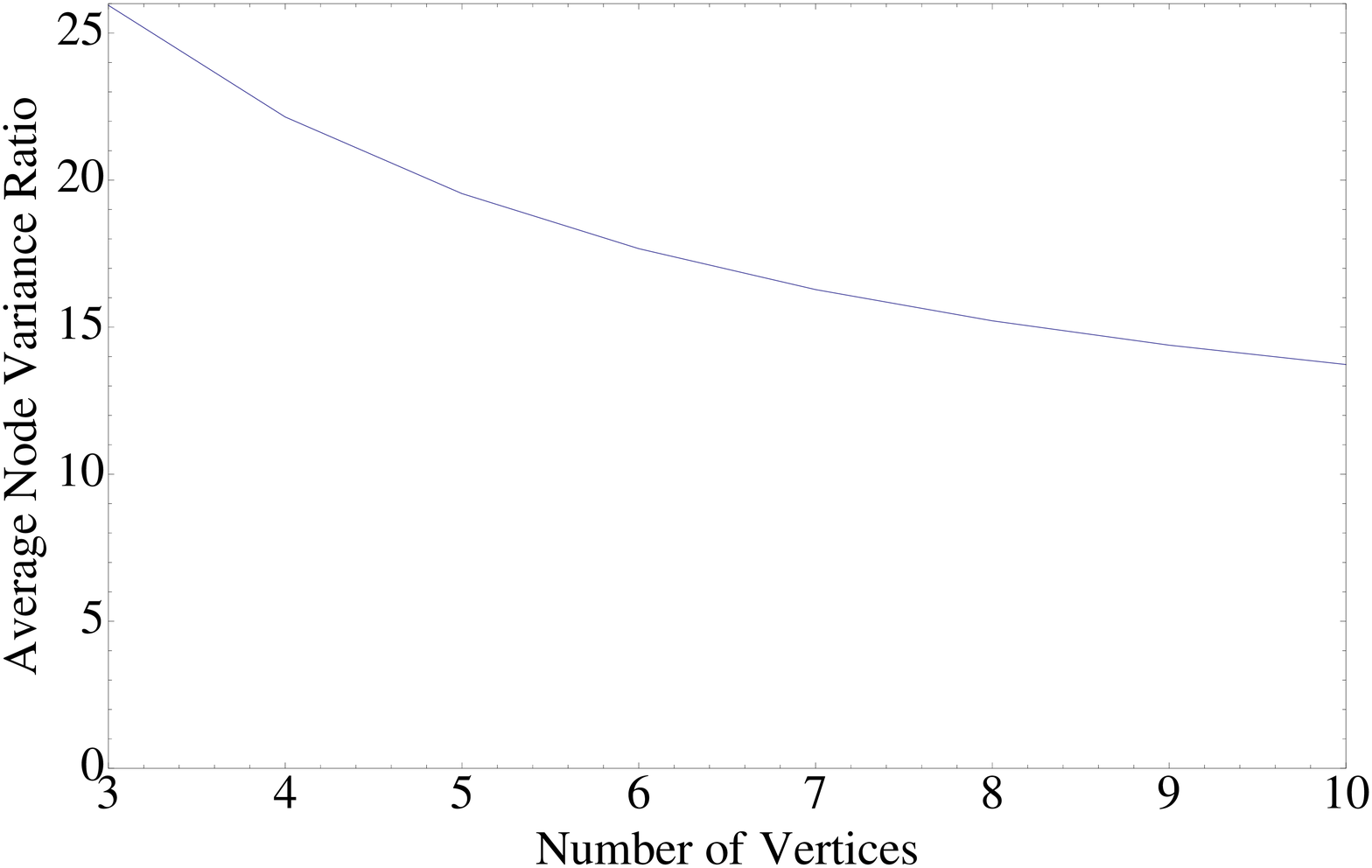}
}
\end{minipage}
\end{center}
\caption{Average variance of all nodes in a network in a cycle as compared to an identical network without the feedback loop. Every node has a noise input which is then spread through the network. The average variance of all the nodes for both the cycle  is normalized by the variance of the respective serial pathway.
The variance cycle is always much larger than the variance of the simple serial pathway when the noise inputs for each node are uncorrelated (bottom left). The ratio becomes even larger when the inputs are correlated (bottom right).}
\label{AverageVariancePathwayVsCycle}
\end{figure}

The effect of cycles on the output noise can be reduced if we make sure that each independent pathway also changes the phase of its input by different amounts.
Different phases in the output (for at least a relatively large frequency spectrum) will ensure that the various frequencies partially cancel each other, reducing the output variation.
When a pathway significantly reduces the frequency content, or has small gain for most frequencies, then correlations do not play a significant role.
This behavior is clearly shown in Figure \ref{TwoPathsDifferentLengthsUnidirectional} for a unidirectional cycle and in Figure \ref{TwoPathsDifferentLengthsBidirectional} for a bidirectional cycle.
Phase shifts in a pathway are equivalent to time delays, as we will see in the next section.

\begin{figure}[!htb]
\centering
\subfigure[]{
\centering
\psscalebox{0.5}{
\begin{pspicture}(-3,-3)(13,5)
\pnode(0,1){I}
\rput[l](-1,2){\rnode{NoiseInput}{\Huge{Noise}}}
\cnodeput(2,1){A}{\strut}
\cnodeput(4,2){B}{\strut}
\cnodeput(6,2){C}{\strut}
\cnodeput(5,0){D}{\strut}
\cnodeput(8,1){F}{\strut}
\pnode(10,1){T}
\rput[l](9,2){\rnode{VarianceOutput}{\Huge{Output}}}
\nczigzag[coilwidth=0.2cm,coilheight=1.5,coilarm=0.4cm,arrowsize=9pt]{->}{I}{A}
\ncline[arrowsize=8pt]{->}{A}{B}
\ncline[arrowsize=8pt]{->}{B}{C}
\ncline[arrowsize=8pt]{->}{A}{D}
\ncline[arrowsize=8pt]{->}{C}{F}
\ncline[arrowsize=8pt]{->}{D}{F}
\ncline[arrowsize=8pt]{->}{F}{T}
\end{pspicture}
}
}
\subfigure[]{
\centering
\includegraphics[scale=0.18]{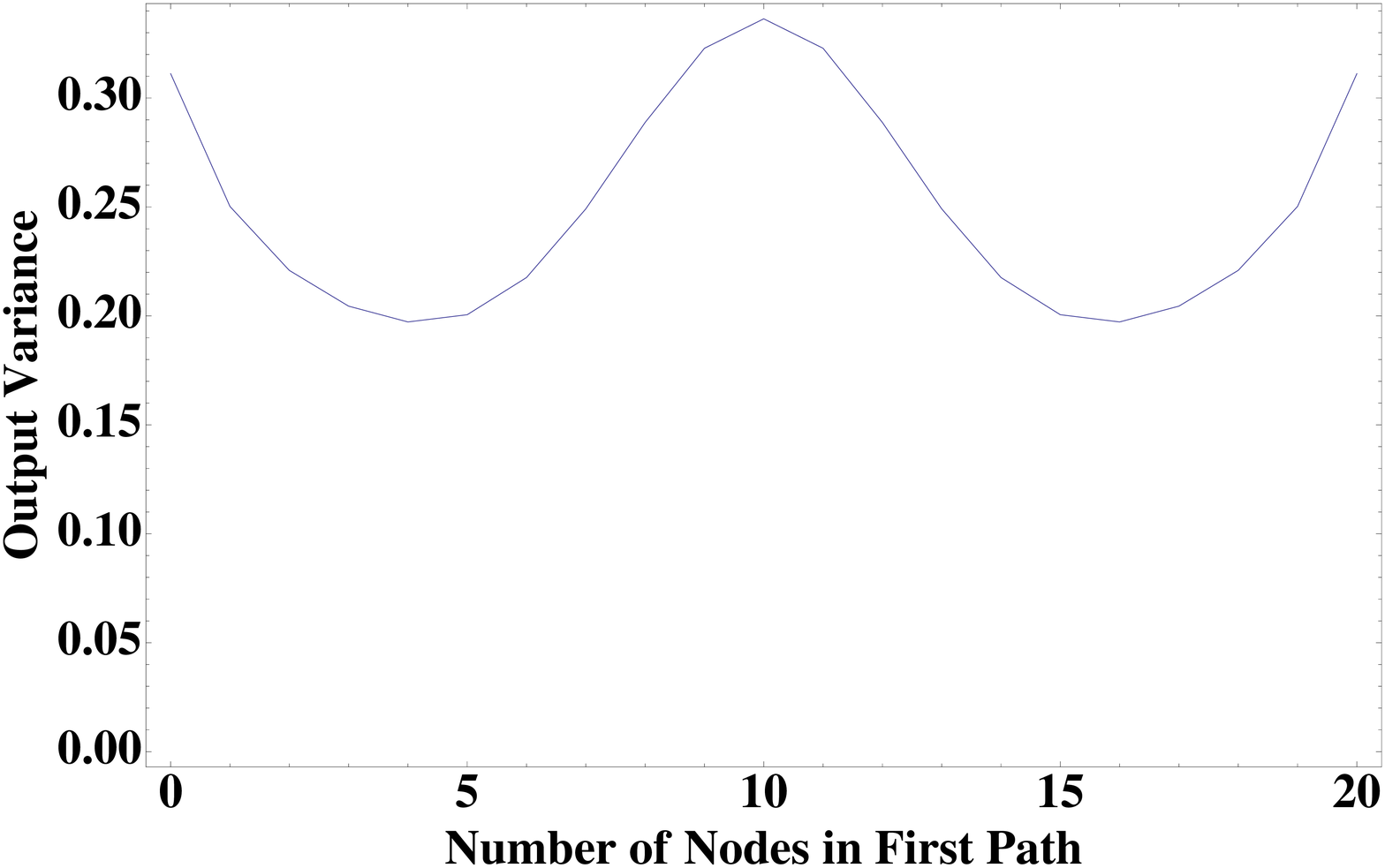}
}
\subfigure[]{
\centering
\includegraphics[scale=0.18]{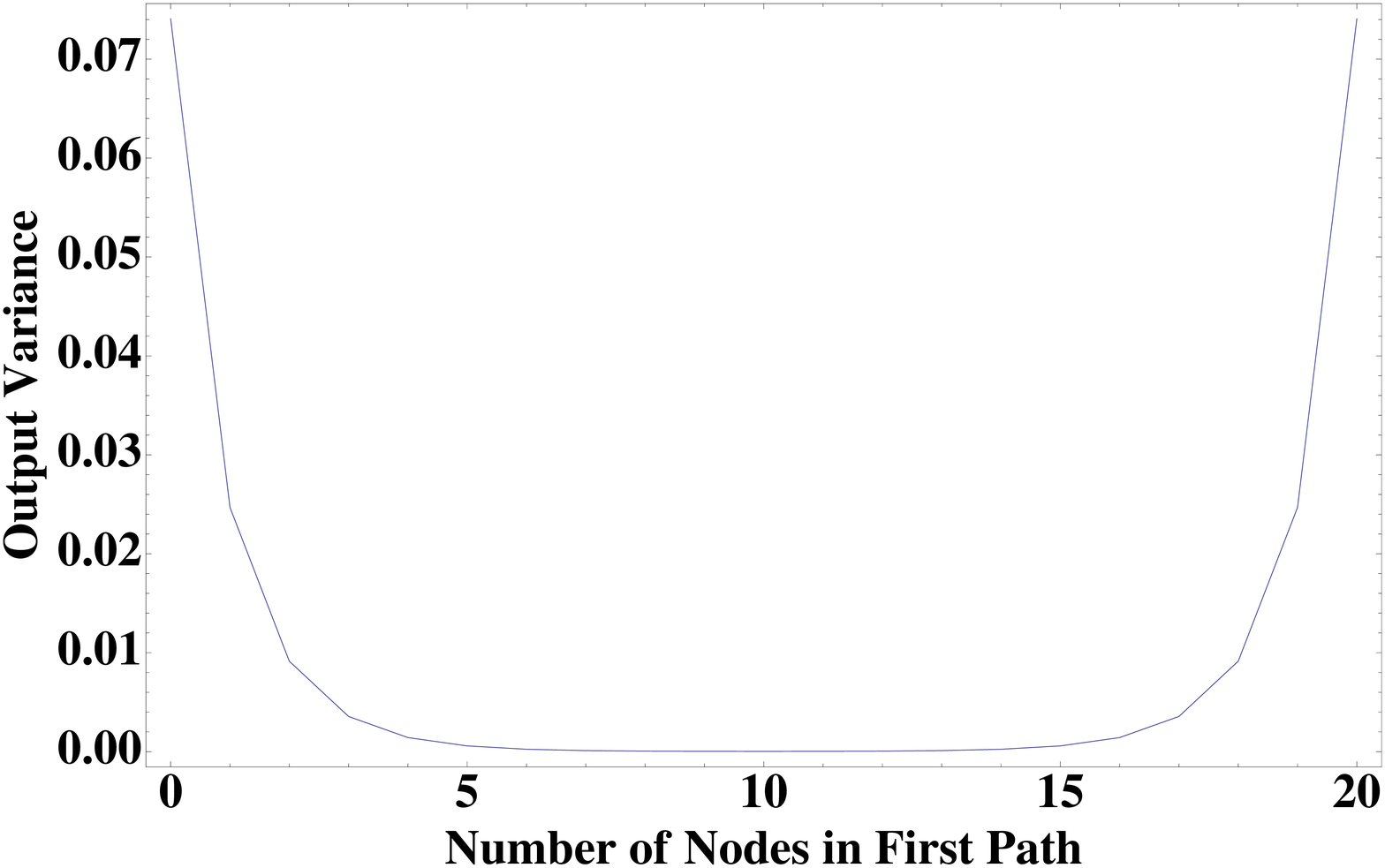}
}
\caption{A network consisting of a feedforward cycle and the corresponding noise strength in its output. If the nodes of the network have poles with relatively small absolute values, then the output variance may be larger than the variance in the intermediate nodes. A fixed number of identical nodes is divided into two pathways, whose output is combined in the output node. If the number of nodes is similar in both pathways, then their outputs are highly correlated, and when combined produce large random swings. This does not happen when the poles of each node have a large negative real part (right). In the first case, the poles are placed at $a=-1$ whereas in the second the poles are place at $a=-1.5$.}
\label{TwoPathsDifferentLengthsUnidirectional}
\end{figure}
Similarly, negative feedback carefully applied to a network contributes to better disturbance rejection.
When the disturbance is white noise, the effect of feedback is smaller as the feedback cycle gets longer.
\begin{figure}[!htb]
\centering
\subfigure[]{
\centering
\psscalebox{0.5}{
\begin{pspicture}(-3,-2)(13,3)
\pnode(0,1){I}
\rput[l](-1,2){\rnode{NoiseInput}{\Huge{Noise}}}
\cnodeput(2,1){A}{\strut}
\cnodeput(5,2){B}{\strut}
\cnodeput(4,0){C}{\strut}
\cnodeput(6,0){D}{\strut}
\cnodeput(8,1){F}{\strut}
\pnode(10,1){T}
\rput[l](9,2){\rnode{VarianceOutput}{\Huge{Output}}}
\nczigzag[coilwidth=0.2cm,coilheight=1.5,coilarm=0.4cm,arrowsize=9pt]{->}{I}{A}
\ncline[arrowsize=8pt]{->}{F}{T}
\ncarc[arrowsize=5pt,arcangle=9]{->}{A}{B}
\ncarc[arrowsize=5pt,arrowsize=5pt,arcangle=9]{->}{B}{A}
\ncarc[arrowsize=5pt,arcangle=9]{->}{B}{F}
\ncarc[arrowsize=5pt,arcangle=9]{->}{F}{B}
\ncarc[arrowsize=5pt,arcangle=9]{->}{A}{C}
\ncarc[arrowsize=5pt,arcangle=9]{->}{C}{A}
\ncarc[arrowsize=5pt,arcangle=9]{->}{C}{D}
\ncarc[arrowsize=5pt,arcangle=9]{->}{D}{C}
\ncarc[arrowsize=5pt,arcangle=9]{->}{D}{F}
\ncarc[arrowsize=5pt,arcangle=9]{->}{F}{D}
\end{pspicture}
}
}
\subfigure[]{
\centering
\includegraphics[scale=0.16]{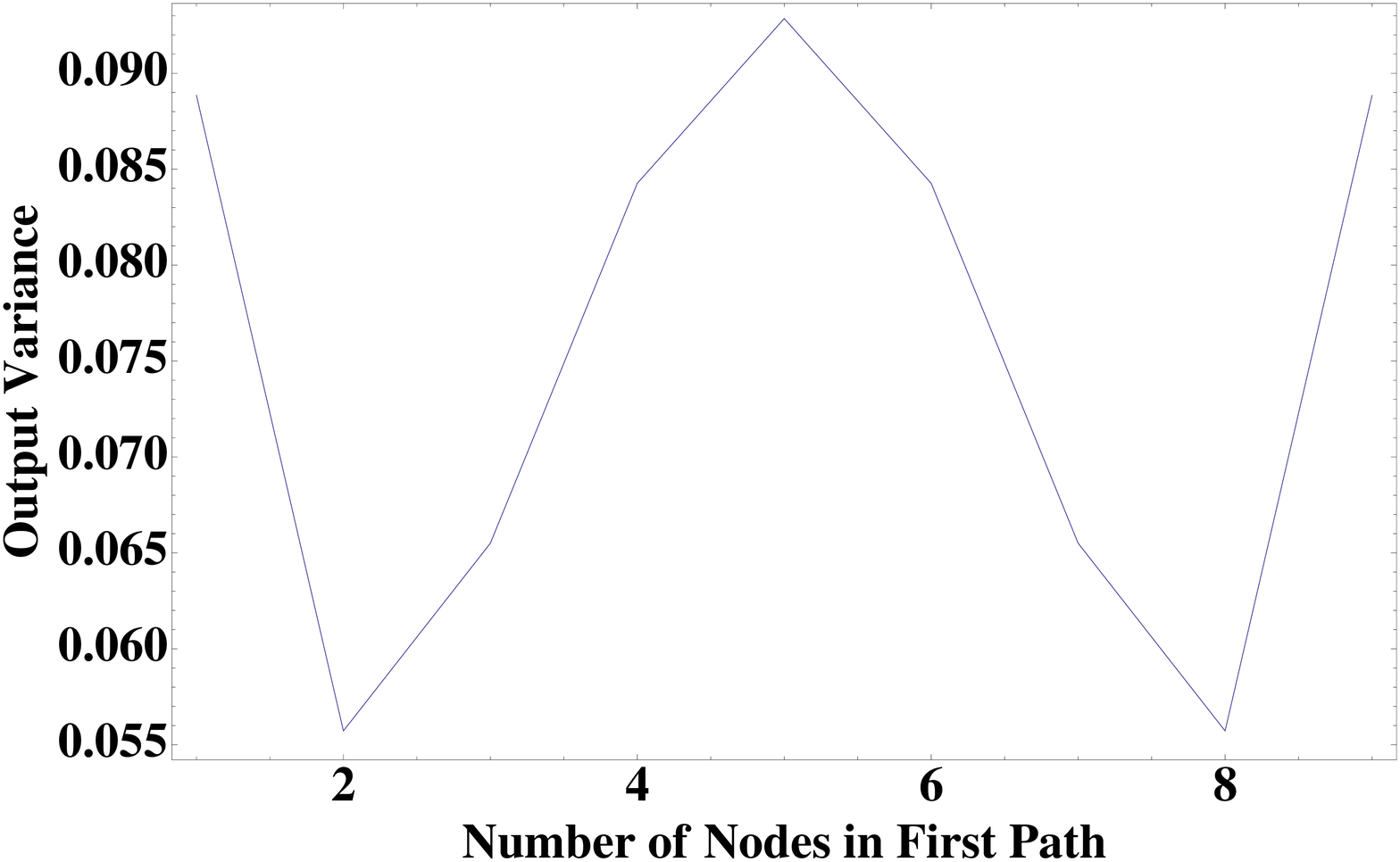}
}
\subfigure[]{
\centering
\includegraphics[scale=0.16]{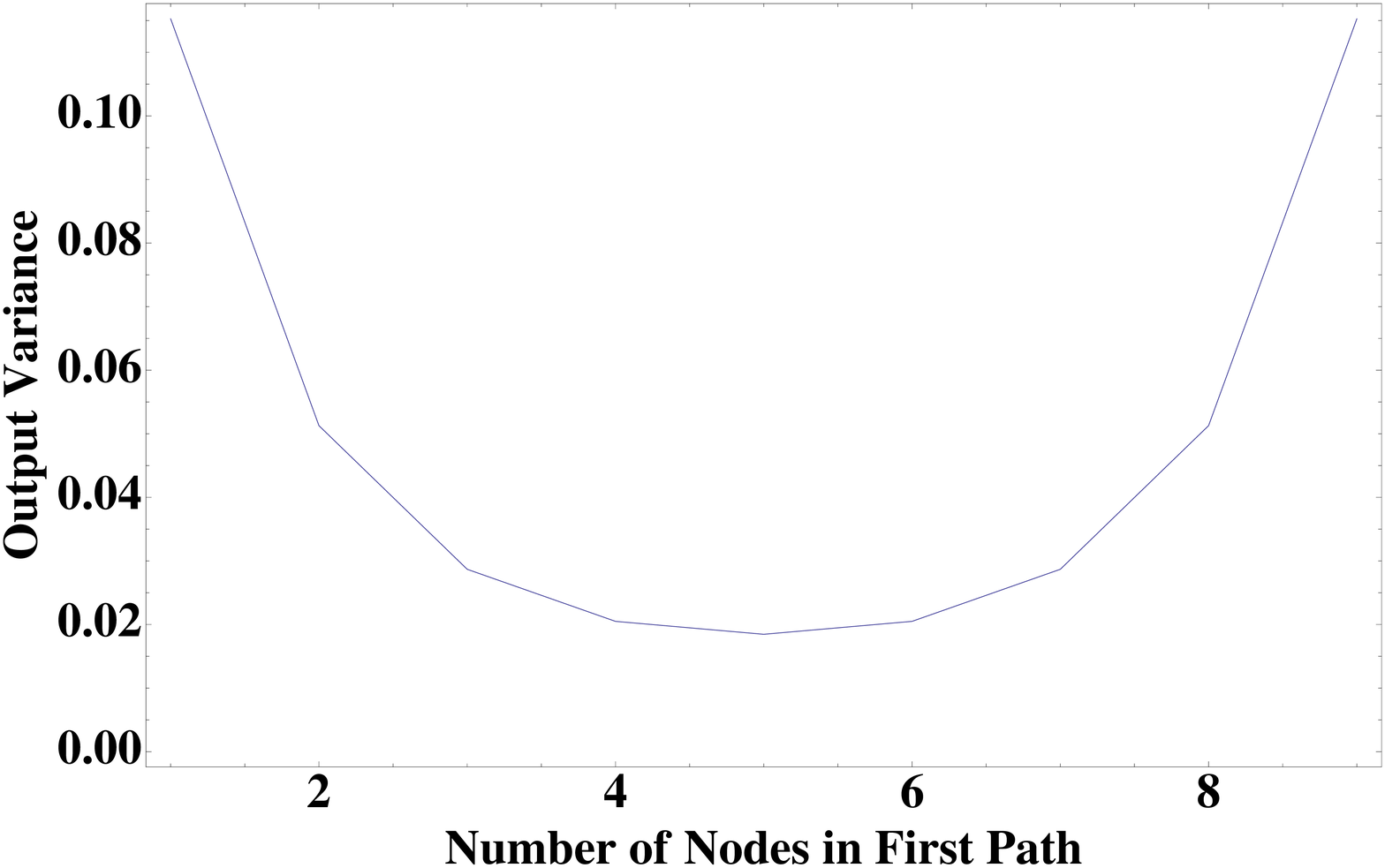}
}
\caption{Correlations increase the variance in bidirectional networks. If the outputs of two pathways that are correlated are combined, then the output has relatively large variance. Here, a single output receives input from two pathways of different lengths, which consist of identical nodes.
Bidirectional pathways filter noise very effectively as shown before, and the output variance is still small.}
\label{TwoPathsDifferentLengthsBidirectional}
\end{figure}

The correlation and covariance among vertices decreases with distance and the variance of each node decreases as the length of the pathway increases.
Furthermore, as we move towards the end of the pathway, the covariance of nodes of a given distance decreases but the correlation of nodes of a given distance increases.
The last observation is easily justified taking into account that each new node introduces a virtual filter, and the output of nodes will tend to have very similar frequency content the more filters it has gone through.
Moreover, from the Bode plot of a filter, we can easily see that for the frequencies that are not affected by the filter, their phase is also relatively unaffected, which does not decrease their correlation.

The previous analysis hints to the fact that feedback cycles have limited utility when applied to long pathways.
Figure \ref{FeedbackStatisticalAnalysis} shows the variance of the output after we apply negative feedback to a linear pathway.
The darkness of each element $(m,n)$ of the upper triangular matrix shows the standard deviation of the pathway output when we apply feedback from node $n$ to node $m$.
As one would expect, the effect of feedback is directly proportional to the correlation between the source and target vertices. The same holds for feedforward loops, both positive and negative.
\begin{figure}[htbp]
\subfigure[Feedback Topology Figure]{
\begin{pspicture}(0,-2)(6,1)
\psscalebox{0.4}{
{
\cnodeput[](0,0){A}{\strut}
\cnodeput[](2,0){B}{\strut\boldmath$m$}
\cnodeput[](4,0){C}{\strut}
\cnodeput[](6,0){D}{\strut}
\cnodeput[](8,0){E}{\strut}
\cnodeput[](10,0){F}{\strut\boldmath$n$}
\cnodeput[](12,0){G}{\strut}
\rput[l](-2,1){\rnode{NoiseInput}{\Huge{Noise}}}
\rput[l](-2,0){\rnode{Input}{}}
\rput[l](14,0){\rnode{Output}{}}
\rput[l](12,1){\rnode{VarianceOutput}{\Huge{Output}}}
\rput[l](1,-1){\rnode{Minus}{\Huge{$-$} }}
}
\nczigzag[coilwidth=0.2cm,coilheight=1.5,coilarm=0.4cm,arrowsize=9pt]{->}{Input}{A}
\ncline[arrowsize=10pt]{->}{A}{B}
\ncline[arrowsize=10pt]{->}{B}{C}
\ncline[arrowsize=10pt]{->}{C}{D}
\ncline[arrowsize=10pt]{->}{D}{E}
\ncline[arrowsize=10pt]{->}{E}{F}
\ncline[arrowsize=10pt]{->}{F}{G}
\ncline[arrowsize=10pt]{->}{G}{Output}
\ncangle[angleA=270,angleB=-90,arm=2,arrowsize=10pt]{->}{F}{B}
}
\end{pspicture}
}
\subfigure[Output Variance]{
\includegraphics[scale=0.2]{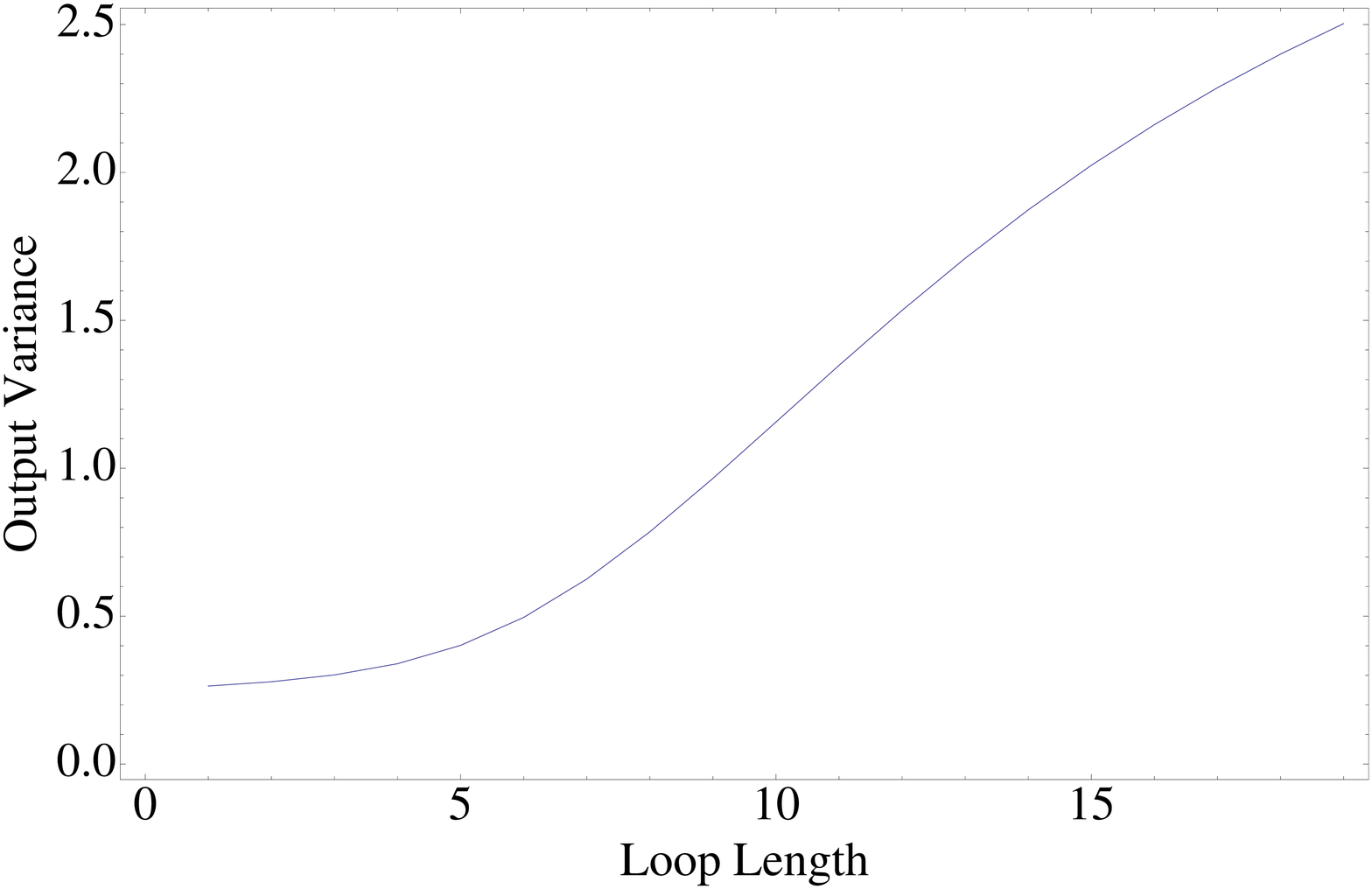}
}
\caption{A serial pathway with a unit feedback loop. The matrix on the right consists of squares $(m,n)$, each of which represents the variance of the output when feedback is applied from node $n$ to node $m$. The result of the feedback loop only depends on the distance $d=|n-m|$, and the variance decreases as the length of the feedback loop becomes smaller, and vice versa. }
\label{FeedbackStatisticalAnalysis}
\end{figure}

\begin{figure}[htbp]
\centering
\subfigure[Feedforward Loop]{
\centering
\begin{pspicture}(-2,-1)(8,1)
\psscalebox{0.4}{
{
\cnodeput[](0,0){A}{\strut}
\cnodeput[](2,0){B}{\strut \boldmath$m$}
\cnodeput[](4,0){C}{\strut}
\cnodeput[](6,0){D}{\strut}
\cnodeput[](8,0){E}{\strut}
\cnodeput[](10,0){F}{\strut \boldmath$n$}
\cnodeput[](12,0){G}{\strut}
\rput[l](-3,1){\rnode{Input}{\Huge{Noise}}}
\rput[l](-2,0){\rnode{Input}{}}
\rput[l](13,1){\rnode{VarianceOutput}{\Huge{Output}}}
\rput[l](14,0){\rnode{Output}{}}

}
\ncline[arrowsize=10pt]{->}{Input}{A}
\ncline[arrowsize=10pt]{->}{A}{B}
\ncline[arrowsize=10pt]{->}{B}{C}
\ncline[arrowsize=10pt]{->}{C}{D}
\ncline[arrowsize=10pt]{->}{D}{E}
\ncline[arrowsize=10pt]{->}{E}{F}
\ncline[arrowsize=10pt]{->}{F}{G}
\ncline[arrowsize=10pt]{->}{G}{Output}
\ncangle[angleA=270,angleB=-90,arm=2,arrowsize=10pt]{->}{B}{F}
}
\end{pspicture}
}
\subfigure[Negative Feedforward Loop]{
\centering
\includegraphics[scale=0.16]{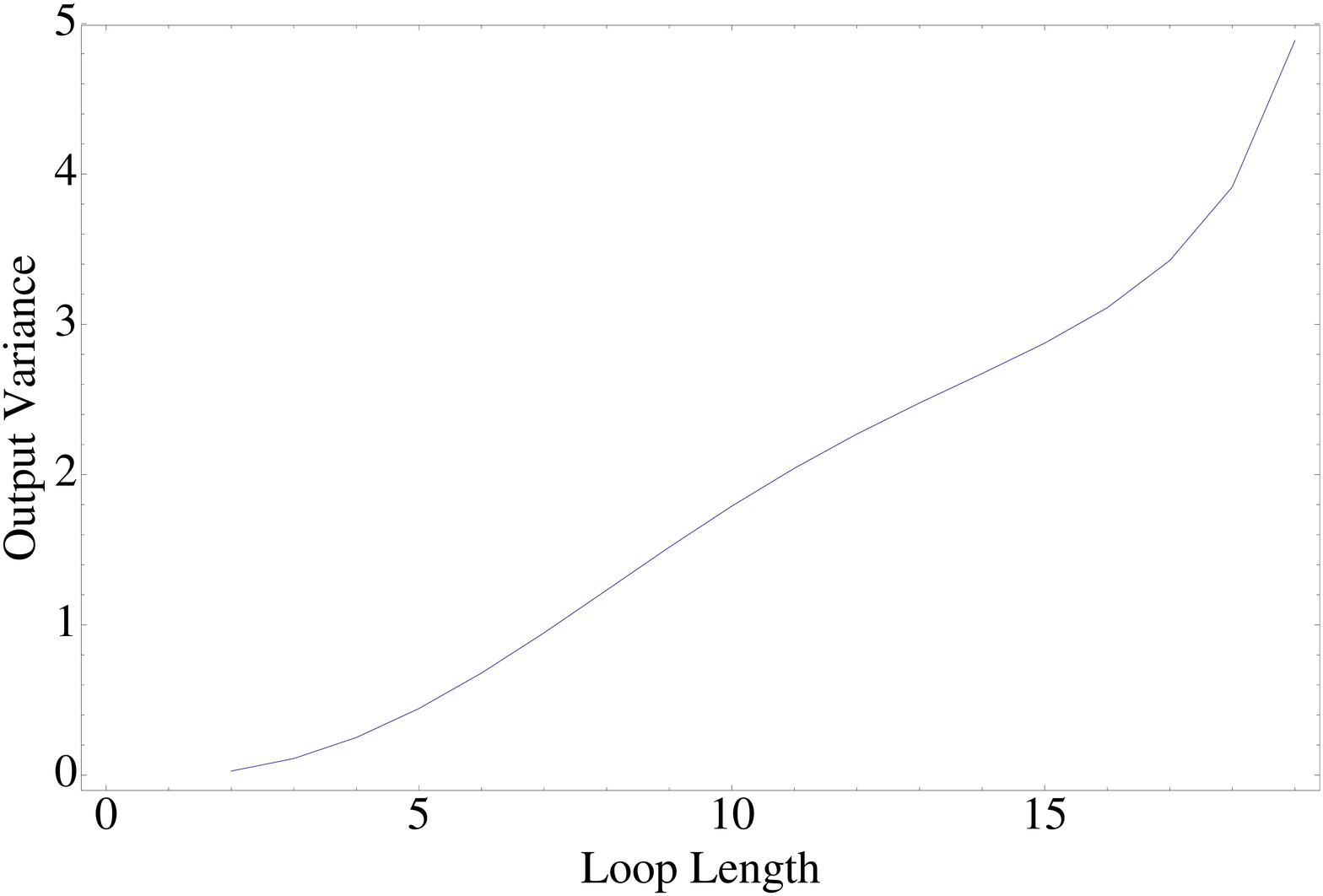}
}
\subfigure[Positive Feedforward Loop]{
\centering
\includegraphics[scale=0.16]{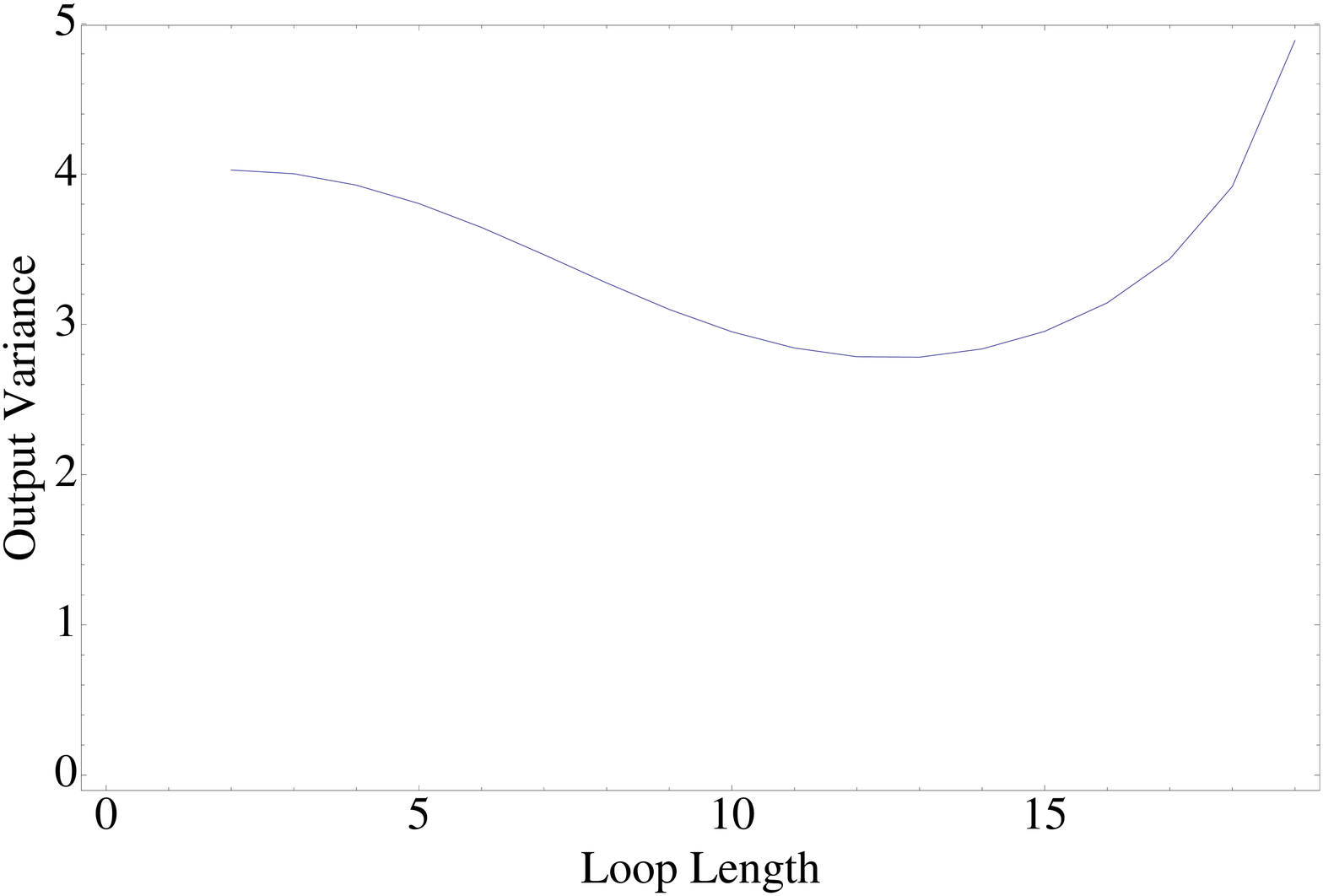}
}
\caption{Output variance of a linear pathway when the input is white noise, and we add a negative (left) or positive (right) feedforward loop starting from the first vertex. For the positive loop, the variance is largest when we connect nearby vertices (large correlation) or we connect an early vertex to the end of the pathway, since it has a large variance that is transmitted directly to the output without being further filtered. }
\label{FeedForwardStatisticalAnalysis}
\end{figure}

In the case of negative feedforward loop, the variance in the output increases as the loop length increases. 
When the feedforward interaction is positive, the variance decreases at first, since the correlation among the different states also decreases, but then goes up, partly because when it affects a node towards the end of the pathway, it does not pass through successive filters, so the variance does not have the chance to decrease (see Figure \ref{FeedForwardStatisticalAnalysis}).

\subsection{Delayed Feedforward and Feedback Cycles}

As one would expect, adding delay to the interactions among any nodes in a network driven by noise decreases their correlation, meaning that any feedforward or feedback cycles will have a smaller effect.
The covariance of a white noise process with a delayed version of the same signal can be computed the same way as in equation \eqref{OutputVarianceFromImpulseResponse}:
\begin{equation}
\begin{aligned}
\mathbb{V}_{\tau}[y] &=\lim _{t\rightarrow \infty} \mathbb{E}[y(t)y(t+\tau)] \\
&=\lim _{t\rightarrow \infty} \mathbb{E} \left[\left(\int _{-\infty}^{t} h (t-r) \Sigma_{r}dW_{r}\right )\left(\int _{-\infty}^{t+\tau} h (t+\tau-s) \Sigma _{s} dW_{s}\right )^{T} \right] \\
&=\lim _{t\rightarrow \infty} \int _{-\infty}^{t} \int _{-\infty}^{t+\tau} h (t-s) \Sigma_{r} \mathbb{E} \left[dW_{r}dW_{s}^{T} \right] \Sigma^{T}_{s} h^{T} (t+\tau-s)\\
&=\lim _{t\rightarrow \infty} \int _{-\infty}^{t} \int _{-\infty}^{t+\tau} h (t-r) \Sigma _{r}\sqrt{dr} \delta(s-r) \sqrt{ds} \Sigma _{s} h^{T} (t+\tau-s)\\
&=\lim _{t\rightarrow \infty} \int _{-\infty}^{t}h (t-s) V_{s} h^{T} (t+\tau-s) ds\\
&=\lim _{t\rightarrow \infty} \int _{0}^{t}h (t-s) V_{s} h^{T} (t+\tau-s) ds . \\
\end{aligned}
\end{equation}
If the system is causal, linear  and time invariant, and the disturbance is white noise of constant strength added to the input, 
\begin{equation}
\begin{aligned}
\mathbb{V}_{\tau}[y] &=\int _{0}^{\infty}h (u) V h^{T} (u+\tau) du.\\
\end{aligned}
\end{equation}
As a specific example, if the impulse response is $h(t)=Ce^{At}B$ and the covariance matrix is  constant:
\begin{equation}
\begin{aligned}
\mathbb{V}_{\tau}[y] &=\int _{0 }^{\infty} Ce^{As}B  V  B^{T} e^{(s+\tau)A^{T}}C^{T}  ds\\
&=C \left( \int _{-\infty}^{\infty}e^{As} B  V B^{T}e^{s A^{T}}ds \right) e^{\tau A^{T}} C^{T}. \\
\end{aligned}
\end{equation}
Note that the last equation is similar to equation \eqref{TimeDomainLTINoiseResponse}, except for the exponential delay term in the end.
We assume that the dynamical matrix $A$ has negative eigenvalues, otherwise the system is not stable. 
If the delay is  $\tau>0$,

\begin{equation}
\begin{aligned}
\norm{\mathbb{V}_{\tau}} & =\norm{C \left( \int _{0}^{\infty}e^{As} B V B^{T}e^{A^{T}s}ds \right) e^{A^{T}\tau} C^{T} } \\
& \leq \norm{C \left( \int _{0}^{\infty}e^{As} B V B^{T}e^{A^{T}s}ds \right) C^{T} } \cdot \norm{ e^{A^{T}\tau} }  \\
& \leq \norm{C \left( \int _{0}^{\infty}e^{As} B V B^{T}e^{A^{T}s}ds \right) C^{T} } \\
&=\norm{\mathbb{V}_{0} }.
\end{aligned}
\end{equation}

The matrix norm used here is the first order elementwise norm, since we are usually interested in the average variance of all parts of the network.
\begin{equation}
\norm{M}= \sum _{i=1}^{N} \sum _{j=1}^{N} |m_{i,j}|. 
\end{equation}

If we only know the autocorrelation function of the disturbance, we can compute the output variance by moving to the frequency domain.
\begin{equation}
\begin{aligned}
R_{y}(\tau) &=\int _{-\infty}^{+\infty}S_{y}(f) \cos(2\pi f\tau) df \\
&=\int _{-\infty}^{+\infty}S_{y}(f) \cos(2\pi f\tau) df \\
&=\int _{-\infty}^{+\infty} |H(f)|^{2} S_{x}(f) \cos(2\pi f\tau) df \\
&=\frac{1}{2\pi} \int _{-\infty}^{+\infty} |C(j\omega I -A)^{-1} B|^{2} \left( \int _{-\infty}^{+\infty}R_{x}(u) \cos(\omega u) du \right)  \cos(\omega\tau) d\omega . \\
\end{aligned}
\end{equation}
The shape of the autocorrelation function is a good indicator of how a feedback or feedforward loop will affect the output variation.
A correlation function that quickly goes to zero as $\tau$ increases shows that the feedback cycle will not change the variance of the output by a lot. 
Conversely, a random signal with a correlation structure can be easily filtered out by applying an appropriate feedback mechanism.

\FloatBarrier
\subsection{Minimization of the Average Vertex Variance}

In a general network, signals are propagated from one node to its neighbors.
Every vertex receives a filtered version of the noise signal, since  every node acts as a single pole filter.
The pole is always real, and proportional to the degree of each vertex, if we assume that each node receives input proportional to the differences of concentrations among its neighbors and itself, or that nodes that interact with many others have proportionally large degradation rates.
In this case, we can  model the dynamics of a first order linear network through its Laplacian matrix.
In such a network,  the state of each node $x_{k}$ follows the differential equation
\begin{equation}
\frac{dx_{k}}{dt}=\sum _{m\in \mathcal{N}_{k}} a_{km}(x_{k}-x_{m}),
\end{equation}
where $a_{km}>0$ for every $k,m \in \mathcal{V}$.
The Laplacian of a matrix has been used to model a wide range of systems, including formation stabilization for groups of agents, collision avoidance of swarms and synchronization of coupled oscillators \cite{ConsensusPaper}.
It can also be used in biological and chemical reaction networks, if the degradation rate of each species is equal to the sum of the rates with which it is produced.
In this section, we will model the dynamics of each network with its Laplacian matrix, where each node is affected by a noise source which is independent of all other nodes, but has the same standard deviation.
Given that each vertex contributes equally to the overall noise measure of the graph, and since the noise entering each node propagates towards all its neighbors, 
we can use Lemma \ref{SymmetricMultivariableOptimization} to see that the degrees of the network vertices have to be as similar as possible  (see also \cite{DegreeRealizability} and \cite{MinXtalkNetworks}).
In addition,  Figure \ref{AverageVariancePathwayVsCycle} shows that the cycles need to be as long as possible in order to avoid any correlations of signals through two different paths. 
For longer cycles, the noise inputs go through more filters  before they are combined.
Moreover, the phase shift is larger for all their frequencies, which reduces their correlation.
On the other hand, there are bounds on how long a cycle can be given the network's order and size.
Networks with long cycles tend to have large radius and larger average distance, as shown in \cite{ExtremalNetworks}, which makes noise harder to propagate, having to pass through many filters.
By the same token, networks with a small clustering coefficient will tend to be more immune to noise in their output, since these networks tend to create cliques or densely connected subnetworks \cite{MaxClusteringNetworks}, which will facilitate noise propagation, especially if the noise sources that affect the nodes are correlated, as shown in previous sections.
A method to find these graphs is first to determine their degree sequence, and then determine which one has the largest average cycle length.
This procedure can be simplified by working recursively, building networks with progressively larger order and size.

\begin{lemma}
There is always a connected graph of order $N$ and size $m$ in which there are $k$ vertices with degree $d+1$ and $N-k$ vertices with degree $d$ where
\begin{equation}
d=\left \lfloor \frac{2m}{N} \right \rfloor \quad \textrm{and} \quad k=2m-N d.
\end{equation}
\end{lemma}
\begin{proof}
We will prove the existence of such a graph by starting with its degree distribution and, by successive transformations, convert it to a graph that is known to exist.
Specifically, at each step we will remove one vertex along with its edges, repeating the process until we end up having a cycle graph.
Assume that the degree sequence of the graph $\mathcal{G}_{0}$ is as above, and we arrange the degrees of the vertices in a decreasing order.
\begin{equation}
s_{0}=\{\underbrace{d+1,d+1,\ldots d+1}_{\parbox{2cm}{\centering $k$ \\ vertices}},\underbrace{d,d,\ldots , d}_{\parbox{2cm}{\centering $N-k$ vertices}} \}.
\end{equation}
According to the Havel-Hakimi theorem \cite{DegreeRealizability}, the above sequence is a graph sequence if and only if the graph sequence in which the largest degree vertex is connected to vertices  $2,3,...,d+2$ is also a graph sequence.
The new graph will have a degree sequence of 
\begin{equation}
s_{1}= \left\{ \begin{array}{ll}
\{\underbrace{d+1,d+1,\ldots d+1}_{\parbox{2cm}{\centering $k-d-2$ \\ vertices}},\underbrace{d,d,\ldots d}_{\parbox{2cm}{\centering $N-k+d+1$ vertices}} \} & \textrm{if $d<k-2$}\\
\{\underbrace{d,d,\ldots ,d, d}_{\parbox{3.5cm}{\centering $N+k-d-3$\\ vertices}},\underbrace{d-1,d-1,\ldots , d-1}_{\parbox{2cm}{\centering $d-k+2$ \\ vertices}} \}& \textrm{if $d\geq k-2$}.
\end{array} \right.
\end{equation}
The key observation is that the transformation above preserves the property of degree homogeneity, in other words, in the new graph $\mathcal{G}_{1}=\mathcal{G}_{1}(N-1,m-d+1)$, the minimum and maximum vertex degrees are 
\begin{equation}
d_{min}=\left \lfloor \frac{m-d+1}{N-1} \right \rfloor
\end{equation}
and
\begin{equation}
d_{min}\leq d_{max}\leq d_{min}+1.
\end{equation}
Repeating the process,  there will be a graph $\mathcal{G}_{r}$ with at least one vertex of degree $d_{min}=1$.
It follows from the analysis above that the graph $\mathcal{G}_{r}$ will include either one or two vertices of degree $d_{min}=1$.
If it has two vertices with degree one, it is the path graph.
If it has only one vertex  with degree one, its degree sequence is not a graph sequence.
But this would mean that the sum of all the degrees is an odd number, which is not possible, since at every transformation, we remove $2d_{max}$ from the sum of degrees.
The graph $\mathcal{G}_{r}$ is a connected graph, and implementing the inverse transforms, we connect new vertices to an already connected network, which guarantees that the final graph is connected.
\end{proof}

For networks with a small number of vertices , we can find all graphs with the desired degree sequence, and among them, exhaustively search for the ones with the largest average cycle length that have the smallest average variance.
For $N=6$ nodes, all connected networks (with $5\leq m \leq 15$ edges) with most homogeneous degree distribution and longest average cycles are shown in Figure \ref{All6MinAverageNoiseVariance}.

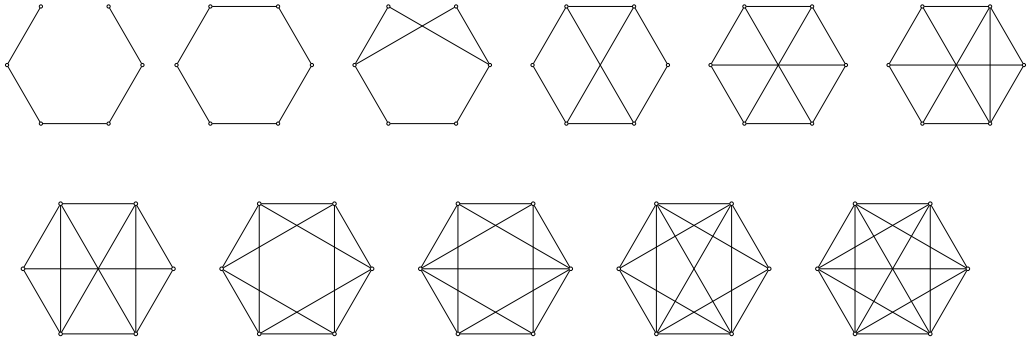
\begin{figure}[htb]
\subfigure{
\psscalebox{0.18}{
\begin{pspicture}(-5,-5)(5,6)
{
\cnodeput(2.5,4.33){A}{}
\cnodeput(-2.5,4.33){B}{}
\cnodeput(-5,0){C}{}
\cnodeput(-2.5,-4.33){D}{}
\cnodeput(2.5,-4.33){E}{}
\cnodeput(5,0){F}{}
}
\ncline{-}{D}{E}
\ncline{-}{B}{C}
\ncline{-}{C}{D}
\ncline{-}{F}{E}
\ncline{-}{A}{F}
\end{pspicture}
}
}
\subfigure{
\psscalebox{0.18}{
\begin{pspicture}(-5,-5)(5,8)
{
\cnodeput(2.5,4.33){A}{}
\cnodeput(-2.5,4.33){B}{}
\cnodeput(-5,0){C}{}
\cnodeput(-2.5,-4.33){D}{}
\cnodeput(2.5,-4.33){E}{}
\cnodeput(5,0){F}{}
}
\ncline{-}{A}{B}
\ncline{-}{B}{C}
\ncline{-}{C}{D}
\ncline{-}{D}{E}
\ncline{-}{E}{F}
\ncline{-}{F}{A}
\end{pspicture}
}
}
\subfigure{
\psscalebox{0.18}{
\begin{pspicture}(-5,-5)(5,8)
{
\cnodeput(2.5,4.33){A}{}
\cnodeput(-2.5,4.33){B}{}
\cnodeput(-5,0){C}{}
\cnodeput(-2.5,-4.33){D}{}
\cnodeput(2.5,-4.33){E}{}
\cnodeput(5,0){F}{}
}
\ncline{-}{A}{C}
\ncline{-}{A}{F}
\ncline{-}{B}{C}
\ncline{-}{B}{F}
\ncline{-}{C}{D}
\ncline{-}{D}{E}
\ncline{-}{E}{F}
\end{pspicture}
}
}
\subfigure{
\psscalebox{0.18}{
\begin{pspicture}(-5,-5)(5,8)
{
\cnodeput(2.5,4.33){A}{}
\cnodeput(-2.5,4.33){B}{}
\cnodeput(-5,0){C}{}
\cnodeput(-2.5,-4.33){D}{}
\cnodeput(2.5,-4.33){E}{}
\cnodeput(5,0){F}{}
}
\ncline{-}{C}{B}
\ncline{-}{C}{D}
\ncline{-}{A}{F}
\ncline{-}{E}{F}
\ncline{-}{A}{B}
\ncline{-}{D}{E}
\ncline{-}{A}{D}
\ncline{-}{B}{E}
\end{pspicture}
}
}
\subfigure{
\psscalebox{0.18}{
\begin{pspicture}(-5,-5)(5,8)
{
\cnodeput(2.5,4.33){A}{}
\cnodeput(-2.5,4.33){B}{}
\cnodeput(-5,0){C}{}
\cnodeput(-2.5,-4.33){D}{}
\cnodeput(2.5,-4.33){E}{}
\cnodeput(5,0){F}{}
}
\ncline{-}{C}{B}
\ncline{-}{C}{D}
\ncline{-}{A}{F}
\ncline{-}{E}{F}
\ncline{-}{A}{B}
\ncline{-}{D}{E}
\ncline{-}{A}{D}
\ncline{-}{B}{E}
\ncline{-}{C}{F}
\end{pspicture}
}
}
\subfigure{
\psscalebox{0.18}{
\begin{pspicture}(-5,-5)(5,8)
{
\cnodeput(2.5,4.33){A}{}
\cnodeput(-2.5,4.33){B}{}
\cnodeput(-5,0){C}{}
\cnodeput(-2.5,-4.33){D}{}
\cnodeput(2.5,-4.33){E}{}
\cnodeput(5,0){F}{}
}
\ncline{-}{C}{B}
\ncline{-}{C}{D}
\ncline{-}{A}{F}
\ncline{-}{E}{F}
\ncline{-}{A}{B}
\ncline{-}{D}{E}
\ncline{-}{A}{D}
\ncline{-}{B}{E}
\ncline{-}{C}{F}
\ncline{-}{A}{E}
\end{pspicture}
}
}
\subfigure{
\psscalebox{0.2}{
\begin{pspicture}(-5.5,-5)(5.5,8)
{
\cnodeput(2.5,4.33){A}{}
\cnodeput(-2.5,4.33){B}{}
\cnodeput(-5,0){C}{}
\cnodeput(-2.5,-4.33){D}{}
\cnodeput(2.5,-4.33){E}{}
\cnodeput(5,0){F}{}
}
\ncline{-}{C}{B}
\ncline{-}{C}{D}
\ncline{-}{A}{F}
\ncline{-}{E}{F}
\ncline{-}{A}{B}
\ncline{-}{D}{E}
\ncline{-}{A}{D}
\ncline{-}{B}{E}
\ncline{-}{C}{F}
\ncline{-}{A}{E}
\ncline{-}{B}{D}
\end{pspicture}
}
}
\subfigure{
\psscalebox{0.2}{
\begin{pspicture}(-5.5,-5)(5.5,8)
{
\cnodeput(2.5,4.33){A}{}
\cnodeput(-2.5,4.33){B}{}
\cnodeput(-5,0){C}{}
\cnodeput(-2.5,-4.33){D}{}
\cnodeput(2.5,-4.33){E}{}
\cnodeput(5,0){F}{}
}
\ncline{-}{A}{B}
\ncline{-}{A}{C}
\ncline{-}{A}{E}
\ncline{-}{A}{F}
\ncline{-}{B}{C}
\ncline{-}{B}{D}
\ncline{-}{B}{F}
\ncline{-}{C}{D}
\ncline{-}{C}{E}
\ncline{-}{D}{E}
\ncline{-}{D}{F}
\ncline{-}{E}{F}
\end{pspicture}
}
}
\subfigure{
\psscalebox{0.2}{
\begin{pspicture}(-5.5,-5)(5.5,8)
{
\cnodeput(2.5,4.33){A}{}
\cnodeput(-2.5,4.33){B}{}
\cnodeput(-5,0){C}{}
\cnodeput(-2.5,-4.33){D}{}
\cnodeput(2.5,-4.33){E}{}
\cnodeput(5,0){F}{}
}
\ncline{-}{A}{B}
\ncline{-}{A}{C}
\ncline{-}{A}{E}
\ncline{-}{A}{F}
\ncline{-}{B}{C}
\ncline{-}{B}{D}
\ncline{-}{B}{F}
\ncline{-}{C}{D}
\ncline{-}{C}{E}
\ncline{-}{D}{E}
\ncline{-}{D}{F}
\ncline{-}{E}{F}
\ncline{-}{C}{F}
\end{pspicture}
}
}
\subfigure{
\psscalebox{0.2}{
\begin{pspicture}(-5.5,-5)(5.5,8)
{
\cnodeput(2.5,4.33){A}{}
\cnodeput(-2.5,4.33){B}{}
\cnodeput(-5,0){C}{}
\cnodeput(-2.5,-4.33){D}{}
\cnodeput(2.5,-4.33){E}{}
\cnodeput(5,0){F}{}
}
\ncline{-}{A}{B}
\ncline{-}{A}{C}
\ncline{-}{A}{D}
\ncline{-}{A}{E}
\ncline{-}{A}{F}
\ncline{-}{B}{C}
\ncline{-}{B}{D}
\ncline{-}{C}{D}
\ncline{-}{B}{E}
\ncline{-}{C}{E}
\ncline{-}{D}{E}
\ncline{-}{B}{F}
\ncline{-}{E}{F}
\ncline{-}{D}{F}
\end{pspicture}
}
}
\subfigure{
\psscalebox{0.2}{
\begin{pspicture}(-5.5,-5)(5.5,8)
{
\cnodeput(2.5,4.33){A}{}
\cnodeput(-2.5,4.33){B}{}
\cnodeput(-5,0){C}{}
\cnodeput(-2.5,-4.33){D}{}
\cnodeput(2.5,-4.33){E}{}
\cnodeput(5,0){F}{}
}
\ncline{-}{A}{B}
\ncline{-}{A}{C}
\ncline{-}{A}{D}
\ncline{-}{A}{E}
\ncline{-}{A}{F}
\ncline{-}{B}{C}
\ncline{-}{B}{D}
\ncline{-}{C}{D}
\ncline{-}{B}{E}
\ncline{-}{C}{E}
\ncline{-}{D}{E}
\ncline{-}{B}{F}
\ncline{-}{C}{F}
\ncline{-}{D}{F}
\ncline{-}{E}{F}
\end{pspicture}
}
}
\caption{All connected networks of order $N=6$ and size $5\leq m \leq 15$ and with minimum output variance. We assume that every vertex is affected by an independent noise source. In addition, each vertex acts as a single pole filter. The total noise of the network is measured as the average of the variances of all nodes.}
\label{All6MinAverageNoiseVariance}
\end{figure}

To summarize this section, positive correlations increase the output variance, and cycles create correlations that make the system more prone to random inputs.
The longer the cycles, the smaller their effect.
The immunity to noise is increased when pathways with the same output introduce different phase shifts, so that the different noise contributions cancel each other at least partially.
This result holds both for feedforward and feedback loops.
When we have some convex constraint on the strength of the various filters, placing the poles, we can find the optimal placement such that the output noise is reduced.
Specifically, for a linear network where all nodes act as single pole filters and the dynamics of the network are described by its Laplacian matrix, there is a systematic way to find the network with the smallest average variance. 
The optimal networks have homogeneous degree distribution, and cycles that are as long as possible.

\FloatBarrier
\section{Crosstalk Reduces Noise In Pathway Outputs}

\subsection{Motivating Example}

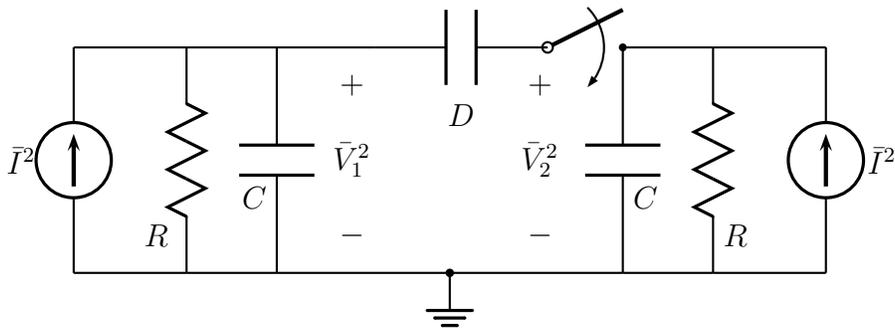
\begin{figure}[htbp]
\centering
\psscalebox{1.0}{
\begin{pspicture}(0,0)(10,4) 
\pnode(0,1){A} 
\pnode(0,4){B} 
\pnode(1.5,1){C} 
\pnode(1.5,4){D}  
\pnode(10,1){F}
\pnode(10,4){G} 
\pnode(8.5,1){H}
\pnode(8.5,4){I} 
\pnode(5,1){J} 
\pnode(2.7,1){K} 
\pnode(7.3,1){L} 
\pnode(2.7,4){M} 
\pnode(7.3,4){N} 
\pnode(4,4){P} 
\pnode(6.3,4){Q} 
\Ucc[labelInside=1](A)(B){$\bar{I}^{2}$} 
\Ucc[labelInside=1](F)(G){$\qquad \qquad \qquad \quad \bar{I}^{2}$} 
\resistor[dipolestyle=zigzag](C)(D){}
\capacitor[parallel,parallelarm=1.2](D)(C){}
\resistor[dipolestyle=zigzag](H)(I){}
\capacitor[parallel,parallelarm=1.2](H)(I){}
\switch(Q)(N){}
\capacitor[](P)(Q){}
\newground(J)
\rput(3.7,1.5){$-$}
\rput(3.7,2.5){$\bar{V}_{1}^{2}$}
\rput(3.7,3.5){$+$}
\rput(6.2,1.5){$-$}
\rput(6.2,2.5){$\bar{V}_{2}^{2}$}
\rput(6.2,3.5){$+$}
\rput(1.1,1.5){$R$}
\rput(8.8,1.5){$R$}
\rput(2.4,2){$C$}
\rput(7.6,2){$C$}
\rput(5.15,3.1){$D$}
\wire(A)(C)
\wire(B)(D)
\wire(I)(G)
\wire(H)(F)
\wire(K)(L)
\wire(M)(P)
\end{pspicture}
}
\caption{A simple circuit with two noise sources. The two resistors generate thermal noise, which is modeled as current sources in parallel to them. When the switch is open, the two circuits are independent. When the switch is closed, the noise in both outputs has smaller variance than before. }
\label{NoisyRC}
\end{figure}

Assume that we have a resistor without any external voltage source.
If we measure the voltage between its endpoints, we will find that in any infinitesimal frequency interval $df$ there is thermal noise $V_{t}$ with
\begin{equation}
\mathbb{E}\left [V_{t} \right ]=0 \qquad \textrm{and} \qquad \mathbb{E}\left [V^{2}_{t} \right ]=4kTR df
\end{equation}
where $R$ is the resistance.
The above equation shows that the noise increases as temperature and resistance increase.
We connect a capacitor in parallel with the resistor, and measure the voltage between its endpoints.
We are interested in the total amount of variance of the voltage in the output of the parallel combination of the resistor and the capacitor.
When the switch is open, each of the two subcircuits operate independently, and the output variance for both of them is
\begin{equation}
\begin{aligned}
\bar{V}_{1}^{2}=\bar{V}_{2}^{2} &=\int _{0}^{+\infty} \frac{4kT}{R} \left| \frac{R}{1+j 2\pi f RC}\right | ^{2} df\\
&=\int _{0}^{+\infty} \frac{4kT}{R} \frac{R^{2}}{1+ (RC)^{2}(2\pi f)^{2} } df \\
&=\frac{4kTR}{2\pi RC} \int _{0}^{+\infty} \frac{du}{1+ u^{2}}\\
&=\frac{kT}{C}.
\end{aligned}
\end{equation}

If we close the switch, the output variance is 
\begin{equation}
\begin{aligned}
\bar{V}_{1}^{2}=\bar{V}_{2}^{2} &=\frac{4kT}{R} \int _{0}^{+\infty}\left|\frac{R(1+j2\pi f R (C+D))}{(j2\pi f R(C+2D)+1) (1+j2\pi f R C)}\right | ^{2} df\\
&\qquad +\frac{4kT}{R} \int _{0}^{+\infty} \left|\frac{j2\pi f R^{2}D}{(1+j2\pi f RC) (1+j2\pi f R (C+2D)}\right | ^{2}  df\\
&=\frac{kTD^{2}}{2C(C+D)(C+2D)}+\frac{kT(C+D)}{C(C+2D)}\\
&=\frac{kT}{C}\cdot \frac{C+D}{C+2D}. \\
\end{aligned}
\end{equation}
If the capacitor that connects the two subcircuits has capacitance $D>0$ and the two noise sources are uncorrelated, then both outputs have smaller variances.

In biology, there are countless sources of noise, and the noise is often larger than the signal itself.
It is possible that the cell needs to employ the same technique for reducing noise, distributing it among many different and unrelated components.
Crosstalk between different elements of a biological network couples the behavior of different parts of the network, introducing more poles in the network dynamics, as we will see next.
This is equivalent to introducing capacitances between random parts of an electrical network.
The new system filters out noise much more effectively, but on the other hand may be slower to react to various inputs, so there seems to be a tradeoff between how fast a network can respond to changes and how well it filters out noise.
The next section studies the effect of crosstalk on the behavior of a small network.

\subsection{Crosstalk on Single Nodes}

We analyze the four simple subgraphs of Figure \ref{SingleNodeXtalkInteractions}.
\begin{center}
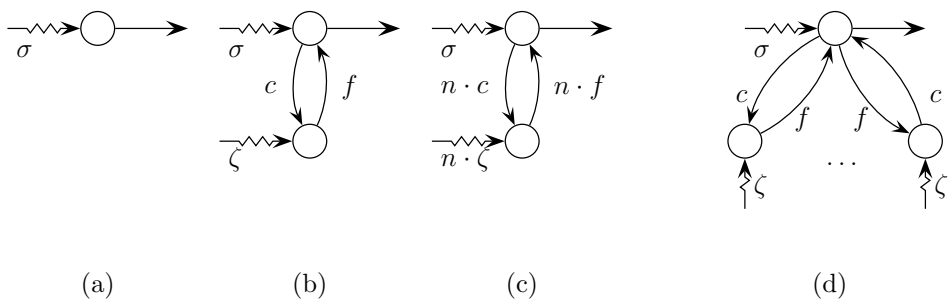
\begin{figure}[!hb]
\subfigure[]{
\psscalebox{0.6}{
\begin{pspicture}(-2,-4)(2,2)
{
\rput[l](-2,1){\rnode{InputA}{}}
\rput[l](-1.8,0.5){\rnode{NoiseStrength1}{\Large$\sigma$}}
\cnodeput[](0,1){A}{\strut }
\rput[l](2,1){\rnode{OutputA}{}}
}
\nczigzag[coilwidth=0.2cm,coilheight=1.5,coilarm=0.4cm,arrowsize=9pt]{->}{InputA}{A}
\ncline[arrowsize=10pt]{->}{A}{OutputA}

\end{pspicture}
}
\label{NoXtalkNode}
}
\subfigure[]{
\psscalebox{0.6}{
\begin{pspicture}(-2,-4)(2,2)
{
\rput[l](-2,1){\rnode{InputA}{}}
\rput[l](-1.8,0.5){\rnode{NoiseStrength1}{\Large$\sigma$}}
\cnodeput[](0,1){A}{\strut }
\rput[l](2,1){\rnode{OutputA}{}}
\rput[l](-2,-1.5){\rnode{InputB}{}}
\rput[l](-1.8,-1.9){\rnode{NoiseStrength2}{\Large$\zeta$}}
\cnodeput[](0,-1.5){B}{\strut }
\rput[l](-1,-0.3){\rnode{OutRate}{\Large $c$}}
\rput[l](0.7,-0.3){\rnode{InRate}{\Large$f$}}
}
\nczigzag[coilwidth=0.2cm,coilheight=1.5,coilarm=0.4cm,arrowsize=9pt]{->}{InputA}{A}
\ncline[arrowsize=10pt]{->}{A}{OutputA}
\nczigzag[coilwidth=0.2cm,coilheight=1.5,coilarm=0.4cm,arrowsize=9pt]{->}{InputB}{B}
\ncarc[arrowsize=8pt,arcangle=25]{<-}{A}{B}
\ncarc[arrowsize=8pt,arcangle=25]{<-}{B}{A}
\end{pspicture}
}
\label{SingleXtalkNode}
}
\subfigure[]{
\psscalebox{0.6}{
\begin{pspicture}(-2,-4)(2,2)
{
\rput[l](-2,1){\rnode{InputA}{}}
\rput[l](-1.8,0.5){\rnode{NoiseStrength1}{\Large$\sigma$}}
\rput[l](-1.8,-1.9){\rnode{NoiseStrength2}{\Large$n\cdot \zeta$}}
\cnodeput[](0,1){A}{\strut }
\rput[l](2,1){\rnode{OutputA}{}}
\rput[l](-2,-1.5){\rnode{InputB}{}}
\cnodeput[](0,-1.5){B}{\strut }
\rput[l](-1.8,-0.3){\rnode{OutRate}{\Large $n\cdot c$}}
\rput[l](0.7,-0.3){\rnode{InRate}{\Large$n \cdot f$}}
}
\nczigzag[coilwidth=0.2cm,coilheight=1.5,coilarm=0.4cm,arrowsize=9pt]{->}{InputA}{A}
\ncline[arrowsize=10pt]{->}{A}{OutputA}
\nczigzag[coilwidth=0.2cm,coilheight=1.5,coilarm=0.4cm,arrowsize=9pt]{->}{InputB}{B}
\ncarc[arrowsize=8pt,arcangle=25]{<-}{A}{B}
\ncarc[arrowsize=8pt,arcangle=25]{<-}{B}{A}
\end{pspicture}
}
\label{MultipleXtalkSingleNode}
}
\subfigure[]{
\psscalebox{0.6}{
\begin{pspicture}(-4,-4)(4,2)
{
\rput[l](-2,1){\rnode{InputA}{}}
\rput[l](-1.8,0.5){\rnode{NoiseStrength1}{\Large$\sigma$}}
\rput[l](-1.8,-2.5){\rnode{NoiseStrength2}{\Large$\zeta$}}
\rput[l](2.2,-2.5){\rnode{NoiseStrength3}{\Large$\zeta$}}
\rput[l](-0.2,-2){\rnode{NodeDots}{\Large$\dots$}}
\cnodeput[](0,1){A}{\strut }
\rput[l](2,1){\rnode{OutputA}{}}
\cnodeput[](-2,-1.5){B}{\strut }
\cnodeput[](2,-1.5){C}{\strut }
\rput[l](-2,-3){\rnode{InputB}{}}
\rput[l](2,-3){\rnode{InputC}{}}
\rput[l](-2.2,-0.5){\rnode{OutRate1}{\Large $c$}}
\rput[l](-0.9,-1){\rnode{InRate1}{\Large$f$}}
\rput[l](2.1,-0.5){\rnode{OutRate2}{\Large $c$}}
\rput[l](0.4,-1){\rnode{InRate2}{\Large$f$}}
}
\nczigzag[coilwidth=0.2cm,coilheight=1.5,coilarm=0.4cm,arrowsize=9pt]{->}{InputA}{A}
\ncline[arrowsize=10pt]{->}{A}{OutputA}
\nczigzag[coilwidth=0.2cm,coilheight=1.5,coilarm=0.4cm,arrowsize=9pt]{->}{InputB}{B}
\nczigzag[coilwidth=0.2cm,coilheight=1.5,coilarm=0.4cm,arrowsize=9pt]{->}{InputC}{C}
\ncarc[arrowsize=8pt,arcangle=25]{<-}{A}{B}
\ncarc[arrowsize=8pt,arcangle=25]{<-}{B}{A}
\ncarc[arrowsize=8pt,arcangle=25]{<-}{A}{C}
\ncarc[arrowsize=8pt,arcangle=25]{<-}{C}{A}
\end{pspicture}
}
\label{UnitXtalkMultipleNodes}
}
\caption{Crosstalk topologies involving one network node. \textbf{(a)} A node without crosstalk interactions with white noise input having standard deviation equal to $\sigma$. \textbf{(b)} A node with crosstalk interaction with one other node in the network, which also is affected by noise with standard deviation $\zeta$. \textbf{(c)} Same as before, but we assume that both the crosstalk and the noise are increased. \textbf{(d)} Crosstalk interactions with many other nodes, each of which has an independent noise input of the same strength. See text for quantitative analysis of these subsystems.}
\label{SingleNodeXtalkInteractions}
\end{figure}
\end{center}
For simplicity, we may disregard any deterministic inputs, since we assume these are linear systems, and any deterministic inputs only affect the output mean, but not its variance.
The stochastic differential equations for all the systems are shown next.

System $(a)$ obeys a simple stochastic differential equation, with one noise input, and it has no other interactions with any other parts of the network.
\begin{equation}
dX=-aXdt+\sigma dW_{t}.
\end{equation}
We have found the solution to this equation in the first section, and the variance in the output is found is equal to
\begin{equation}
V_{a}=\frac{\sigma^{2}}{2a}.
\end{equation}
This is the trivial case without any crosstalk, and will be used for comparison to the performance of the other subnetworks.

Subsystem $(b)$ consists of one vertex that interacts with another node which may also be prone to other noise sources.
Crosstalk is modeled through a new vertex in the network, with which the studied node exchanges flows.
In chemical reaction networks for example, the species of interest $X$ may be forming a complex $Y$ with species $I$, whose concentration is supposed to be constant:

\begin{equation}
X+I \xrightleftharpoons[f]{c} Y.
\end{equation}
We also expect $X$ to have a constant degradation rate $a$.
The equations for the concentrations of $X$ and $Y$ are 
 
\begin{equation}
\begin{aligned}
dX&=-(a+c)Xdt+fYdt+\sigma dW_{t} \\
dY&=cXdt-fYdt+\zeta dU_{t} \\
\end{aligned}
\label{SingleXtalkSingleNode}
\end{equation}

and the output variation is 
\begin{equation}
V_{b}=\frac{a+f}{2a(a+c+f)}\sigma^{2}+\frac{f}{2a(a+c+f)}\zeta^{2}.
\label{SingleXtalkSingleNodeVariance}
\end{equation}

\begin{figure}[htb]
\centering
\subfigure[]{
\centering
\includegraphics[scale=0.17]{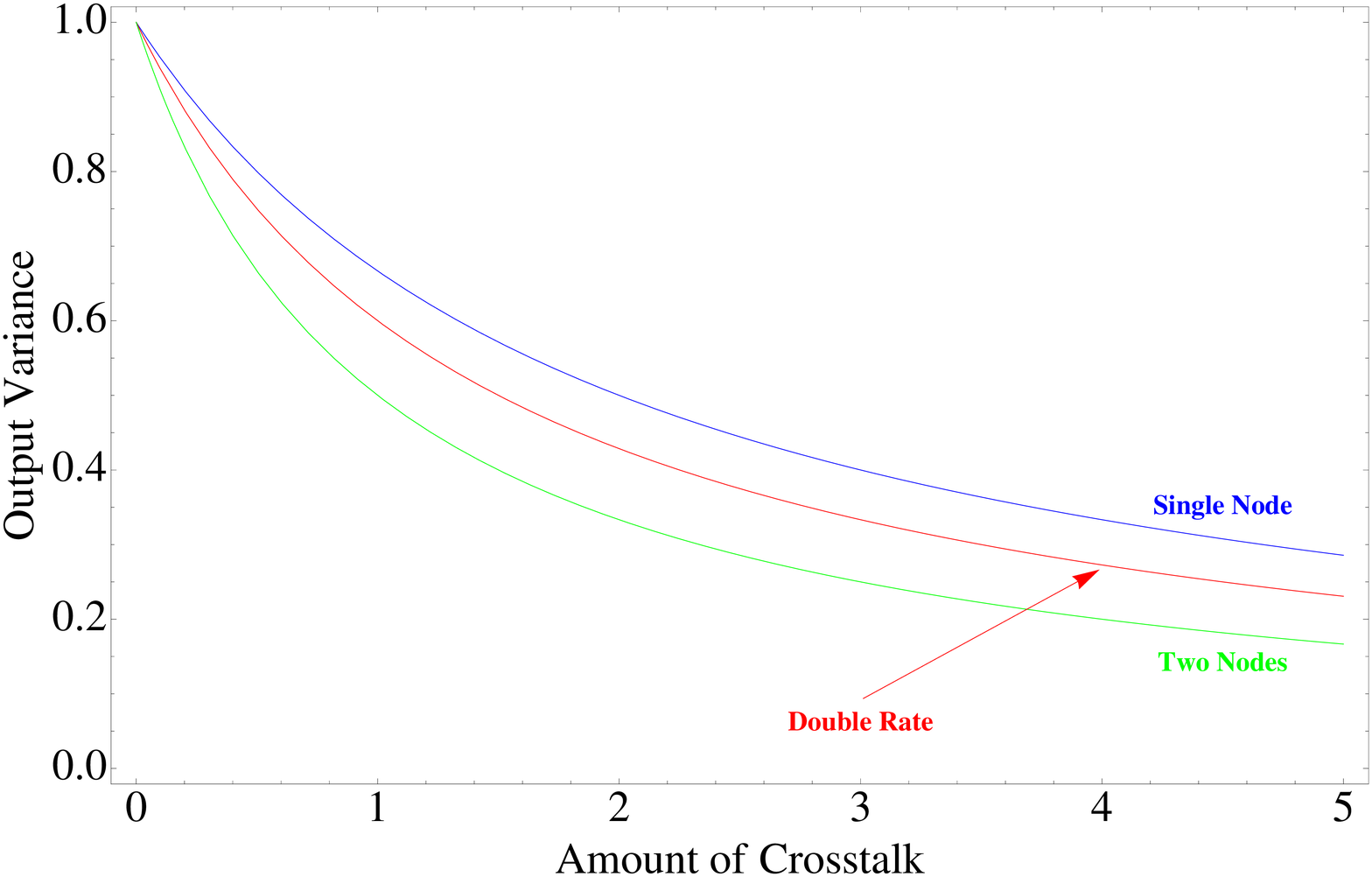}
}
\subfigure[]{
\centering
\includegraphics[scale=0.17]{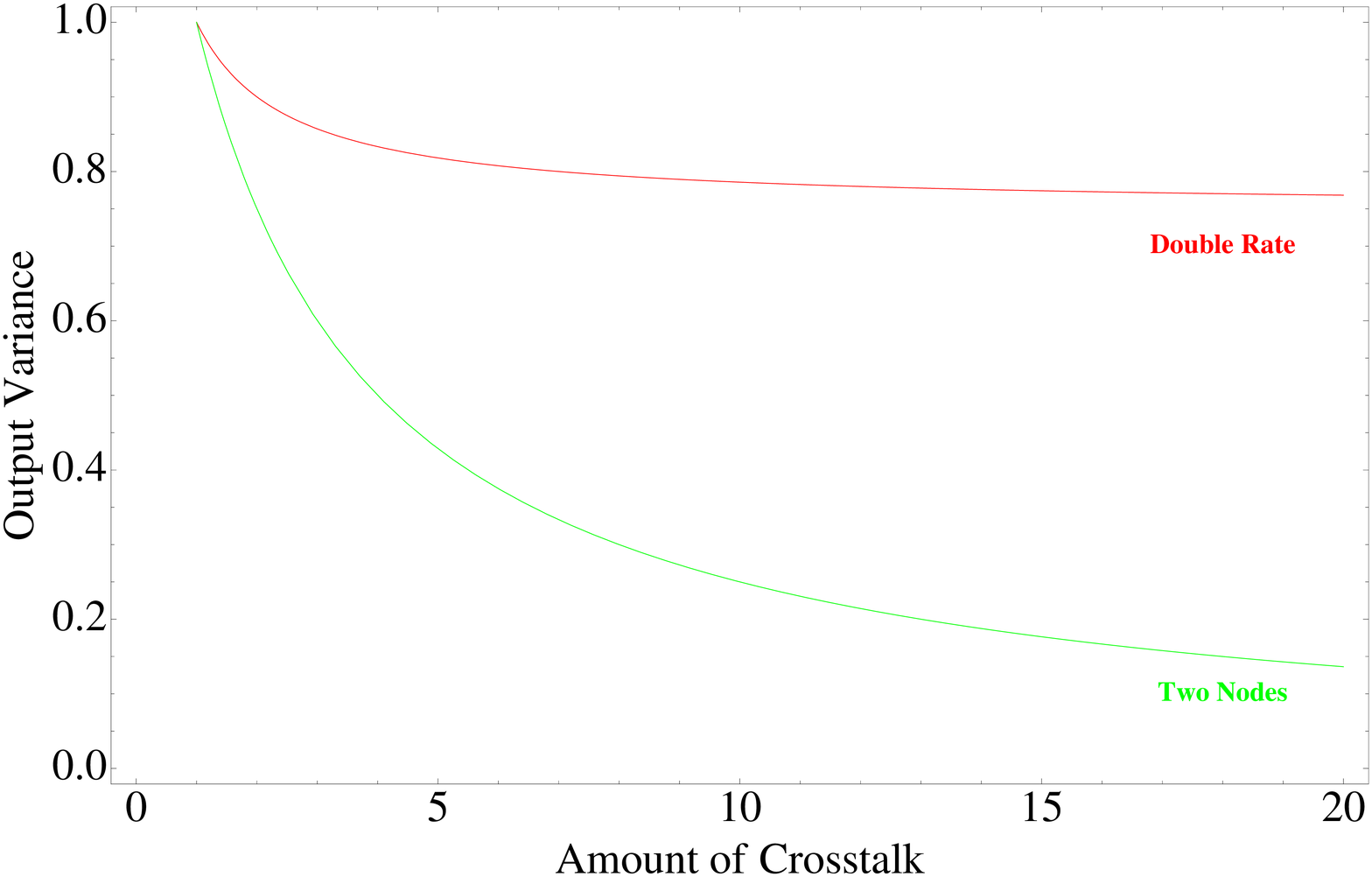}
}
\caption{Output variance as a result of noise input for a single vertex in the network in the existence of crosstalk interactions with other vertices. \textbf{(a)} Output variance as a function of the amount of crosstalk (concentration of crosstalk complex), when no additional noise is introduced. Crosstalk clearly mitigates the output variance. Also, having crosstalk with two independent nodes reduces the variance even more, compared to having a single crosstalk node. \textbf{(b)} Normalized output variance as a fraction of the variance when only one crosstalk node is present. Having many small sources of crosstalk is clearly better than having one strong crosstalk interaction. For the same amount of total crosstalk, dividing it among many nodes drives the output noise variance to zero as the number of nodes grows large.}
\label{SingleNodeXtalkComparisonNoiselessNodes}
\end{figure}

The next step is to see what happens if we increase the crosstalk intensity.
We can distinguish two cases.
The first is when there is crosstalk with one other node (Figure \ref{MultipleXtalkSingleNode}).
In the chemical reaction network analogy, 
\begin{equation}
X+A \xrightleftharpoons[n\cdot f]{n\cdot c} Y.
\end{equation}
It is straightforward to find the new differential equations, and the variance in the output.
\begin{equation}
\begin{aligned}
dX&=-(a+nc)Xdt+nfYdt+\sigma dW_{t} \\
dY&=ncXdt-nfYdt+ n \zeta dU_{t} \\
\end{aligned}
\end{equation}
\begin{equation}
V_{c}=\frac{a+nf}{2a(a+n(c+f))}\sigma^{2}+\frac{n^{3}f}{2a(a+n(c+f))}\zeta^{2}.
\label{MultipleXtalkSingleNode}
\end{equation}
Finally, we consider the case where one node has crosstalk interactions with many different nodes, each of which is affected by a different noise process (Figure \ref{UnitXtalkMultipleNodes}).
The equations that the nodes obey are
\begin{equation}
\begin{aligned}
dX&=-(a+nc)Xdt+nfYdt+\sigma dW_{t} \\
dY_{k}&=cXdt-fY_{k}dt+ \zeta dU_{t}^{k} \qquad 1\leq k \leq n  \\
\end{aligned}
\end{equation}
and the output variance can be computed as
\begin{equation}
V_{d}=\frac{a+f}{2a(a+n c+f)}\sigma^{2}+\frac{nf}{2a(a+nc+f)}\zeta^{2}.
\label{UnitXtalkMultipleNodesVariance}
\end{equation}
When no noise is introduced from the crosstalk nodes ($\zeta=0$), crosstalk reduces the output variance.
Figure \ref{SingleNodeXtalkComparisonNoiselessNodes} compares the last three cases, as the strength of crosstalk interactions among the nodes increases.
The crosstalk strength in this case is quantified by the ratio 
\begin{equation}
r_{x}=\frac{c}{c+f}
\end{equation}
which is equal to the concentration of the crosstalk product $Y$ in equation \eqref{SingleXtalkSingleNode} in the absence of degradation rates and noise inputs.
It is shown that distributing the crosstalk among many nodes (equation \eqref{UnitXtalkMultipleNodesVariance}) decreases the effect of noise noticeably more compared to the single node case.
This is even more pronounced when we normalize by the variance in the base case (equation \eqref{SingleXtalkSingleNodeVariance}).

\begin{figure}[htb]
\centering
\includegraphics[scale=0.33]{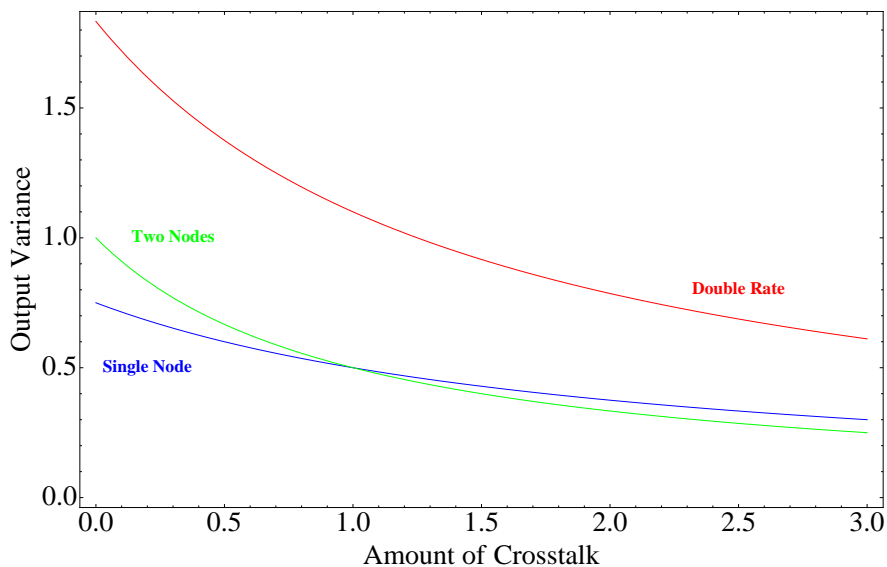}
\caption{Normalized variance of the output when the crosstalk introduces additional noise.  Having strong crosstalk interactions with one single node increases the variance because noise propagates easily. When crosstalk is distributed among many nodes, the variance may be smaller or larger than before, depending on the strength of the interactions. This is because having crosstalk interactions with many other vertices introduces a proportional amount of noise.
}
\label{SingleNodeXtalkComparisonNoisyNodes}
\end{figure}

When crosstalk introduces additional noise, it may increase the variance in the output of any given node if crosstalk is not strong enough to make up for the introduced noise (Figure \ref{SingleNodeXtalkComparisonNoisyNodes}).

\FloatBarrier

\subsection{Parallel Pathways}

We consider two pathways with crosstalk among more than one of their nodes.
We distinguish two cases, when the two pathways have different or the same outputs.
In the first case, since the two outputs are independent, it is easier to reduce the noise variance in both of them, by ``exchanging'' their noise through each node, assuming that the different noise sources are independent.
When the output is the same, there is little reduction in the output variance from crosstalk, since every disturbance eventually reaches the output, and is combined with other correlated versions of the same signal, as shown in Figure \ref{ParallelPathwaysXtalkNoiseReduction}.
The variance reduction in this case is caused by the increase of the effective pathway lengths, since they follow on average a longer path towards the output.

\begin{center}
\begin{figure}[htb]
\begin{center}
 \begin{minipage}[t]{0.5\linewidth}
 \centering
 \subfigure{
\centering
\psscalebox{0.45}{
\begin{pspicture}(-1,-3)(13,2.5)
\rput[l](-1,2){\rnode{Noise1}{\Huge{Noise 1}}}
\pnode(0,1){I1}
\rput[l](-1,-2){\rnode{Noise2}{\Huge{Noise 2}}}
\pnode(0,-1){I2}
\cnodeput(2,1){A1}{\strut}
\cnodeput(4,1){B1}{\strut}
\rput[l](5.5,1){\rnode{EtcF1}{}}
\rput[l](5.6,1){\rnode{Etc1}{\Huge{...}}}
\rput[l](6.3,1){\rnode{EtcS1}{}}
\cnodeput(8,1){C1}{\strut}
\cnodeput(10,1){D1}{\strut}
\cnodeput(2,-1){A2}{\strut}
\cnodeput(4,-1){B2}{\strut}
\rput[l](5.5,-1){\rnode{EtcF2}{}}
\rput[l](5.6,-1){\rnode{Etc2}{\Huge{...}}}
\rput[l](6.3,-1){\rnode{EtcS2}{}}
\cnodeput(8,-1){C2}{\strut}
\cnodeput(10,-1){D2}{\strut}
\pnode(12,1){T1}
\rput[l](10,2){\rnode{Noise1}{\Huge{Output 1}}}
\pnode(12,-1){T2}
\rput[l](10,-2){\rnode{Noise1}{\Huge{Output 2}}}
\nczigzag[coilwidth=0.2cm,coilheight=1.5,coilarm=0.4cm,arrowsize=9pt]{->}{I1}{A1}
\ncline[arrowsize=8pt]{->}{A1}{B1}
\ncline[arrowsize=8pt]{->}{B1}{EtcF1}
\ncline[arrowsize=8pt]{->}{EtcS1}{C1}
\ncline[arrowsize=8pt]{->}{C1}{D1}
\ncline[arrowsize=8pt]{->}{D1}{T1}
\nczigzag[coilwidth=0.2cm,coilheight=1.5,coilarm=0.4cm,arrowsize=9pt]{->}{I2}{A2}
\ncline[arrowsize=8pt]{->}{A2}{B2}
\ncline[arrowsize=8pt]{->}{B2}{EtcF2}
\ncline[arrowsize=8pt]{->}{EtcS2}{C2}
\ncline[arrowsize=8pt]{->}{C2}{D2}
\ncline[arrowsize=8pt]{->}{D2}{T2}
\ncarc[arrowsize=8pt,arcangle=20]{->}{A1}{A2}
\ncarc[arrowsize=8pt,arcangle=20]{->}{A2}{A1}
\ncarc[arrowsize=8pt,arcangle=20]{->}{B1}{B2}
\ncarc[arrowsize=8pt,arcangle=20]{->}{B2}{B1}
\ncarc[arrowsize=8pt,arcangle=20]{->}{C1}{C2}
\ncarc[arrowsize=8pt,arcangle=20]{->}{C2}{C1}
\ncarc[arrowsize=8pt,arcangle=20]{->}{D1}{D2}
\ncarc[arrowsize=8pt,arcangle=20]{->}{D2}{D1}
\end{pspicture}
}
}
\setcounter{subfigure}{0}
\centering
\subfigure{
\includegraphics[scale=0.18]{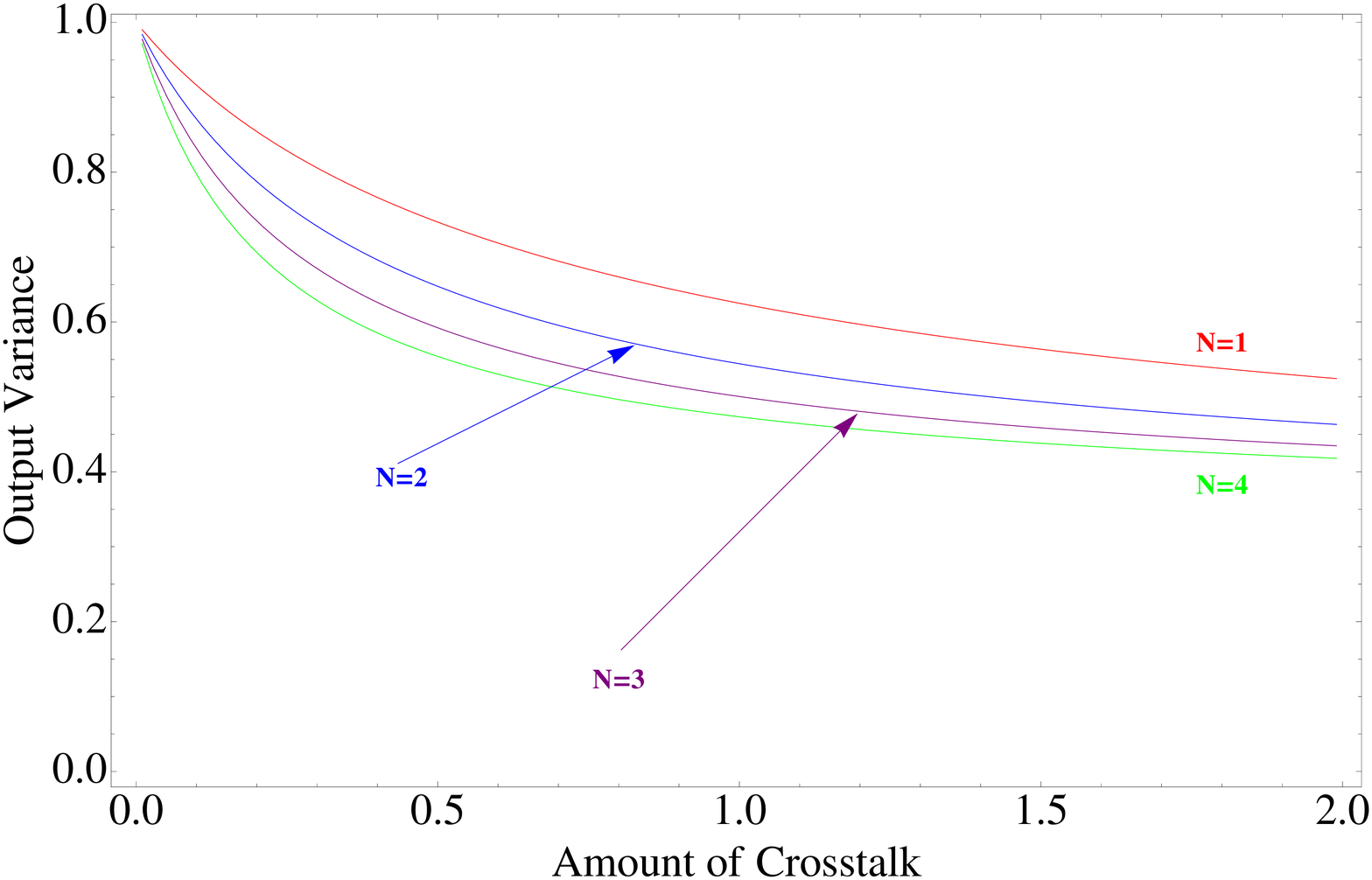}
}
\end{minipage}\hfill
\begin{minipage}[t]{0.5\linewidth}
\centering
\subfigure{
\centering
\psscalebox{0.45}{
\begin{pspicture}(-1,-3)(13,2.5)
\rput[l](-1,2){\rnode{Noise1}{\Huge{Noise 1}}}
\pnode(0,1){I1}
\rput[l](-1,-2){\rnode{Noise2}{\Huge{Noise 2}}}
\pnode(0,-1){I2}
\cnodeput(2,1){A1}{\strut}
\cnodeput(4,1){B1}{\strut}
\rput[l](5.5,1){\rnode{EtcF1}{}}
\rput[l](5.6,1){\rnode{Etc1}{\Huge{...}}}
\rput[l](6.3,1){\rnode{EtcS1}{}}
\cnodeput(8,1){C1}{\strut}
\cnodeput(2,-1){A2}{\strut}
\cnodeput(4,-1){B2}{\strut}
\rput[l](5.5,-1){\rnode{EtcF2}{}}
\rput[l](5.6,-1){\rnode{Etc2}{\Huge{...}}}
\rput[l](6.3,-1){\rnode{EtcS2}{}}
\cnodeput(8,-1){C2}{\strut}
\cnodeput(10,0){D}{\strut}
\pnode(12,0){T}
\rput[l](11,1){\rnode{Output}{\Huge{Output}}}
\nczigzag[coilwidth=0.2cm,coilheight=1.5,coilarm=0.4cm,arrowsize=9pt]{->}{I1}{A1}
\ncline[arrowsize=8pt]{->}{A1}{B1}
\ncline[arrowsize=8pt]{->}{B1}{EtcF1}
\ncline[arrowsize=8pt]{->}{EtcS1}{C1}
\ncline[arrowsize=8pt]{->}{C1}{D}
\nczigzag[coilwidth=0.2cm,coilheight=1.5,coilarm=0.4cm,arrowsize=9pt]{->}{I2}{A2}
\ncline[arrowsize=8pt]{->}{A2}{B2}
\ncline[arrowsize=8pt]{->}{B2}{EtcF2}
\ncline[arrowsize=8pt]{->}{EtcS2}{C2}
\ncline[arrowsize=8pt]{->}{C2}{D}
\ncarc[arrowsize=8pt,arcangle=20]{->}{A1}{A2}
\ncarc[arrowsize=8pt,arcangle=20]{->}{A2}{A1}
\ncarc[arrowsize=8pt,arcangle=20]{->}{B1}{B2}
\ncarc[arrowsize=8pt,arcangle=20]{->}{B2}{B1}
\ncarc[arrowsize=8pt,arcangle=20]{->}{C1}{C2}
\ncarc[arrowsize=8pt,arcangle=20]{->}{C2}{C1}
\ncline[arrowsize=8pt]{->}{D}{T}
\end{pspicture}
}
}
\setcounter{subfigure}{1}
\centering
\subfigure{
\includegraphics[scale=0.18]{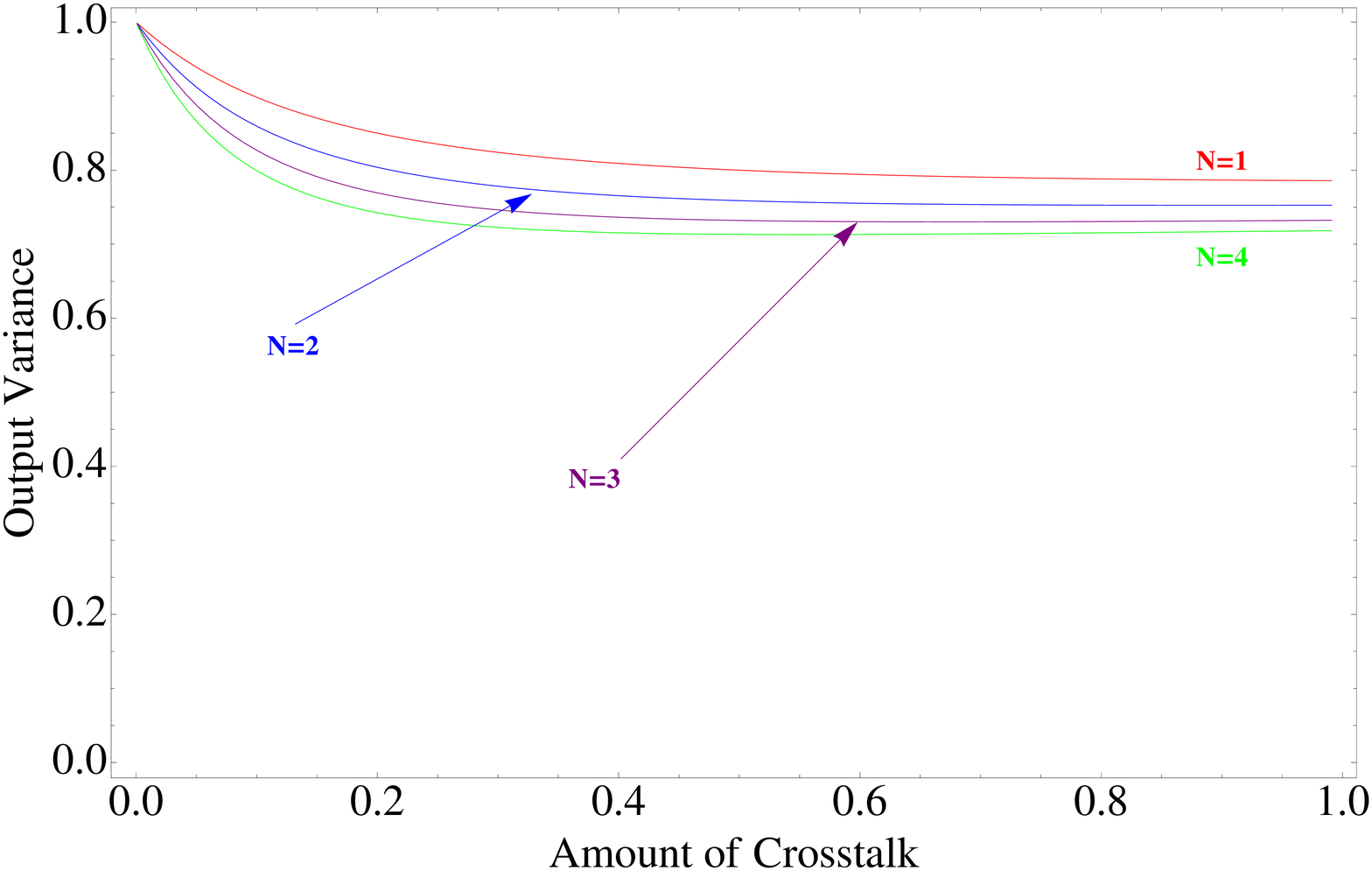}
}
\end{minipage}
\end{center}
\caption{Output variance when crosstalk is present among all stages of two different pathways for various pathway lengths, when their output is different (left) or the same (right). The output variances are normalized by the variance of a pathway without crosstalk. We assume that every stage of the pathway has some noise input.
A small amount of crosstalk can help reduce the effect of noise in the output, but more crosstalk does not help filtering out the noise of the system.
Crosstalk has a much smaller effect when the two pathways have the same output. Although it reduces the variance of the intermediate nodes, it creates correlations among them, that in turn increases the variance in the output.
}
\label{ParallelPathwaysXtalkNoiseReduction}
\end{figure}
\end{center}

\subsection{Crosstalk Modeling: Direct Conversion and Intermediate Nodes}
Suppose we have a simple decomposed system:
\begin{equation}
\begin{aligned}
dY_{1}&=-aY_{1}dt+\sigma dU_{t} \\
dY_{2}&=-aY_{2}dt+\sigma dW_{t}.\\
\end{aligned}
\end{equation}
The two outputs of the system are completely independent, since they do not interact in any way, and therefore are uncorrelated.
The variance of each output is:
\begin{figure}[!b]
\centering
\subfigure[]{
\psscalebox{0.5}{
\begin{pspicture}(-4,-0.5)(20,7)
{
\rput[l](-2,1){\rnode{InputA}{}}
\rput[l](-4,1){\rnode{Noise3A}{\Large{Noise 3}}}
\cnodeput[](0,1){A}{\strut }
\rput[l](2,1){\rnode{OutputA}{}}
\rput[l](2.3,1){\rnode{VarianceOuput3A}{\Large{Output 3}}}
\rput[l](-2,3){\rnode{InputB}{}}
\rput[l](-4,3){\rnode{Noise2A}{\Large{Noise 2}}}
\cnodeput[](0,3){B}{\strut}
\rput[l](2,3){\rnode{OutputB}{}}
\rput[l](2.3,3){\rnode{VarianceOuput2A}{\Large{Output 2}}}
\rput[l](-2,5){\rnode{InputC}{}}
\rput[l](-4,5){\rnode{Noise1A}{\Large{Noise 1}}}
\cnodeput[](0,5){C}{\strut}
\rput[l](2,5){\rnode{OutputC}{}}
\rput[l](2.3,5){\rnode{VarianceOuput1A}{\Large{Output 1}}}

\rput[l](7,3){\rnode{TRANSFORMATION}{\Huge{$\Rightarrow$} }}

\rput[l](12,5){\rnode{InputD}{}}
\rput[l](10,5){\rnode{Noise3B}{\Large{Noise 3}}}
\cnodeput[](14,5){D}{\strut }
\rput[l](16,5){\rnode{OutputD}{}}
\rput[l](16.3,5){\rnode{VarianceOuput3}{\Large{Output 3}}}
\rput[l](15,3){\rnode{InputE}{}}
\rput[l](13.3,3){\rnode{Noise2B}{\Large{Noise 2}}}
\cnodeput[](17,3){E}{\strut}
\rput[l](19,3){\rnode{OutputE}{}}
\rput[l](19.3,3){\rnode{VarianceOuput2}{\Large{Output 2}}}
\rput[l](12,1){\rnode{InputF}{}}
\rput[l](10,1){\rnode{Noise1B}{\Large{Noise 1}}}
\cnodeput[](14,1){F}{\strut}
\rput[l](16,1){\rnode{OutputF}{}}
\rput[l](16.3,1){\rnode{VarianceOuput1}{\Large{Output 1}}}
}
\nczigzag[coilwidth=0.2cm,coilheight=1.5,coilarm=0.4cm,arrowsize=9pt]{->}{InputA}{A}
\ncline[arrowsize=10pt]{->}{A}{OutputA}
\nczigzag[coilwidth=0.2cm,coilheight=1.5,coilarm=0.4cm,arrowsize=9pt]{->}{InputB}{B}
\ncline[arrowsize=10pt]{->}{B}{OutputB}
\nczigzag[coilwidth=0.2cm,coilheight=1.5,coilarm=0.4cm,arrowsize=9pt]{->}{InputC}{C}
\ncline[arrowsize=10pt]{->}{C}{OutputC}
\nczigzag[coilwidth=0.2cm,coilheight=1.5,coilarm=0.4cm,arrowsize=9pt]{->}{InputD}{D}
\ncline[arrowsize=10pt]{->}{D}{OutputD}
\nczigzag[coilwidth=0.2cm,coilheight=1.5,coilarm=0.4cm,arrowsize=9pt]{->}{InputE}{E}
\ncline[arrowsize=10pt]{->}{E}{OutputE}
\nczigzag[coilwidth=0.2cm,coilheight=1.5,coilarm=0.4cm,arrowsize=9pt]{->}{InputF}{F}
\ncline[arrowsize=10pt]{->}{F}{OutputF}
\ncline[arrowsize=10pt]{<->}{D}{E}
\ncline[arrowsize=10pt]{<->}{E}{F}
\ncarc[arrowsize=8pt,arcangle=-55]{<->}{D}{F}

\end{pspicture}
}
}
\subfigure[]{
\centering
\includegraphics[scale=0.30]{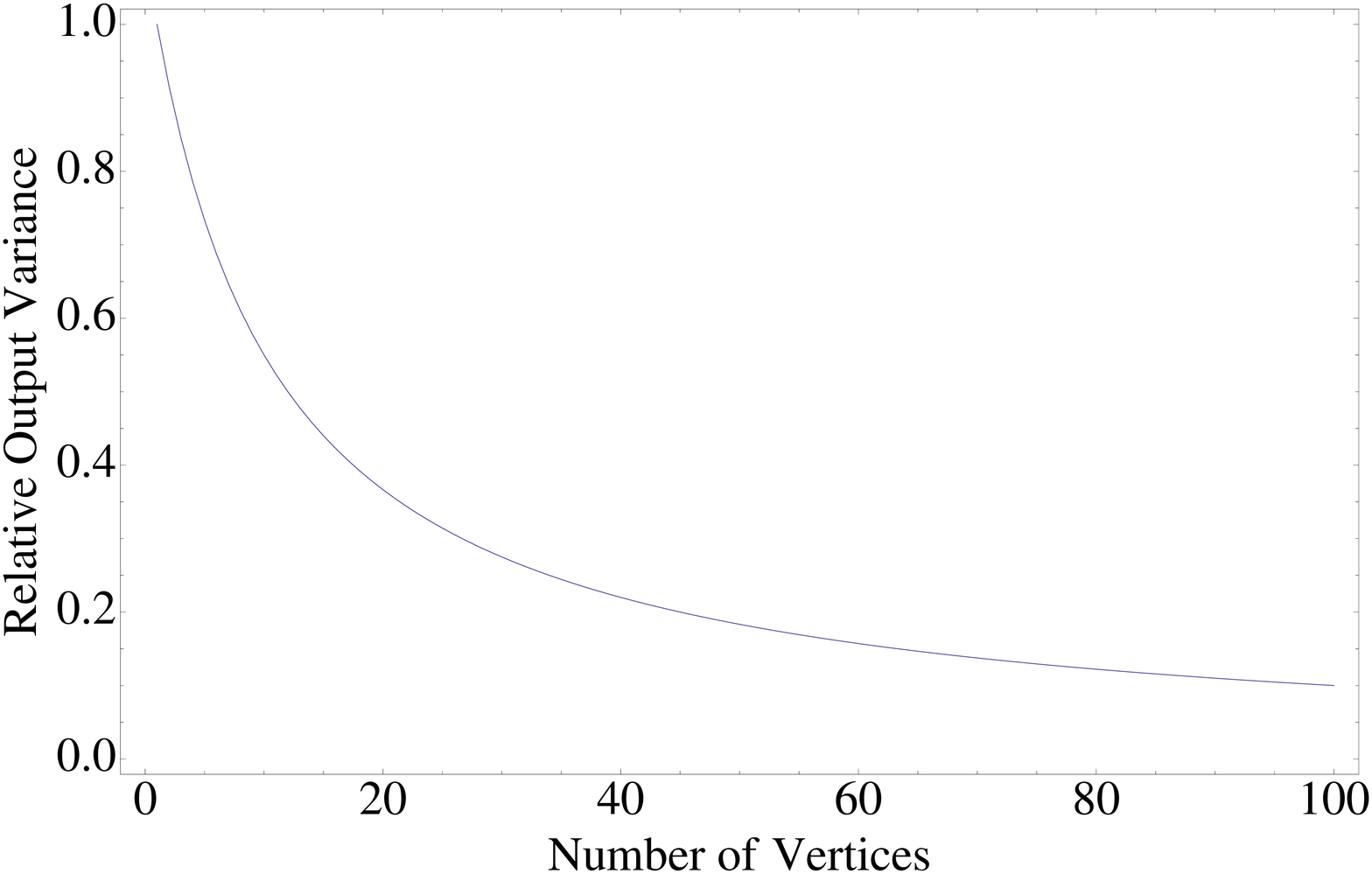}
}
\caption{Output variation for each node in a system of $N$ nodes, when there are crosstalk interactions among every pair of nodes. The variance has been normalized by the corresponding variance without crosstalk. Each node is identical, and receives an independent noise input of the same intensity. When the number of vertices increases, the noise is distributed among all the nodes, thus the output variance is reduced.}
\label{SingleStepCompleteGraphXtalk}
\end{figure}

\begin{equation}
\mathbb{V}[Y_{1}]=\sigma^{2}_{y}=\frac{\sigma^{2}}{2\pi}\int _{-\infty}^{+\infty} \frac{1}{\omega^{2}+ a^{2}}d\omega =\frac{\sigma^{2}}{2a}.
\end{equation}
The system is symmetric, thus $\mathbb{V}[Y_{1}]=\mathbb{V}[Y_{2}]=\sigma^{2}_{y}$.
If there is crosstalk, then the different states of the system are correlated.
If we model crosstalk as a positive conversion rate from one state to another, with the conversion rates being equal among every pair of states, the $2-$state system above becomes:
\begin{equation}
\begin{aligned}
dY_{1}&=-(a+c)Y_{1}dt+cY_{2}dt+\sigma dU_{t} \\
dY_{2}&=-(a+c)Y_{2}dt+cY_{1}dt+\sigma dW_{t}.\\
\end{aligned}
\label{DirectXtalkAmongTwoNodes}
\end{equation}
The variance of each of the outputs now becomes:
\begin{equation}
\begin{aligned}
\mathbb{V}[Y_{1}]&=\int _{-\infty}^{+\infty} (|h_{11}(f)|^{2}+|h_{21} (f)|^{2})df \\
&=\frac{\sigma^{2}}{2\pi}\int _{-\infty}^{+\infty} \left( \left| \frac{a+c+j\omega}{(a+j\omega) (a+2c+j\omega)}\right|^{2}+\left| \frac{c}{(a+j\omega) (a+2c+j\omega)}\right|^{2}   \right)d\omega\\
&=\frac{\sigma^{2}}{2a}\cdot \frac{a+c}{a+2 c},
\end{aligned}
\end{equation}
where $h_{11}$ and $h_{21}$ are the impulse responses of the first node when the input is an impulse response to the first and second node respectively.
The symmetry is preserved, so $\mathbb{V}[Y_{1}]=\psi^{2}_{y}=\mathbb{V}[Y_{2}]$.
The variance when crosstalk is present ($c>0$) is always smaller than the initial variance of the outputs.
Generalizing the equations above for $N$ nodes (see Figure \ref{SingleStepCompleteGraphXtalk}), we find that 
\begin{equation}
\sigma^{2}_{y}=\frac{\sigma^{2}}{2a} \qquad \psi^{2}_{N}=\frac{a+c}{(a+Nc)}\sigma_{y}^{2}
\end{equation}
and as a result,
\begin{equation}
\frac{\psi^{2}_{N}}{\sigma^{2}_{y}}=\frac{a+c}{a+Nc}
\end{equation}
which tends to zero as $N$ becomes large.

Alternatively, we can model crosstalk interactions as two species being converted to an intermediate complex, as has been done in the previous sections.
A very simple example of a chemical reaction network which demonstrates this type of behavior is

\begin{equation}
\begin{aligned}
Y_{1} & \xrightarrow{} A  \\
Y_{2} & \xrightarrow{} B  \\
Y_{1} + Y_{2} & \xrightleftharpoons[]{} Z.
\end{aligned}
\end{equation}
Crosstalk is defined by the presence of the last reaction.
We are interested in the variance in the concentration of the output products $A$ and $B$, which are directly affected by the variance of $Y_{1}$ and $Y_{2}$.
The two pathways will interact through an intermediate vertex.
The system can be written as
\begin{equation}
\begin{aligned}
dY_{1}&=-aY_{1}dt-cY_{1}Y_{2}dt+f Z dt+\sigma dU_{t} \\
dY_{2}&=-aY_{2}dt-cY_{1}Y_{2}dt+f Z dt+\sigma dW_{t} \\
dZ&=cY_{1}Y_{2}dt-fZdt. \\
\end{aligned}
\label{IndirectXtalkAmongTwoNodes}
\end{equation}

We assume that there is a new ``crosstalk vertex''  $Z$ among each pair of original vertices.
After linearizing around an equilibrium point $(\bar{Y}_{1},\bar{Y}_{2},\bar{Z})$, these equations become
\begin{equation}
\begin{aligned}
dY_{1}&=-(a+c\bar{Y}_{2})Y_{1}dt-c\bar{Y}_{1}Y_{2}dt+f Z dt+\sigma dU_{t} \\
dY_{2}&=-(a+c\bar{Y}_{1})Y_{2}dt-c\bar{Y}_{2}Y_{1}dt+f Z dt+\sigma dW_{t} \\
dZ&=c\bar{Y}_{2} Y_{1}dt+c\bar{Y}_{1}Y_{2}dt-fZdt. \\
\end{aligned}
\end{equation}
We find that this network is now more capable of reducing the effect of noise in the output (Figure \ref{CrosstalkNoiseDirectIndirectComparison}).

\FloatBarrier

\begin{figure}[htbp] 
\centering
\includegraphics[scale=0.4]{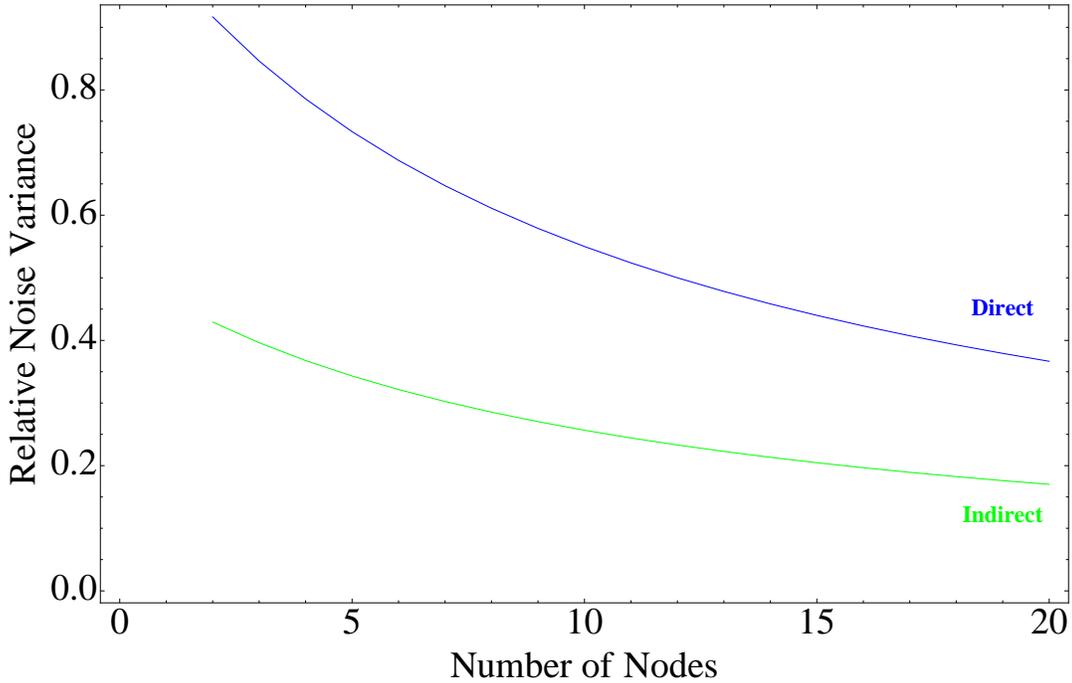} 
\caption{Comparison of the noise in the output of a simple network with two different implementations of crosstalk, direct conversion or forming a new complex, as described by equations \eqref{DirectXtalkAmongTwoNodes} and \eqref{IndirectXtalkAmongTwoNodes}. }
\label{CrosstalkNoiseDirectIndirectComparison}
\end{figure}

\section{Multiplicative Noise}

There are cases where the noise intensity is proportional to a state of the system.
In biological networks for example, the degradation of various proteins depends on specific enzymes, whose concentration may be subject to random fluctuations.
This makes the degradation of a protein prone to noise whose source is independent of the protein concentration, but makes the rate at which it degrades proportional to it.
The noise intensity is also proportional to the state of the system when a state is autoregulated, either with positive or negative feedback, where the rate at which the concentration of that particular state changes is subject to random noise.
We will call this type of noise multiplicative, because it is multiplied by the state of the system.
As a specific example, consider a gene that is regulated by a single regulator \cite{SystemsBiologyBook}.
The transcription interaction can be written as 
\begin{equation}
P\rightarrow X.
\end{equation}
When $P$ is in its active form, gene $X$ starts being transcribed and the mRNA is translated, resulting in accumulation of protein $X$ at a constant rate $b$.
The production of $X$ is balanced by protein degradation (by other specialized proteins) and cell dilution during growth with rate $a$.
A differential equation that describes this simple system is
\begin{equation}
\frac{dX}{dt}=b_{t}-a_{t}X.
\label{SimpleDegradation}
\end{equation}
If there is noise in the concentration of the aforementioned degradation proteins, or the cell growth, the rate $a_{t}$ is not constant, but it consists of a deterministic component, and a random component.
We will now show that noise in the production rate $b_{t}$ has a fundamentally different effect in system behavior compared to the effect of noise in the degradation rate $a_{t}$, because the latter is multiplied by the concentration of the protein itself.
We will first study the homogenous version of the differential equation \eqref{SimpleDegradation}, and then we will add the constant production term.
Ignoring the constant production term, and multiplying by $dt$, equation \eqref{SimpleDegradation} becomes
\begin{equation}
dX=-(a_{t}dt)X.
\end{equation}
After adding a random component in the degradation rate, the last equation becomes
\begin{equation}
dX=(-a_{t}dt + \sigma_{t} dW_{t})X,
\label{MultiplicativeNoiseDefinition}
\end{equation}
where $W_{t}$ is the regular Wiener process and $dW_{t}$ represents the noise term.
Note that the degradation rate and the noise intensity are allowed to be time-dependent.
We will first find the differential of the logarithm of $X$ using It\^o's lemma. 
We will again require that all the input functions are continuous and non-pathological, so that we can always  change the order of taking the limit and the expectation operator. 
We will additionally assume that all integrals are finite, so that we can also change the order of integration.
The technical details mentioned above are covered in more detail in \cite{ProbabilityBook} and \cite{StochasticsHandbook}.

We apply It\^o 's lemma on the logarithm of the random variable $X$, which obeys equation \eqref{MultiplicativeNoiseDefinition}:
\begin{equation}
f(X,t)=\log X(t)
\end{equation}
and applying It\^o's lemma,  we get
\begin{equation}
\begin{aligned}
d \log(X) &=d f(X,t) \\
&=\frac{\partial f}{\partial t} dt +\frac{\partial f}{\partial X} dX +\frac{1}{2}\frac{\partial ^{2} f}{\partial X^{2}} dX^{2} \\
&=0+\frac{dX}{X} -\frac{1}{2}\frac{1}{X^{2}} \left( a_{t}^{2}X^{2}dt^{2}-2a_{t}\sigma_{t} X^{2}dtdW_{t}+\sigma_{t}^{2}X^{2}dW_{t}^{2}  \right) \\
&=\left(-a_{t} dt + \sigma_{t} dW_{t} -\frac{1}{2}\sigma^{2}dW_{t}^{2}\right)+X^{2}\left(a_{t}^{2}dt^{2}-2a_{t}\sigma_{t} dtdW_{t} \right). \\
\end{aligned}
\label{GeometricNoiseItoLemma}
\end{equation}
The last two terms can be neglected, since $dt^{2}=\OO{dt}$ and  $dt\cdot dW_{t}=\OO{dt}$ as $dt\rightarrow 0$.
On the other hand, as $dt$ becomes small, 
\begin{equation}
\lim _{dt \rightarrow 0} dW_{t}^{2}=\mathbb{E}[dW_{t}^{2}]=dt.
\end{equation}
Applying the rules above to equation \eqref{GeometricNoiseItoLemma}, 
\begin{equation}
\begin{aligned}
\log \frac{X(t)}{X_{0}}&=-\int _{0}^{t} \left( a_{s}+\frac{1}{2}\sigma^{2}_{s} \right)ds +\sigma_{t} W_{t}.
\end{aligned}
\end{equation}
We can now solve for $X(t)$:
\begin{equation}
X(t)=X_{0}e^{-\int _{0}^{t} \left( a_{s}+\frac{1}{2}\sigma^{2}_{s}\right) ds} \cdot e^{\sigma _{t} W_{t}}.
\label{GeometricNoiseSolution}
\end{equation}
The above derivation is valid only when the equilibrium state (concentration) is equal to zero and we start from a state $X_{0}\neq 0$.
If the rate $a$ and the noise strength $\sigma$ are constant, it simplifies to 
\begin{equation}
\begin{aligned}
X(t)&=X_{0}e^{-\left( a+\frac{1}{2}\sigma^{2}\right)t} \cdot e^{\sigma W_{t}}.\\
\end{aligned}
\end{equation}
When the equilibrium is positive (which is the case for most systems), the following differential equation is more relevant:
\begin{equation}
dY=b dt+(-adt + \sigma dW_{t})Y.
\label{NonHomogeneousSDE}
\end{equation}
One way to view the terms on the right hand side of equation \eqref{NonHomogeneousSDE} is that the concentration of species $X$ depends on a deterministic input, and is regulated by a negative feedback mechanism which is subject to random disturbances.
It has been shown in \cite{LestasPaper} that when feedback is also noisy, there are fundamental limits on how much the noise in the output can be reduced, because there are bounds on how well we can estimate the state of the system.
In \cite{LestasPaper} the authors focus on discrete random events (birth-death processes) as the source of noise, and the result is that feedback noise makes it harder to control the noise in the output.
We will also show that in our setting multiplicative noise results in larger variance than additive noise of equal strength, and in the next section we will show it propagates in a cascade of linear filters.

Using It\^o's lemma once more, and the solution to the homogeneous equation, we find that the solution to the nonhomogeneous case is
\begin{equation}
\begin{aligned}
Y(t)&=Y_{0}X(t) +b X(t) \int _{0}^{t}X^{-1}(s)ds \\
&=Y_{0}e^{- \int_{0}^{t} \left( a_{u}+\frac{1}{2}\sigma_{u}^{2}\right)du} \cdot e^{\sigma_{t} W_{t}} + b \int _{0}^{t}e^{- \int_{s}^{t} \left( a_{u}+\frac{1}{2}\sigma_{u}^{2}\right)du}\cdot e^{\sigma_{t} W_{t}-\sigma_{s}W_{s}}ds \\
\end{aligned}
\label{GeometricWienerProcessGeneralSolution}
\end{equation}
where $X(t)$ is the solution of the homogeneous equation \eqref{GeometricNoiseSolution} with initial condition $X(t=0)=1$.

If the initial state is equal to zero (or when $t$ is large), and the all the parameters are constant, then we can simplify the last expression as 
\begin{equation}
Y(t)=b \int _{0}^{t}e^{-\left( a+\frac{1}{2}\sigma^{2}\right)u}\cdot e^{\sigma W_{u}}du. \\
\label{GeometricWienerProcessSimpleSolution}
\end{equation}
Note that the form of the last equation is fundamentally different from the response of linear systems to input noise, because here the Wiener process input depends on the same time variable as the kernel of the integral.
In other words, the output is not a convolution of the impulse response of the system with the input.
In order to see how the noise propagates through the network, and given that we cannot use the solution \eqref{FrequencyDomainVarianceInputFrequency}, it is helpful to find the correlation of two versions of this stochastic process, so that we find its frequency content.

As a first step, we will compute the correlation of the exponential of Brownian motion.
The expected value is
\begin{equation}
\begin{aligned}
\mathbb{E}\left[Z_{t} \right] &=\mathbb{E}\left[e^{\sigma W_{t}} \right]  \\
&=\int _{-\infty}^{+\infty}e^{\sigma\sqrt{t} x} \frac{1}{\sqrt{2\pi t }} e^{-\frac{x^{2}}{2 t}}dx \\
&=e^{\frac{1}{2}\sigma^{2}t}.
\end{aligned}
\end{equation}
The expected value of the square of the exponential Wiener process is
\begin{equation}
\begin{aligned}
\mathbb{E}\left[Z_{t}^{2} \right] &=\mathbb{E}\left[e^{2 \sigma W_{t}} \right]  \\
&=\int _{-\infty}^{+\infty}e^{2\sigma\sqrt{t} x} \frac{1}{\sqrt{2\pi t }} e^{-\frac{x^{2}}{2 t}}dx \\
&=e^{2\sigma^{2}t}.
\end{aligned}
\end{equation}
Combining the last two equations:
\begin{equation}
\begin{aligned}
\sigma^{2}_{Z_{t}}=Var\left[Z_{t}\right] &=\mathbb{E}\left[Z_{t}^{2} \right]-\left( \mathbb{E}\left[Z_{t} \right]  \right) ^{2} \\
&=e^{2\sigma^{2}t}-e^{\sigma^{2}t} \\
&=e^{\sigma^{2}t} \left( e^{\sigma^{2}t}-1 \right).
\end{aligned}
\end{equation}
The expected value of $Y(t)$ in equation \eqref  {GeometricWienerProcessSimpleSolution} can now be computed:
\begin{equation}
\begin{aligned}
\mathbb{E}\left[Y(t) \right] &=b \int _{0}^{t}e^{-\left( a+\frac{1}{2}\sigma^{2}\right)u}\cdot \mathbb{E} \left[ e^{\sigma W_{u}} \right] du \\
&=b \int _{0}^{t}e^{-\left( a+\frac{1}{2}\sigma^{2}\right)x}\cdot e^{\frac{1}{2}\sigma^{2}x} dx \\
&=\frac{b}{a} (1-e^{-at})
\end{aligned}
\end{equation}
which means that
\begin{equation}
\bar{Y}=\lim _{t\rightarrow \infty} \mathbb{E}[Y(t)]=\frac{b}{a}.
\end{equation}
As  one would expect, it the same as when the system is completely deterministic.
Next, we need to compute the covariance of two realizations of the random process $Z_{t}$:
\begin{equation}
\begin{aligned}
\text{Cov}\left[Z_{s},Z_{t}\right] &= \mathbb{E}\left[ Z_{s}\cdot Z_{t}\right] - \mathbb{E}\left[ Z_{s}\right] \cdot \mathbb{E}\left[ Z_{t}\right] \\
&=\mathbb{E}\left[ e^{\sigma W_{s}} e^{\sigma W_{t}} \right]  - \mathbb{E}\left[ e^{\sigma  W_{s}}\right] \cdot \mathbb{E}\left[ e^{\sigma W_{t}}\right] \\
&=\mathbb{E}\left[ e^{ \sigma W_{s\wedge t}} e^{\sigma W_{s\vee t}} \right] - e^{\frac{1}{2}\sigma^{2} (s+t) } \\
&=\mathbb{E}\left[ e^{2\sigma W_{s\wedge t}} e^{\sigma (W_{s\vee t}-W_{s\wedge t}) }\right]- e^{\frac{1}{2}\sigma^{2} (s+t) }  \\
&=\mathbb{E}\left[ e^{2\sigma  W_{s\wedge t}}\right] \cdot \mathbb{E}\left[  e^{\sigma (W_{s\vee t}-W_{s\wedge t}) }\right]- e^{\frac{1}{2}\sigma^{2} (s+t) }  \\
&=e^{2\sigma^{2}s\wedge t } \cdot e^{ \frac{1}{2} \sigma ^{2} (s\vee t-s\wedge t) }- e^{\frac{1}{2}\sigma^{2} (s+t) }  \\
\end{aligned}
\end{equation}
where we follow the standard notation $s\wedge t =\min(s,t)$ and $s\vee t =\max(s,t)$.
Combining all the equations above, we can find the correlation for the geometric Brownian motion:
\begin{equation}
\begin{aligned}
R(s,t) &=\text{Corr} \left[Z_{s},Z_{t} \right] \\
&= \frac{\text{Cov} \left[Z_{s},Z_{t} \right]}{\sigma _{Z_{s}}\cdot \sigma_{Z_{t}}} \\
&=\frac{e^{2\sigma^{2}s\wedge t } \cdot e^{ \frac{1}{2} \sigma ^{2} (s\vee t-s\wedge t) }- e^{\frac{1}{2}\sigma^{2} (s+t) }  }{ \sqrt{e^{\sigma^{2}s} \left( e^{\sigma^{2}s}-1 \right)} \sqrt{ e^{\sigma^{2}t} \left( e^{\sigma^{2}t}-1 \right)} }\\
&=\sqrt{\frac{e^{\sigma^{2}s \wedge t}-1}{e^{\sigma^{2}s \vee t}-1} }.
\end{aligned}
\end{equation}
We now define the covariance and correlation of two such processes with time lag $\tau$ in the equilibrium state as:
\begin{equation}
\displaystyle C(\tau)=\lim _{t\rightarrow \infty}C(t,t+\tau) \qquad \textrm{and} \qquad \displaystyle R(\tau)=\lim _{t\rightarrow \infty}R(t,t+\tau).
\end{equation} 
Applying this definition to the general correlation formula of geometric Brownian Motion, 
\begin{equation}
\begin{aligned}
R(\tau)&= \lim _{t\rightarrow \infty} \sqrt{\frac{e^{\sigma^{2}t}-1}{e^{\sigma^{2}(t+\tau)}-1} } \\
&=\sqrt{\frac{e^{\sigma^{2}t}}{e^{\sigma^{2}(t+\tau)}} } \\
&=e^{-\frac{1}{2}\sigma^{2}\tau}.
\end{aligned}
\end{equation}
So the correlation is exponentially decreasing as a function of the time lag.

We can now follow the same procedure in order to find the correlation of the stochastic process defined by equation \eqref{GeometricWienerProcessSimpleSolution}.

Its second moment is equal to
\begin{equation}
\begin{aligned}
\mathbb{E}[Y^{2}(t)] &=b^{2} \int _{0}^{t} \int _{0}^{t} e^{-\left(a+\frac{\sigma^{2}}{2}\right)(x+y)} \cdot \mathbb{E} \left[  e^{\sigma W_{x}} e^{\sigma W_{y}}  \right]dx dy \\
&=b^{2}\int _{0}^{t} \int _{0}^{t} e^{-\left(a+\frac{\sigma^{2}}{2}\right)(x+y)} e^{2\sigma^{2}x\wedge y } e^{ \frac{1}{2} \sigma ^{2} (x\vee y-x\wedge y) }dx dy \\
&=b^{2}\int _{0}^{t} \int _{0}^{x} e^{-\left(a+\frac{\sigma^{2}}{2}\right)(x+y)} e^{2\sigma^{2}y } e^{ \frac{1}{2} \sigma ^{2} (x-y) }dx dy \\
&\qquad \qquad +b^{2}\int _{0}^{t} \int _{x}^{t} e^{-\left(a+\frac{\sigma^{2}}{2}\right)(x+y)} e^{2\sigma^{2}x } e^{ \frac{1}{2} \sigma ^{2} (y-x) }dx dy \\
&=\frac{2 \left(a \left(1-2 e^{-a t}+e^{t \left(-2 a+\sigma ^2\right)}\right)+\left(-1+e^{-a t}\right) \sigma ^2\right)}{a \left(2 a^2-3 a \sigma ^2+\sigma ^4\right)}
\end{aligned}
\end{equation}
where we have assumed that all integrals are finite, which means that the rate $a$ has to be greater than the input variance $\frac{\sigma^{2}}{2}$.
As $t$ goes to infinity, we can ignore all the decaying exponentials.
\begin{equation}
\lim _{t\rightarrow \infty} \mathbb{E}[Y^{2}(t)]= \left\{ \begin{array}{ll}
\infty & \textrm{if $a\leq \frac{\sigma^{2}}{2}$}\\
\frac{b^{2}}{a(a-\frac{\sigma^{2}}{2})} & \textrm{if $a>\frac{\sigma^{2}}{2}$}.\\
\end{array} \right.
\label{GeometricNoiseSecondMoment}
\end{equation}

In what follows, we will only be interested in the behavior of the system when $a>\frac{\sigma^{2}}{2}$, because it only makes sense to compute the correlation when the standard deviation is finite.

Based on equation \eqref{GeometricNoiseSecondMoment}, the standard deviation (when it is defined) is equal to
\begin{equation}
\sigma_{Y}=\frac{b^{2}\sigma ^2}{a^{2}\left(2 a-\sigma ^2\right)}=\frac{\sigma ^2}{2 a-\sigma ^2} \bar{Y}^{2}.
\end{equation}

The standard deviation is proportional to the average value of $Y$, since the larger the value of $Y$, the larger the strength of the disturbance.

\FloatBarrier

\subsection{Multiplicative Noise Through a  Low-Pass Filter}

Assume that a pathway consists of two nodes.
The first one is affected by multiplicative noise, and it is used as an input to the second node.
We first analyze a system where each state has a single real pole, and later on we will generalize it for an arbitrary number of poles.
The equations of the system are
\begin{equation}
\begin{aligned}
dX&=c dt+ (-fdt + \sigma dW_{t} ) X \\
dY&=bX-aY.
\label{MultiplicativeAdditiveNoiseEquations}
\end{aligned}
\end{equation}
Combining the forms for the multiplicative noise and the output of a single pole filter,
\begin{equation}
Y(t)=b c e^{-at} \int _{0}^{t} e^{as} \left(  \int _{0}^{s} e^{-\left(f+\frac{\sigma^{2}}{2} \right)u}e^{\sigma W_{u}} du\right) ds.
\end{equation}
The mean is equal to:
\begin{equation}
\begin{aligned}
\mathbb{E}[Y(t)]&=b c e^{-at} \int _{0}^{t} e^{as} \left(  \int _{0}^{s} e^{-\left(f+\frac{\sigma^{2}}{2} \right)u} \mathbb{E}\left[e^{\sigma W_{u}}\right] du\right) ds \\
&=b c e^{-at} \int _{0}^{t} e^{as} \left(  \int _{0}^{s} e^{-fu} du\right) ds \\
&=\frac{b c \left(a-a e^{-f t}+\left(-1+e^{-a t}\right) f\right)}{a (a-f) f}.
\end{aligned}
\end{equation}
The last equation also holds when $a=f$, and we can find the expected value by finding the limit as $f\rightarrow a$.
Letting the time $t$ go to infinity, 
\begin{equation}
\mathbb{E}[Y]=\lim _{t\rightarrow \infty} \mathbb{E}[Y(t)]=\frac{bc}{af}
\label{MultiplicativeAdditiveMean}
\end{equation}
which is exactly the same as an equivalent system without any noise.

The second moment is
\begin{equation}
\begin{aligned}
\mathbb{E}[Y^{2}]=b^{2}c^{2}e^{-2at} \int_{0}^{t} e^{ar}dr \int_{0}^{t}e^{as}ds \int_{0}^{r} \int_{0}^{s} e^{-(f+\frac{\sigma^{2}}{2})(x+y)} \mathbb{E}\left[ e^{\sigma(W_{x}+W_{y})}\right]dx dy.
\end{aligned}
\end{equation}
We break the integral above in five parts, in order to compute the expected value inside it:
\begin{equation}
\begin{aligned}
\frac{e^{2at}}{b^{2}c^{2}}\mathbb{E}[Y^{2}(t)]&=\int _{0}^{t} e^{a r} dr \int _{r}^{t} e^{a s} ds \int_{0}^{r} \left(\int_0^x e^{-f x} e^{-f y} e^{\sigma ^2 y} \, dy\right) dx\\
&\qquad +\int _0^te^{a r} dr \int _r^te^{a s} ds \int_0^r \left(\int _x^se^{-f x} e^{-f y} e^{\sigma ^2 x}dy\right) dx \\
&\qquad +\int _0^te^{a r} dr \int _0^r e^{a s} ds \int_0^s \left(\int _0^xe^{-f x} e^{-f y} e^{\sigma ^2 y}dy\right) dx \\
&\qquad +\int _0^te^{a r} dr \int _0^re^{a s} ds\int_0^s \left(\int _x^se^{-f x} e^{-f y} e^{\sigma ^2 x}dy\right) dx \\
&\qquad +\int _0^te^{a r} dr\int _0^re^{a s} ds \int_s^r \left(\int _0^se^{-f x} e^{-f y} e^{\sigma ^2 y}dy\right)  dx. \\
\end{aligned}
\end{equation}
After performing all the algebraic calculations,
\begin{equation}
\begin{aligned}
E[Y^{2}]=\lim _{t\rightarrow \infty} \mathbb{E}[Y^{2}(t)]=\frac{b^{2}c^{2}}{a^{2}f(f-\frac{\sigma^{2}}{2})}
\label{MultiplicativeAdditiveVariance}
\end{aligned}
\end{equation}
given that the second moment is finite, which happens when $f>\frac{\sigma^{2}}{2}$.
The variance is 
\begin{equation}
\mathbb{V}[Y]=b^{2}c^{2}\frac{\sigma^{2}}{a^{2}f^{2}(2f-\sigma^{2})}.
\end{equation}
We can write the above equation as a constant times the variance of the first state:
\begin{equation}
\begin{aligned}
\mathbb{V}[Y]&=\left(\frac{b}{a}\right)^{2}\frac{c^{2}\sigma^{2}}{f^{2}(2f-\sigma^{2})} \\
&=\left(\frac{b}{a}\right)^{2} \mathbb{V}[X].
\end{aligned}
\end{equation}

The variance of $Y$ is fundamentally different from the variance in the case when white noise is added directly to the input, in which case, it would be equal to
\begin{equation}
V_{0}[Y]=\frac{b^{2}}{2a} \sigma_{in}^{2}.
\end{equation}

The time evolution of the variance is shown in Figure \ref{GeometricVsAdditiveVarianceEvolution}.
When the noise is multiplicative, it takes longer for the variance to settle to its steady state value, which is also an indication that the output variance consists of lower frequencies than in the case of additive noise.

\begin{figure}[tb]
\centering
\includegraphics[scale=0.40]{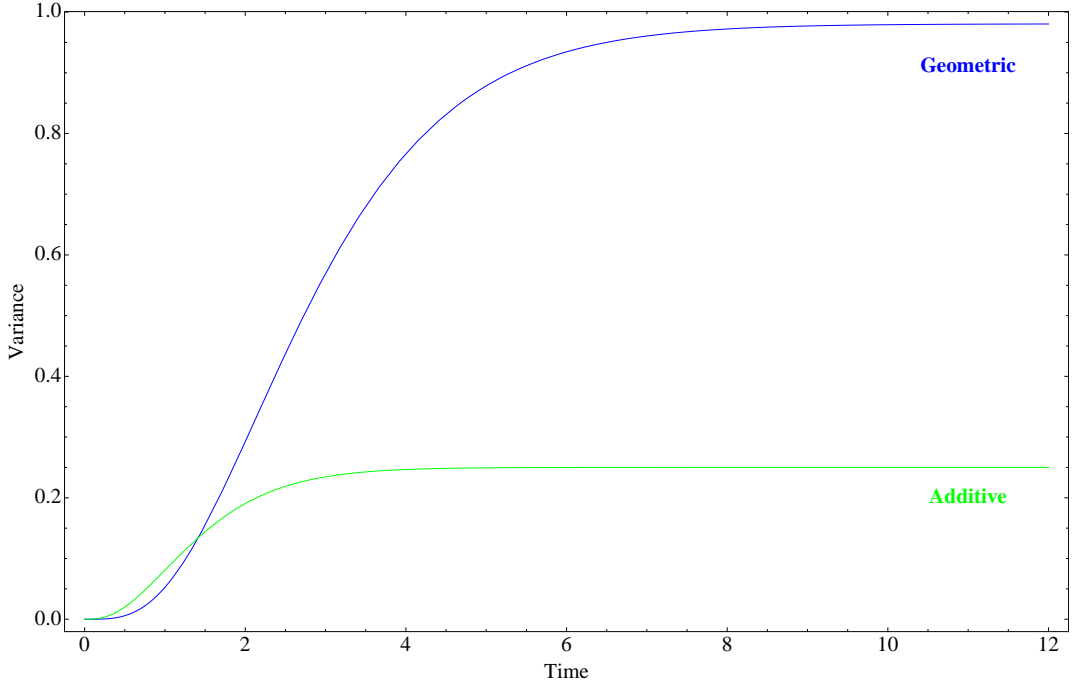}
\caption{Evolution of the output variance of a single pole filter when the input is affected by additive and multiplicative noise respectively. The system with additive noise has less variance in the output compared to the one with multiplicative noise.
Also, in the case of geometric noise, the variance takes more time to settle to its equilibrium value.}
\label{GeometricVsAdditiveVarianceEvolution}
\end{figure}

More generally, if we pass the output of the multiplicative noise through an arbitrary linear filter with impulse response $h(t)$  then the output is defined as the convolution of the impulse response and the input:
\begin{equation}
Y(t)=c\int _{0}^{t} h(t-s) \left(  \int _{0}^{s} e^{-\left(f+\frac{\sigma^{2}}{2} \right)u}e^{\sigma W_{u}} du\right) ds.
\end{equation}
The mean is
\begin{equation}
\begin{aligned}
\mathbb{E}[Y(t)]&=c\int _{0}^{t} h(t-s) \left(  \int _{0}^{s} e^{-fu} du\right) ds  \\
&=\frac{c}{f} \left(  \int _{0}^{t} (1-e^{-fs})h(t-s) ds \right).
\end{aligned}
\end{equation}
The variance is equal to
\begin{equation}
\begin{aligned}
\mathbb{V}[Y(t)]&=\mathbb{E}[Y^{2}(t)]-\left( \mathbb{E}[Y(t)] \right)^{2} \\
&=c\int_{0}^{t} h(t-r)dr \int_{0}^{r} h(t-s)ds \int_{0}^{s} \int_{0}^{y} e^{-f(x+y)}e^{\sigma^{2}x} dx dy \\
&\qquad +c\int_{0}^{t} h(t-r)dr  \int_{0}^{r} h(t-s)ds \int_{0}^{s} \int_{y}^{r} e^{-f(x+y)}e^{\sigma^{2}y} dx dy \\
&\qquad +c\int_{0}^{t} h(t-r)dr  \int_{r}^{t} h(t-s)ds \int_{0}^{r} \int_{0}^{x} e^{-f(x+y)}e^{\sigma^{2}y} dy dx \\
&\qquad +c\int_{0}^{t} h(t-r)dr  \int_{r}^{t} h(t-s)ds \int_{0}^{r} \int_{x}^{s} e^{-f(x+y)}e^{\sigma^{2}x} dy dx \\
&\qquad -\frac{c^{2}}{f^{2}} \left(  \int _{0}^{t} (1-e^{-fs})h(t-s) ds \right)^{2}.
\end{aligned}
\end{equation}
For example, if the filter has one pole at $-a$ with $a>0$, then $h(t,s)=e^{-a(t-s)}$, we can verify that the mean and the variance are equal to the ones found in equations \eqref{MultiplicativeAdditiveMean} and \eqref{MultiplicativeAdditiveVariance}.

If we have n identical single-pole filters in series, with the same pole at $-a$, with $a \in \mathbb{R}$, and their input is multiplied by $b$,  then the mean is
\begin{equation}
\begin{aligned}
\mathbb{E}[Y] &=\lim _{t\rightarrow \infty} \frac{b^{n}c}{f} \int _{0}^{t} (1-e^{-fs}) \frac{(t-s)^{n-1}}{(n-1)!}e^{-as} ds \\
&=\frac{b^{n}}{a^{n}}\cdot \frac{c}{f}
\end{aligned}
\end{equation}
and the variance is equal to
\begin{equation}
\begin{aligned}
\mathbb{V}[Y]=\left( \frac{b}{a} \right)^{2n} \left(  \frac{c}{f}\right)^{2}\frac{\sigma^{2}}{(2f-\sigma^{2})}.
\end{aligned}
\end{equation}

The above results show how variation that enters the system through noisy degradation rates affects the output of a given pathway.
For example, in the two-step cascade
\begin{equation}
\begin{aligned}
X & \rightarrow Y \\
Y & \rightarrow Z \\
\end{aligned}
\end{equation}
described by \eqref{MultiplicativeAdditiveNoiseEquations},
species $Y$ is affected by multiplicative noise, and then is used as an input to the next reaction that produces $Z$.
The second reaction acts as a first order linear filter, and the noise propagates to the pathway output $Z$.
The analysis can be used for any system that can be described by linear differential equations.
If a linear time invariant system is described by \eqref{ClassicLinearSystem} then, if there is noise in the input $u$ or its input matrix $B$, then we can consider noise a new additional input as in equation \eqref{StochasticLinearSystem}, and solve it accordingly.
The same holds for the off-diagonal elements of the dynamical matrix $A$.
But noise in the diagonal elements of $A$ is multiplicative noise, and needs to be considered separately from all other noise sources, and it leads to qualitatively different behavior than the previous kinds of input noise.

\section{Noise Propagation in Chemical Reaction Networks}

In this section, we will examine how noise propagates in general linear chemical reaction networks.
Noise in chemical reaction networks that do not involve bimolecular or higher order reactions has been studied extensively (see for example \cite{NoiseInLinearCRNs}) and chemical reactions have also been analyzed as analog signal processing systems \cite{Samoilov2002}.
In this section, we will study reactions where two or more reactants are noisy, and their disturbances may be correlated with each other.

\subsection{Motivating example}
Consider the following reaction:
\begin{equation}
X+Y \rightarrow Z.
\end{equation}
Further assume that the concentration  of $X$ and $Y$ is subject to random white noise fluctuations around a deterministic mean value:
\begin{equation}
\begin{aligned}
X_{t} &=X_{0}+\sigma _{X}dU_{t} \\
Y_{t} &=Y_{0}+\sigma _{Y}dW_{t} \\
\label{BimolecularReactionInputs}
\end{aligned}
\end{equation}
and $Z$ degrades with a rate proportional to its concentration.
The corresponding stochastic differential equation is
\begin{equation}
\begin{aligned}
dZ&=(X_{0}Y_{0}-aZ_{t})dt +X_{0}\sigma _{Y} dW_{t}+Y_{0} \sigma_{X}dU_{t}+\sigma_{X}\sigma_{Y}d[U_{t},W_{t}]
\end{aligned}
\label{BimolecularReactionWithTwoNoiseTerms}
\end{equation}
where $U_{t}$ and $W_{t}$ are standard Brownian motions.
Equation \eqref{BimolecularReactionWithTwoNoiseTerms} is a natural generalization of the case where we have only one or more noise terms that are added to the deterministic differential equation.
In all stochastic differential equations so far, we multiply the deterministic factors that contribute to the infinitesimal change in the state of the system by $dt$, and then we add the noise terms.
When we have a product of two noisy inputs, we will first consider the noiseless case, and then add all the noise terms, and their products as well.
In equation \eqref{BimolecularReactionWithTwoNoiseTerms} the deterministic term is equal to $X_{0}Y_{0}$ and the noise terms that are added are equal to $X_{t}-Y_{t}-X_{0}Y_{0}$.
The term $dU_{t}dW_{t}=d[U_{t}, W_{t}]$ is the differential of the quadratic covariation process of $U_{t}$ and $W_{t}$.
If the two processes have correlation $\rho$, then 
\begin{equation}
d[U_{t},W_{t}]=\rho dt.
\end{equation}
Simplifying the last expression for $dZ$, 
\begin{equation}
\begin{aligned}
dZ&=(X_{0}Y_{0}-aZ_{t})dt +d(X_{t}Y_{t}) \\
&=(X_{0}Y_{0}+\rho \sigma_{X}\sigma_{Y}-aZ_{t})dt +X_{0}\sigma _{Y} dW_{t}+Y_{0} \sigma_{X}dU_{t}
\end{aligned}
\end{equation}
which is the familiar Ornstein$-$Uhlenbeck process with two noise sources.
The final expression for the concentration of $Z$ is
\begin{equation}
Z(t)=\frac{1}{a}(X_{0}Y_{0}+\rho \sigma_{X}\sigma_{Y})(1-e^{-at})+\sigma_{X}Y_{0}\int_{0}^{t}e^{a(t-s)}dU_{s}+\sigma_{Y}X_{0}\int_{0}^{t}e^{a(t-s)}dW_{s}.
\end{equation}
As the effect of the initial conditions diminishes, the mean is
\begin{equation}
\bar{Z}=\lim _{t\rightarrow \infty} \mathbb{E}[Z(t)]=\frac{1}{a}(X_{0}Y_{0}+\rho \sigma_{X}\sigma_{Y})
\end{equation}
and the variance is equal to
\begin{equation}
\begin{aligned}
\mathbb{V}[Z]&=\lim _{t\rightarrow \infty} \mathbb{V}[Z(t)]=\frac{Y_{0}^{2}\sigma_{X}^{2}+X_{0}^{2}\sigma_{Y}^{2}+2X_{0}Y_{0}\rho \sigma_{X}\sigma_{Y}}{2a}.
\end{aligned}
\end{equation}
An important consequence of correlations in the input noise ($\rho \neq 0$) is that the mean is different from the case where there is no noise, even if both noise terms in \eqref{BimolecularReactionInputs} have themselves zero mean.
If the correlation is negative, the mean is lower and vice versa.
In addition, the variance is larger when there are positive correlations in the two input noise terms, as expected.
When the correlation is negative, the two noise processes partially cancel each other, resulting in lower variance.

\subsection{General Reactions}
We can generalize the above results to general reactions of the form
\begin{equation}
a_{1}X_{1}+\dots +a_{N}X_{N} \rightarrow b_{1}Y_{1}+\dots +b_{M}Y_{M}
\end{equation}
where each of the elements of the left-hand side is assumed to be a random variable that consists of a deterministic mean $\bar{X}_{k}$ and a standard white noise process $dW_{t}^{(k)}$ multiplied by the standard deviation of its concentration.
\begin{equation}
X_{k}(t)=\bar{X}_{k}+\sigma_{k} dW_{t}^{(k)}  \qquad 1\leq k \leq N.
\end{equation}
The concentration of  the product $Y_{j}$ is described by a stochastic differential equation:
\begin{equation}
\begin{aligned}
dY_{j}&=\left(b_{j}\prod _{u=1}^{N}\bar{X}_{u}-f_{j}Y_{j} \right) dt+b_{j}\sum _{k=1}^{N} \sigma_{k}\left(\prod_{\substack{u=1 \\u\neq k}}^{N}\bar{X}_{u} \right)dW_{t}^{(k)} \\
&\qquad\quad + b_{j}\sum _{k=1}^{N} \sum _{m=1}^{N}\sigma_{k}\sigma_{m} \left(\prod_{\substack{u=1 \\ u \neq k,m} }^{N}\bar{X}_{u} \right) \rho_{k,m} dt +\OO{dt}.\\
\end{aligned}
\end{equation}
The last equation is derived by using It\^o's box rule, and the fact that higher order products of Wiener processes have variance that tends to zero faster than $dt$ as $dt\rightarrow 0$.
As in the bimolecular case, we multiply the noiseless input by $dt$, as in the corresponding ordinary differential equation, and then we add all the noise terms, and their products.

The mean (disregarding initial conditions) is
\begin{equation}
\mathbb{E}[Y_{j}]=\frac{b_{j}}{f_{j}}\left(\prod _{u=1}^{N}\bar{X}_{u} + \sum _{k=1}^{N} \sum _{m=1}^{N}\sigma_{k}\sigma_{m} \left(\prod_{\substack{u=1 \\ u \neq k,m} }^{N}\bar{X}_{u} \right) \rho_{k,m}  \right)
\end{equation}
which is different from the case when there is no noise, if there are correlations among the noise terms.
The last equation clearly shows that noisy inputs can have an effect in the average of the concentration of the output, even if their mean is zero.
The amount by which they shift the mean depends on their own variances, their correlations, and the product of concentrations of all \textit{other} reactants.

The variance is equal to
\begin{equation}
\mathbb{V}[Y_{j}]=\frac{b_{j}^{2}}{2f_{j}}\left(  \sum_{k=1}^{N} \sigma_{k}^{2}   \prod_{\substack{u=1 \\ u \neq k} }^{N}\bar{X}_{u}^{2}+ \sum_{k<m}2\rho_{km}\sigma_{k}\sigma_{m}\sigma_{k}^{2}   \prod_{\substack{u=1 \\ u \neq k,m} }^{N}\bar{X}_{u}^{2} \right).
\end{equation}

As before, positive correlations increase variance, negative correlations reduce it, and the extent by which the correlations affect it depends on the concentrations of the other species in the reaction.

\subsection{Reactions With Filtered Noise}

Suppose we have the following simple reaction:
\begin{equation}
X_{1}+X_{2}+\dots +X_{n}\rightarrow Y
\end{equation}
where $X_{1}\dots X_{N}$ fluctuate around an average value, but the noise has already passed through a linear filter.
In this case, we can write the equation that $Y$ satisfies as an ordinary differential equation:
\begin{equation}
dY=-aY dt+\prod_{u=1}^{N} \left(  \bar{X}_{k}+\sigma_{k}\int_{0}^{t}h_{k}(t-s)dW^{k}_{s}  \right) dt
\end{equation}
where once again $dW^{k}_{t}$ is the standard Wiener process corresponding to species $k$.
Expanding the last equation, 
\begin{equation}
\begin{aligned}
dY&= \left(  \prod_{u=1}^{N}\bar{X}_{u} -aY \right) dt + \sum_{k=1}^{N} \sigma_{k}\prod_{\substack{u=1 \\ u \neq k} }^{N}\bar{X}_{u}dt \int_{0}^{t}h_{k}(t-s)dW^{k}_{s} \\
&\qquad + \sum_{k=1}^{N}\sum_{m=1}^{N} \sigma_{k} \sigma_{m} \prod_{\substack{u=1 \\ u \neq k,m} }^{N}\bar{X}_{u} dt \int_{0}^{t} \int_{0}^{t}h_{k}(t-x)h_{m}(t-y)dW^{k}_{x}dW^{m}_{y} \\
&\qquad + \OO{dt}.
\end{aligned}
\end{equation}
We have omitted all the terms whose order is larger than $dt$ as $dt\rightarrow 0$, gathering them under the term $\OO{dt}$.
By using It\^o's box rule again, we can replace the products of Wiener processes by their correlation times the infinitesimal time interval $dt$.

\begin{equation}
\begin{aligned}
dY&= \left(  \prod_{u=1}^{N}\bar{X}_{u} -aY \right) dt + \sum_{k=1}^{N} \sigma_{k}\prod_{\substack{u=1 \\ u \neq k} }^{N}\bar{X}_{u}dt \int_{0}^{t}h_{k}(t-s)dW^{k}_{s} \\
&\qquad + \sum_{k=1}^{N}\sum_{m=1}^{N} \sigma_{k} \sigma_{m} \prod_{\substack{u=1 \\ u \neq k,m} }^{N}\bar{X}_{u} dt \int_{0}^{t}\rho_{km} h_{k}(t-x)h_{m}(t-x)dx \\
&\qquad + \OO{dt}.
\end{aligned}
\end{equation}
Note that the second sum of integrals is deterministic and does not depend on any Wiener process.
Setting 
\begin{equation}
\begin{aligned}
f(t) &=\sum_{k=1}^{N}\sum_{m=1}^{N} \sigma_{k} \sigma_{m} \prod_{\substack{u=1 \\ u \neq k,m} }^{N}\bar{X}_{u} dt \int_{0}^{t}\rho_{km} h_{k}(t-x)h_{m}(t-x)dx \\
c_{N}&=\prod_{u=1}^{N}\bar{X}_{u} \quad , \quad
\hat{\sigma}_{k}=\sigma_{k}\prod_{\substack{u=1 \\ u \neq k} }^{N}\bar{X}_{u}\qquad \textrm{and} \quad \\
q_{k}(t)&=\hat{\sigma}_{k}\int_{0}^{t}h_{k}(t-s)dW^{k}_{s} ,
\end{aligned}
\end{equation}
the solution to the last differential equation (with zero initial conditions) is 
\begin{equation}
\begin{aligned}
Y(t)=c_{N}(1-e^{-at})+\int_{0}^{t}e^{-a(t-u)}f(u)du+ \sum _{k=1}^{N}\int_{0}^{t}e^{-a(t-u)}q_{k}(u)du.
\end{aligned}
\end{equation}
More generally, if the differential equation for the output has impulse response $g(t)$, and initial condition $Y_{0}$, 
\begin{equation}
\begin{aligned}
Y(t)=Y_{0}g(t)+c_{N}\int_{0}^{t}g(t-u)du+\int_{0}^{t}g(t-u)f(u)du+ \sum _{k=1}^{N}\int_{0}^{t}g(t-u)q_{k}(u)du
\end{aligned}
\end{equation}
where all terms except for the last sum are deterministic.
The last equation nicely decomposes the factors that drive the output $Y(t)$.
The first term is the effect of the initial condition, the second term denotes the effect of the mean value of the inputs, the third results from the noise correlations of the inputs, and the last term corresponds to the sum of the random fluctuations of all input sources.

If the output of reaction \eqref{SimpleBimolecularReaction} receives inputs that are affected by both filtered and unfiltered disturbances, then we can use the same methods to find the mean and standard deviation of the output.
We will analyze the case where we have two inputs, one of each type, since the generalization to an arbitrary number of inputs is straightforward.
Suppose that the chemical species $Y$ depends on species $X_{1}$ and $X_{2}$
\begin{equation}
X_{1}+X_{2} \rightarrow Y
\label{SimpleBimolecularReaction}
\end{equation}
where the inputs $X_{1}$ and $X_{2}$ are defined by the following differential equations:
\begin{equation}
X_{1}(t)=\bar{X}_{1}+\sigma_{1}\int_{0}^{t}h(t-s)dU_{s}
\end{equation}
and 
\begin{equation}
X_{2}(t)=\bar{X}_{s}+\sigma_{2} dW_{s}
\end{equation}
where $U_{t}$ and $W_{t}$ are standard Wiener processes.

The stochastic differential equation for $Y$ is
\begin{equation}
\begin{aligned}
dY&=(\bar{X}_{1}\bar{X}_{2}-aY)dt +\sigma_{2}\bar{X}_{1}dW_{t} +\sigma_{1}\bar{X}_{2}dt\int _{0}^{t}h(t-s)dU_{s} \\
&\qquad \qquad +\sigma_{1}\sigma_{2}\int_{0}^{t}h(t-s)dU_{s}dW_{t}\\
&=(\bar{X}_{1}\bar{X}_{2}+\rho h_{0}\sigma_{1}\sigma_{2}-aY)dt+\sigma_{2}\bar{X}_{1}dW_{t}+\sigma_{1}\bar{X}_{2}\int _{0}^{t}h(t-s)dU_{s}
\end{aligned}
\end{equation}
since
\begin{equation}
dU_{s}dW_{t}= \left\{ \begin{array}{ll}
\rho dt & \textrm{if $s=t$}\\
0 & \textrm{otherwise}.\\
\end{array} \right.
\end{equation}
The output is equal to
\begin{equation}
\begin{aligned}
Y(t)&=Y_{0}e^{-at}+ (\bar{X}_{1}\bar{X}_{2}+\rho h_{0}\sigma_{1}\sigma_{2})(1-e^{-at})+ \sigma_{2}\bar{X}_{1}\int_{0}^{t}e^{-a(t-s)}dW_{s} \\
& \qquad \qquad + \sigma_{1}\bar{X}_{2}\int_{0}^{t}\int_{0}^{s} e^{a(t-s)} h(s-x) dU_{x}.
\end{aligned}
\end{equation}
The mean is 
\begin{equation}
\mathbb{E}[Y]=\frac{1}{a} \left( \bar{X}_{1}\bar{X}_{2}+ \rho h_{0}\sigma_{1}\sigma_{2} \right)
\end{equation}
which differs from the noiseless case by the last term, which is proportional to the correlation and the standard deviation of the noise inputs.
Similarly, the variance is found to be equal to
\begin{equation}
\mathbb{V}[Y(t)]=V_{1}(t)+V_{2}(t)+V_{12}(t)
\end{equation}
where 
\begin{equation}
V_{1}(t)=e^{-2at} \sigma_{1}^{2}\bar{X}_{2}^{2} \int_{0}^{t}\int_{0}^{t}e^{a(r+s)}
 \left(\int_{0}^{r\wedge s}h(s-u)h(r-u)du \right)dr ds,
 \end{equation}
 \begin{equation}
V_{2}(t)=\frac{\sigma_{2}^{2}\bar{X}_{1}^{2}}{2a}(1-e^{-2at})
\end{equation}
 and
 \begin{equation}
 V_{12}(t)=\rho \sigma_{1}\sigma_{1}\bar{X}_{1}\bar{X}_{2}\int_{0}^{t}e^{-a(t-y)}\left( \int_{0}^{t\wedge y} e^{-a(t-x)}h(y-x)\right)dy.
\end{equation}
The first component $V_{1}(t)$ is the variance because of the noise in the first input $dU_{t}$, $V_{2}(t)$ the variance because of noise in the second input, and the last term $V_{12}(t)$ is the variance emanating from their correlation.

When the inputs $X_{1}$ and $X_{2}$ in \eqref{SimpleBimolecularReaction} both have a filtered multiplicative noise component, then the differential equation becomes 
\begin{equation}
dY=-aYdt+\left(\bar{X}_{1} \lambda_{1} \int_{0}^{t}e^{-(\lambda_{1}+\frac{\sigma_{1}^{2}}{2})x}e^{\sigma_{1}U_{x}}dx \right)\left( \bar{X}_{2} \lambda_{2}\int_{0}^{t}e^{-(\lambda_{2}+\frac{\sigma_{1}^{2}}{2})y}e^{\sigma_{1}W_{y}}dy \right)dt.
\end{equation}
In order to account for the possibly nonzero correlation between processes $U_{t}$ and $W_{t}$, we write each of them as a sum of two uncorrelated standard processes:
\begin{equation}
\begin{aligned}
U_{t}&=aA_{t}+\sqrt{1-a^{2}}B_{t} \\
W_{t}&=bA_{t}+\sqrt{1-b^{2}}C_{t}.
\end{aligned}
\end{equation}
The processes $A_{t}, B_{t}$ and $C_{t}$ have correlation zero, and $\rho=ab$ is the correlation between $U_{t}$ and $W_{t}$ :
\begin{equation}
-1\leq a \leq 1 \quad , \quad  -1\leq b \leq 1 \quad \textrm{and} \quad -1\leq \rho \leq 1.
\end{equation}

We are interested in finding the mean and variance of $Y$.
First, we will compute the expected value of the product of the two exponential Wiener processes $U_{t}$ and $W_{t}$.
\begin{equation}
\begin{aligned}
\mathbb{E}\left[e^{\sigma_{1}U_{x}}  e^{\sigma_{2}W_{y}} \right] &= \mathbb{E}\left[e^{\sigma_{1}\left(aA_{x}+\sqrt{1-a^{2}}B_{x}\right)}  e^{\sigma_{2} \left(bA_{y}+\sqrt{1-b^{2}}C_{y}\right)} \right]  \\
&= \mathbb{E}\left[e^{a\sigma_{1} A_{x}+b\sigma_{2}A_{y}}\right] \cdot \mathbb{E}\left[e^{\sigma_{1}\sqrt{1-a^{2}}B_{x}}\right] \cdot \mathbb{E}\left[e^{\sigma_{2}\sqrt{1-b^{2}}C_{y}}\right] \\
&=e^{\frac{1}{2} \left(a\sigma_{1}+b\sigma_{2}\right)^{2} x\wedge y}
e^{\frac{1}{2} \left(\left(a\sigma_{1}\right)^{2}\delta_{x}+\left(b\sigma_{2}\right)^{2}\delta_{y} \right)(x\vee y-x\wedge y)} 
e^{\frac{1}{2} \sigma_{1}^{2}\left(1-a^{2} \right)x} e^{\frac{1}{2} \sigma_{2}^{2}\left(1-b^{2} \right)y} \\
\end{aligned}
\end{equation}
where $\delta$ denotes the Kronecker delta with $\delta_{x}=\delta(x\geq y)$ and $\delta_{y}=\delta(y\geq x)$.

The expected value of the input of the differential equation is
\begin{equation}
\begin{aligned}
\mathbb{E}[u(t)] &=\mathbb{E}\left[\left(\bar{X}_{1} \int_{0}^{t}e^{-(\lambda_{1}+\frac{\sigma_{1}^{2}}{2})x}e^{\sigma_{1}U_{x}}dx \right)\left( \bar{X}_{2}\int_{0}^{t}e^{-(\lambda_{2}+\frac{\sigma_{2}^{2}}{2})y}e^{\sigma_{2}W_{y}}dy \right)  \right] \\
&=\lambda_{1}\lambda_{2}\bar{X}_{1}\bar{X}_{2}  \int_{0}^{t} e^{-(\lambda_{1}+\frac{\sigma_{1}^{2}}{2})x}  \left( \int_{0}^{t} e^{-(\lambda_{2}+\frac{\sigma_{2}^{2}}{2})y} \mathbb{E}\left[e^{\sigma_{1}U_{x}} e^{\sigma_{2}W_{y}}\right] dy\right) dx  \\
&=\lambda_{1} \lambda_{2}\bar{X}_{1}\bar{X}_{2}  \int_{0}^{t} e^{-\lambda_{1}x}  \left( \int_{0}^{x} e^{-(\lambda_{2}-\rho \sigma_{1}\sigma_{2})y} dy\right) dx  \\
&\qquad +\lambda_{1} \lambda_{2} \bar{X}_{1}\bar{X}_{2}  \int_{0}^{t} e^{-(\lambda_{1}+\rho \sigma_{1}\sigma_{2})x}  \left( \int_{x}^{t} e^{-\lambda_{2}y}dy\right) dx  \\
&=\lambda_{2}\bar{X}_{1}\bar{X}_{2} \frac{e^{-t \lambda _1} \left(\left(1-e^{\rho t \sigma_{2}  \sigma_{1} -t \lambda _2}\right) \lambda _1+\left(-1+e^{t \lambda _1}\right) \left(\rho \sigma_{2}  \sigma_{1} -\lambda _2\right)\right)}{\left(\rho \sigma_{2}  \sigma_{1} -\lambda _2\right) \left(-\rho \sigma_{2}  \sigma_{1} +\lambda _1+\lambda _2\right)} \\
& \qquad + \lambda_{1}\bar{X}_{1}\bar{X}_{2} \frac{e^{-t \lambda _2} \left(\left(1-e^{t \lambda _2}\right) \rho \sigma_{2}  \sigma_{1} -\left(1-e^{t \lambda _2}\right) \lambda _1-\left(1-e^{\rho t \sigma_{2}  \sigma_{1} -t \lambda _1}\right) \lambda _2\right)}{\left(\rho \sigma_{2}  \sigma_{1} -\lambda _1\right) \left(\rho \sigma_{2}  \sigma_{1} -\lambda _1-\lambda _2\right)}
\end{aligned}
\end{equation}
where we assume that
\begin{equation}
\lambda_{1}>\frac{\sigma_{1}^{2}}{2} \quad, \quad \lambda_{2}>\frac{\sigma_{2}^{2}}{2}  \implies \lambda_{1}+\lambda_{2}>\rho \sigma_{1} \sigma_{2}.
\end{equation}
The inequalities above guarantee that the inputs have finite variances, as shown in equation \eqref{GeometricNoiseSecondMoment}. 
In the equilibrium state, 
\begin{equation}
\lim_{t\rightarrow \infty} \mathbb{E}\left[ u(t) \right] = \bar{X}_{1}\bar{X}_{2} \frac{\lambda _1+\lambda _2}{\left(\lambda _1+\lambda _2 -\rho \sigma_{1}  \sigma_{2} \right)}.
\end{equation}
The output average is then equal to
\begin{equation}
\mathbb{E}[Y]=\lim_{t \rightarrow \infty} \mathbb{E}\left[ Y(t) \right] = 
\frac{\bar{X}_{1}\bar{X}_{2}}{a}\cdot \frac{\lambda _1+\lambda _2}{\left(\lambda _1+\lambda _2 -\rho \sigma_{1}  \sigma_{2} \right)}.
\end{equation}

The last equation clearly shows that if the input noise sources are correlated ($\rho \neq 0$), the average value of the output will be different from the value when there is no correlation ($\rho=0$). 
As shown in the other types of noise, positive correlations increase the mean, and negative correlations reduce it.

The variance can be computed using the same methods.
First, we will calculate the expected value of a product of different instances of a standard Wiener process. 

\begin{lemma}
If $t_{1},t_{1},\dots t_{n} \in \mathbb{R}^{+}$ is an ordered set of times such that $t_{1}\leq t_{2} \leq \ldots \leq  t_{n}$ and $\sigma_{1},\sigma_{2}\dots \sigma_{n} \in \mathbb{R}^{+}$ are arbitrary positive numbers denoting standard deviations, then
\begin{equation}
\mathbb{E}\left[\prod _{k=1}^{n}e^{\sigma_{k}W_{t_{k}}}\right] =\exp \left[ \frac{1}{2} \sum_{k=1}^{n}  \left( \sum_{m=k}^{n} \sigma_{m}\right)^{2} (t_{k}-t_{k-1}) \right]
\end{equation}
where $W_{t}$ is the standard Wiener process.
\end{lemma}
\begin{proof}
For each $t_{k}$, we decompose the Wiener process $W_{t_{k}}$ as a sum of independent processes:
\begin{equation}
W_{t_{k}}=\sum _{m=1}^{k} \left(  W_{t_{m}}-W_{t_{m-1}}  \right).
\end{equation}
Based on the sum above, we can write
\begin{equation}
\begin{aligned}
\prod_{k=1}^{n}e^{\sigma_{k}W_{t_{k}}}&=\exp \left[ \sum_{k=1}^{n}\sigma_{k}W_{t_{k}} \right] \\
&=\exp \left[ \sum_{k=1}^{n}\sigma_{k} \sum _{m=1}^{k} \left(  W_{t_{m}}-W_{t_{m-1}}  \right)\right] \\
&=\exp \left[ \sum_{k=1}^{n} \left(  W_{t_{k}}-W_{t_{k-1}}  \right)\sum _{m=k}^{n} \sigma_{k} \right] \\
&= \prod_{k=1}^{n} \exp\left[\left(  W_{t_{k}}-W_{t_{k-1}}  \right) \displaystyle \sum _{m=k}^{n} \sigma_{k}\right] \\
\end{aligned}
\end{equation}
where in the last equation, we changed the order of summation making use of the triangle rule.
All terms in the last product are independent:
\begin{equation}
\begin{aligned}
\mathbb{E} \left[ \prod_{k=1}^{n}e^{\sigma_{k}W_{t_{k}}} \right]&=\mathbb{E}\left[ \prod_{k=1}^{n} \exp\left[\left(  W_{t_{k}}-W_{t_{k-1}}  \right) \displaystyle \sum _{m=k}^{n} \sigma_{k}\right] \right] \\
&= \prod_{k=1}^{n} \mathbb{E}\left[\exp\left[\left(  W_{t_{k}}-W_{t_{k-1}}  \right) \displaystyle \sum _{m=k}^{n} \sigma_{k}\right] \right] \\
&= \prod_{k=1}^{n} \exp\left[\frac{1}{2}\left(t_{k}-t_{k-1}\right) \displaystyle \left(\sum _{m=k}^{n} \sigma_{k}\right)^{2} \right] \\
&= \exp\left[\frac{1}{2} \sum_{k=1}^{n} \displaystyle \left(\sum _{m=k}^{n} \sigma_{k}\right)^{2}\left(t_{k}-t_{k-1}\right)  \right] .\\
\end{aligned}
\end{equation}
\end{proof}

When one of the inputs is affected by multiplicative noise, and the other by additive noise, the mean value of the output is not affected, even if the driving noise is the same in both cases.
If consider again the chemical reaction \eqref{SimpleBimolecularReaction}, the differential equation in that case becomes
\begin{equation}
\frac{dY}{dt}=-aY+\left(\lambda_{1}\bar{X}_{1} \int_{0}^{t}e^{-(\lambda_{1}+\frac{\sigma^{2}}{2})x}e^{\sigma W_{x}}dx \right)  \left(  \bar{X}_{2}+\sigma \int_{0}^{t} e^{-a(t-y)}dW_{y}  \right).
\end{equation}
The input is equal to
\begin{equation}
u(t)=\left(\lambda_{1}\bar{X}_{1} \int_{0}^{t}e^{-(\lambda_{1}+\frac{\sigma^{2}}{2})x}e^{\sigma W_{x}}dx \right)  \left(  \bar{X}_{2}+\sigma \int_{0}^{t} e^{-\lambda_{2}(t-y)}dW_{y}  \right)
\end{equation}
and its expected value is
\begin{equation}
\begin{aligned}
\mathbb{E}[u(t)] &= \lambda_{1}\bar{X}_{1}\bar{X}_{2} \int_{0}^{t}e^{-(\lambda_{1}+\frac{\sigma^{2}}{2})x}\mathbb{E} \left[ e^{\sigma W_{x}}\right]dx  \\
&\qquad +\sigma \lambda_{1} \bar{X}_{1} \int_{0}^{t}  \int_{0}^{t} e^{-(\lambda_{1}+\frac{\sigma^{2}}{2})x}  e^{-\lambda_{2}(t-y)} \mathbb{E} \left[ e^{\sigma W_{x}} dW_{y} \right]  dx.
\end{aligned}
\label{ProductAdditiveGeometricExpectedValueOfInput}
\end{equation}
In order to compute the second term of the last equation, we will need the following Lemma about the expected value of the product an exponential Wiener process with an infinitesimal difference of the same process.
\begin{lemma}
If $W_{t}$ is a standard Wiener process, then 
\begin{equation}
\mathbb{E} \left[ e^{\sigma W_{s}} dW_{t} \right] = \left\{ \begin{array}{ll}
0 & \textrm{if $s\leq t$}\\
\sigma^{2} e^{\frac{\sigma^2}{2}s} dt & \textrm{if $s>t$}.
\end{array} \right.
\end{equation}
\end{lemma}
\begin{proof}
If $s<t$, then $W_{s}$ and $dW_{t}=W_{t+dt}-W_{t}$ are uncorrelated, so
\begin{equation}
\mathbb{E}\left[e^{\sigma W_{s}} dW_{t}\right]=\mathbb{E}\left[e^{\sigma W_{s}}\right]\mathbb{E}\left[dW_{t}\right]=0.
\end{equation}
Now, if $0<a<b<s$, then 
\begin{equation}
\begin{aligned}
\mathbb{E} \left[e^{\sigma W_{s}} (W_{b}-W_{a}) \right] &=\mathbb{E} \left[e^{\sigma W_{a}}  \right] \mathbb{E} \left[e^{\sigma (W_{b}-W_{a})}(W_{b}-W_{a})  \right] \mathbb{E} \left[e^{\sigma(W_{s}-W_{b})}  \right] \\
&=e^{\frac{1}{2}\sigma^{2}a} e^{\frac{1}{2} \sigma^{2}(b-a)} \sigma^{2}(b-a)
e^{\frac{1}{2}\sigma^{2} (s-b)}  \\
&= \sigma^{2} e^{\frac{1}{2}\sigma^{2}s}(b-a). \\
\end{aligned}
\end{equation}
Setting $a=t$ and $b=t+dt$, we get the desired result.
\end{proof}
Recalling equation \eqref{ProductAdditiveGeometricExpectedValueOfInput}, 
\begin{equation}
\begin{aligned}
\mathbb{E}[u(t)] &= \lambda_{1}\bar{X}_{1}\bar{X}_{2} \int_{0}^{t}e^{-\lambda_{1}x}dx+\sigma^{3} \lambda_{1} \bar{X}_{1}  e^{-\lambda_{2}t}  \int_{0}^{t}  \left(\int_{y}^{t} e^{-\lambda_{1} x}  e^{\lambda_{2}y} dx \right) ds \\
&=\bar{X}_{1}\bar{X}_{2} \left(1-e^{-\lambda_{1}t}  \right)+  \sigma^{3}\bar{X}_{1} \frac{e^{-t \left(\lambda _1+\lambda _2\right)}  \left(\lambda _1 \left(1-e^{t \lambda _2}\right) -\left(1-e^{t \lambda _1}\right) \lambda _2\right)}{\lambda _2\left(\lambda _1-\lambda _2\right)}.
\end{aligned}
\end{equation}
As time $t$ grows large, 
\begin{equation}
\lim _{t\rightarrow \infty} \mathbb{E}\left[ u(t) \right]=\bar{X}_{1}\bar{X}_{2}
\end{equation}
and the mean of the output is
\begin{equation}
\mathbb{E}[Y]=\frac{1}{a}\bar{X}_{1}\bar{X}_{2}
\end{equation}
which is exactly the same as in the case where the two noise inputs are completely uncorrelated.
So, input noise correlation does not affect the average concentration of the output in this case.

This section has analyzed how noise propagates in an arbitrary chemical reaction network where one or more inputs include a random component.
The different noise sources may have arbitrary correlations with each other.
We have studied the propagation of both additive and multiplicative noise.
One of the main results is that even if all noise sources have mean equal to zero, their correlations shift the mean of the outputs, for both types of noise.
If there is positive correlation, the mean of the output increases, and when the correlation is negative, it shifts lower, and the same is true for  the output variance.

\section{Conclusions}

We have shown how noise propagates in networks and how a network's noisy parameters can affect its output.
Since many biological networks are locally tree-like, we have studied how noise propagates in the absence of feedforward or feedback cycles.
Tree networks are relatively easy to quantitatively analyze, since there is only one path from each node to another. 
We have derived a method to compute the variance of the output of any tree network, and shown that the variance is minimized when there are no ``bottlenecks'' in each pathway, in other words when there is no rate limiting step.
When a network is not a tree, there are cycles, which means that a signal (along with its noise) can propagate through two or more paths towards the output. 
Feedback cycles typically reduce the output variance, and feedforward cycles increase it.
When the noise sources are correlated, the variance in the output is larger, and small cycles have a stronger influence on the output, compared to longer cycles in both cases.
Delays contribute to the decrease of the output noise when we have two or more noise sources, since their correlation is diminished.
Crosstalk is also shown to decrease the output variance, but the tradeoff is that the output mean is lowered, or the concentration of the inputs needs to be proportionally higher in order to ensure the same output.
In biological and chemical reaction networks, the reaction rates are prone to noise, since they depend on the concentration of other species.
When the degradation rates are affected by noise, the result is increased output variance, which also depends on the concentration of the respective species, and the form of the output is different from when the noise is in the inputs, in the sense that higher concentrations also correspond to larger deviations from the mean.
Finally, we have extensively studied how noise propagates through chemical reaction networks where one or more of the reactants are noisy, and their disturbances may be correlated. 
Even when the disturbances have zero average, correlations change the output mean, and variance.


\begin{thebibliography}{20}

\bibitem{Paulsson2004} Paulsson, J {\it Summing Up the Noise in Gene Networks}, Nature {\bf 427} 415--418 (2004)

\bibitem{Raj2008} Raj, A and van Oudenaarden, A {\it Nature, Nurture, or Chance: Stochastic Gene Expression and its Consequences}, Cell {\bf 135} 216--226 (2008)

\bibitem{FunctionalNoise} Eldar, A and Elowitz, MB {\it Functional Roles for Noise in Genetic Circuits}, Nature {\bf 467} 167--173 (2010)

\bibitem{MinXtalkNetworks} Barmpoutis, D and Murray, RM {\it Quantification and minimization of crosstalk sensitivity in networks}, \href{http://arxiv.org/abs/1012.0606v1}{arXiv:1012.0606v1}  (2010)

\bibitem{ControlBook} {\AA}str{\"o}m, KJ and Murray, RM  {\it Feedback Systems: An Introduction for Scientists and Engineers} Princeton University Press (2008).

\bibitem{ProbabilityBook} Liptser, RS and Shiryaev, AN {\it Statistics of Random Processes I: General Theory} Springer (2000)

\bibitem{GraphTheoryBook} Newman, MEJ  {\it  Networks: An introduction}, Oxford University Press (2010).

\bibitem{BiostatisticsBook} Rosner, B {\it Fundamentals of Biostatistics} Duxbury Press (2005)

\bibitem{ScaleFreeMetabolicNetworks} Jeong, H {\it The large-scale organization of metabolic networks}, Nature {\bf 407} 651--654 (2000)

\bibitem{ConsensusPaper} Olfati-Saber, R and Murray, RM {\it Consensus Problems in Networks of Agents with Switching Topology and Time-Delays}, IEEE Transactions on Automatic Control {\bf 49} 1520--1533  (2004)

\bibitem{DegreeRealizability} Hakimi, SL {\it On realizability  of a set of integers as degrees of the vertices of a linear graph}, Applied Mathematics  {\bf 10} 496--506  (1962)

\bibitem{ExtremalNetworks} Barmpoutis, D and Murray, RM {\it Extremal Properties of Complex Networks}, \href{http://arxiv.org/abs/1104.5532}{arXiv:1104.5532v1} (2011)

\bibitem{MaxClusteringNetworks} Barmpoutis, D and Murray, RM {\it Networks with the smallest average distance and the largest average clustering}, \href{http://arxiv.org/abs/1007.4031v1}{arXiv:1007.4031v1} (2010)

\bibitem{SystemsBiologyBook} Alon, U {\it  An introduction to systems biology: design principles of biological circuits}, Chapman and Hall/CRC (2006)

\bibitem{StochasticsHandbook} Gardiner, CW {\it Handbook of Stochastic Methods for Physics, Chemistry, and the Natural Sciences}, Springer
(1983).

\bibitem{LestasPaper} Lestas, I  et al. {\it Fundamental limits on the suppression of molecular fluctuations}, Nature {\bf 467} 174--178 (2010)

\bibitem{NoiseInLinearCRNs} Warren, PB et al. {\it Exact Results for Noise Power Spectra in Linear Biochemical Reaction Networks}, The Journal of Chemical Physics {\bf 125} 144904 (2006)

\bibitem{Samoilov2002} Samoilov, M et al. {\it Signal Processing by Simple Chemical Systems}, Journal of Physical Chemistry {\bf 106} 10205-10221 (2002)

\end{thebibliography}
\end{document}